\documentclass[12pt]{article}
\usepackage[T1]{fontenc}
\usepackage[dvips]{graphicx}
\graphicspath{{images/}}
\usepackage{verbatim}
\usepackage{amsmath,amsthm}
\usepackage{mathtools}
\usepackage{xspace}
\usepackage{pifont}
\usepackage{graphicx}
\usepackage{amssymb}
\usepackage{epic, eepic}
\usepackage{dsfont}
\usepackage{amssymb}
\usepackage{makeidx}
\usepackage{mathrsfs}
\usepackage{exscale}
\usepackage{color} 
\usepackage{overpic} 
\usepackage{bm}
\usepackage{bbm}
\usepackage{booktabs} 
\usepackage{color, colortbl}
\usepackage{subcaption}
\usepackage[numbers]{natbib}

\RequirePackage[colorlinks,citecolor=blue,urlcolor=blue]{hyperref}

\definecolor{Gray}{gray}{0.9}

\usepackage{amsmath,afterpage}
\usepackage{epsf}
\usepackage{graphics,color}
\RequirePackage[colorlinks,citecolor=blue,urlcolor=blue]{hyperref}

\usepackage{float} 

\usepackage{hyperref}
\hypersetup{%
  colorlinks = true,
  linkcolor  = black
}

\def\a{\mathbf{a}}

\def\0{\mathbf{0}}

\def\eps{\varepsilon}

\def\lam{\lambda}

\def \< {\langle}
\def \> {\rangle}

\def\ol{\overline}

\def\to{\rightarrow}
 \def\ol{\overline}    
\def\Om{\Omega}  \def\om{\omega} 

  \def\a{\alpha} 
 
\def\eps{\varepsilon}

 \def\t{\times}  \def\lam{\lambda}

\def\E{{\bf E}}

\def\beqa{\begin{eqnarray}}
\def\eeqa{\end{eqnarray}}
\def\beqas{\begin{eqnarray*}}
\def\eeqas{\end{eqnarray*}}

\newtheorem{theorem}{Theorem}[section]
\newtheorem{lemma}[theorem]{Lemma}

\newtheorem{proposition}[theorem]{Proposition}

\newtheorem{corollary}[theorem]{Corollary}

\newtheorem{remark}[theorem]{Remark}

\newtheorem{definition}[theorem]{Definition}

\newtheorem{assumption}[theorem]{Assumption} 
\numberwithin{equation}{section}
\newcommand{\hatd}[1]{{}}

\newcommand{\bd}{\begin{displaymath}}
\newcommand{\ed}{\end{displaymath}}
\newcommand{\be}{\begin{equation}}
\newcommand{\ee}{\end{equation}}
\newcommand{\bq}{\begin{eqnarray}}
\newcommand{\eq}{\end{eqnarray}}
\newcommand{\bn}{\begin{eqnarray*}}
\newcommand{\en}{\end{eqnarray*}}
\newcommand{\dl}{\delta}

\newcommand{\cT}{\mathcal{T}}

\def\wt{\widetilde}

\def\P{\mathbb{P}}
\def\E{{\mathbb{E}}}
\newcommand{\R}{{\mathbb R}}
\newcommand{\id}{{\rm id}}

\newcommand{\Mid}{{\ \Big|\ }}

\usepackage[normalem]{ulem}

\usepackage[shortlabels]{enumitem}
\usepackage[linesnumbered, ruled,vlined]{algorithm2e}
\SetKwInput{KwInput}{Input}                
\SetKwComment{Comment}{/* }{ */}

\def\cA{\mathcal{A}}

\def\cE{\mathcal{E}}
\def\cF{\mathcal{F}}
\def\cG{\mathcal{G}}

\def\cI{\mathcal{I}}

\def\cN{\mathcal{N}}
\def\cO{\mathcal{O}}

\def\cT{\mathcal{T}}

\def\cV{\mathcal{V}}

\def\sB{\mathbb{B}}
\def\sC{\mathbb{C}}

\def\sE{{\mathbb{E}}}

\def\sN{{\mathbb{N}}}
\def\sP{\mathbb{P}}

\def\sR{{\mathbb R}}

\DeclareMathOperator*{\argmin}{arg\,min}

\newcommand{\todo}[1]{{\color{red} \fbox{#1}}}

\usepackage{authblk}

\title{Statistical Learning with Sublinear Regret of Propagator Models}
 
\author{Eyal Neuman}
\author{Yufei Zhang}
\affil{Department of Mathematics, Imperial College London}

\begin{document}

 \vspace{-0.5cm}
\maketitle

\begin{abstract}
We consider a class of learning problems in which an agent liquidates a risky asset while creating both transient price impact driven by an unknown  convolution propagator and linear temporary price impact with
 an unknown parameter. We characterize the trader's performance as 
maximization of a revenue-risk functional, where the trader also exploits available information on a price predicting signal. We present a trading algorithm that alternates between exploration and exploitation phases and achieves sublinear regrets with high probability. For the exploration phase we propose a novel approach for 
non-parametric estimation of the price impact kernel by observing only the visible price process and derive sharp bounds on the convergence rate, which are characterised by the singularity of the propagator. These kernel estimation methods extend existing methods from the area of Tikhonov regularisation for inverse problems and are of independent interest. The bound on the regret in the exploitation phase is obtained by deriving stability results for the optimizer and value function of the associated class of infinite-dimensional stochastic control problems. As a complementary result we propose a regression-based algorithm to estimate the conditional expectation of non-Markovian signals and derive its convergence rate.   
\end{abstract} 

\begin{description}
\item[Mathematics Subject Classification (2010):] 
62L05,
60H30, 91G80, 68Q32, 93C73, 93E35, 62G08
\item[JEL Classification:] C02, C61, G11
\item[Keywords:] optimal portfolio liquidation, price impact, propagator models, predictive signals, Volterra stochastic control, 
non-parametric estimation, 
reinforcement learning, 
regret analysis
\end{description}

\bigskip

 \tableofcontents

\newpage 
\section{Introduction}

Price impact refers to the empirical fact that execution of a large order
affects the risky asset's price in an adverse and persistent manner  leading to less favourable prices. Propagator models are a central tool in describing this phenomena mathematically. 
This class of models provides deep insight into the nature of price impact and price dynamics.  It expresses price moves in terms of the influence of past trades, which gives reliable reduced form view on the limit order book. It also provides interesting insights on liquidity, price formation and on the interaction between different market participants through price impact. The model's tractability provides a convenient formulation for stochastic control problems arising from optimal execution \citep{bouchaud_bonart_donier_gould_2018, gatheral2010no}. 
More precisely, if the trader's holdings in a risky asset is denoted by $Q=\{Q_t\}_{t\geq 0}$, then the asset price $S_t$ is given by
\be \label{ex-framwork} 
S_t = S_0+ \lambda\dot Q_t + \int_0^tG(t-s)dQ_s + P_t,
\ee
where $P$ is a semimartingale, 
$\lambda$ is the positive temporary price impact coefficient and the price impact kernel $G$ is also known as the \emph{propagator}. We refer to $\dot Q$ as the execution trading speed. Since $G(t)$ typically decays for large values of $t$, the convolution on the right-hand-side of \eqref{ex-framwork} is referred to as transient price impact (see e.g. \citet[Chapter 13]{bouchaud_bonart_donier_gould_2018}). A well-known example à la Almgren and Chriss, introduces to the case where $G$ is a constant, then the above convolution represents permanent price impact (see \citep{AlmgrenChriss1,OPTEXECAC00}). 

In the aforementioned setting, the trader can only observe the visible price process $S$ and   her own inventory $Q$.
In order to quantify the price impact and hence the trading costs, the trader needs a good estimation of $G$ and $\lambda$.
Some estimators for discrete-time versions of the model were proposed in \citep{bouchaud-gefen,Forde:2022aa, Toth_17,Benzaquen_22} and in Chapter 13.2 of \cite{bouchaud_bonart_donier_gould_2018}, where 
only a finite amount of values $\{G(t_n)\}_{n =1}^N$ are estimated for a predetermined grid $0\leq t_1<...<t_n$. However, even in this finite dimensional projection of the problem, the convergence of the estimators remains unproved, hence rigorous results on the estimation of $G$ are considered as a long-standing open problem. In one of the main results of this paper we propose a novel approach for non-parametric estimation of the price impact kernel by observing only the visible price process and we derive sharp bounds on the convergence rate of our estimators, which are characterised by the singularity of the kernel. 

Precise quantification of price impact is a crucial ingredient in portfolio liquidation problems. Considering the adverse effect of the price impact on the execution price, a trader who wishes to minimize her
trading costs has to split her order into a
sequence of smaller orders which are executed over a finite time
horizon. At the same time, the trader also has an incentive to execute
these split orders rapidly because she does not want to carry the risk
of an adverse price move far away from her initial decision
price. This trade-off between price impact and market risk is usually
translated into a stochastic optimal control problem where the trader
aims to minimize a risk-cost functional over a suitable class of
execution strategies, see~\citep{cartea15book, GatheralSchiedSurvey, GokayRochSoner, Gueant:16, citeulike:12047995,N-Sch16} among others. In addition, many traders and trading algorithms also strive for using short term price predictors in their dynamic order execution schedules, which are often related to order book dynamics as discussed in~\cite{leh16moun,Lehalle-Neum18, citeulike:12820703, cont14}.  From the modelling point of view, incorporating signals into execution problems translates into taking a non-martingale price process $P$, in contrast to a martingale price in the classical setting (see \cite{Car-Jiam-2016, BMO:19, NeumanVoss:20, N-V-2021, AJ-N-2022}). This changes the problem significantly as the resulting optimal strategies are often random and in particular signal-adaptive, in contrast to deterministic strategies, which are typically obtained in the martingale price case \cite{stat-adp18}.   

The main goal of this paper is to estimate the price impact kernel while trading a risky asset in a cost-effective manner. In order to do that we propose a learning algorithm that alternates between exploration and exploitation phases. In the exploration phase we proposes a novel approach for non-parametric estimation of the price impact kernel by observing only the visible price process $S$ 
and the signal process, which is the non-martingale component of $P$   in \eqref{ex-framwork} {(see a note about observables in Section \ref{sec:episodic_learning_setting})}. Our estimation method extends the existing theory of Tikhonov regularisation for inverse problems and is of independent interest (see Remark \ref{rem-opt}). Specifically, we propose a regularised least-squares estimator for a squared integrable 
kernel $G$, where samples of the visible price process $S$ are generated by a deterministic trading strategy executed by the trader.  We derive \emph{sharp bounds} on the convergence rate of the estimator with arbitrary high probability under two different assumptions.  For a regular kernel, which has a squared integrable weak derivative, we prove that the convergence rate is of order $N^{-1/6}$. For a singular kernel with a decay rate $G(t) \sim t^{-\alpha}$ for some $\alpha \in(0,1/2)$ we find that the convergence rate is of order $N^{-\frac{1-2\alpha }{2(3-2\alpha)}}$ (see Theorem \ref{thm:estimate_error_high_probability}). Here $N$ is the  sample size   
for the least-squares estimation.
Moreover, from Proposition \ref{prop-opt} 
it follows that the convergence rates given in Theorem \ref{thm:estimate_error_high_probability} are optimal.
More precisely, 
{the outlined upper bounds for the rate of convergence match the lower bound rates of convergence under the assumption of regular kernel and of a power law kernel $G^\star(t)= t^{-\alpha}$,
with $\alpha\in (0,1/2)$. See Remark \ref{rmk:rate_optimal} for specific details.}
Theorem \ref{thm:estimate_error_high_probability} is the first result that proves convergence of any estimator for a propagator, which is based on market price quotes. These results answer an open question that arises from \citep{bouchaud-gefen, bouchaud_bonart_donier_gould_2018,Forde:2022aa, Toth_17,Benzaquen_22} among others. 

After each exploration episode,  we perform a number of exploitation episodes in which we use the current estimation of the kernel and temporary price impact coefficient $\theta^n=(\lam^n, G^n)$ in a framework of optimal execution (see \eqref{ex-framwork}). We execute the optimal trading strategy subject to the estimated parameter $\theta^n$  and derive the performance gap between the revenues of the aforementioned strategy and the optimal strategy with the real choice of $\theta^{\star}$. The bound on the performance gap, which is presented in Theorem \ref{thm:gap}, is derived by proving a stability result (see Proposition \ref{prop:u_lipschitz}) for the optimizer of the associated infinite-dimensional stochastic control problem proposed in \cite{AJ-N-2022}. 

Finally, by combining the exploration and exploration schemes we propose Algorithm  \ref{alg:phased} which achieves sublinear regrets 
 in high probability. In Theorem \ref{thm:regret_general_rate} and Corollary \ref{cor-reg-bnd} we derive a bound on the regret after $N$ trading episodes, which is of order $N^{3/4}$ for a regular kernel and of order $N^{\frac{ 3-2\alpha}{4-4\alpha }}$ for a singular kernel. 
 These  sublinear regret bounds
   underperform the square-root (or  logarithmic) regret for  reinforcement   learning problems with finite-dimensional parametric models  (see e.g., \cite{basei2022logarithmic, guo2023reinforcement,szpruch2021exploration, gao2022logarithmic,gao2022square, szpruch2024optimal}), 
   due to the present  infinite-dimensional non-parametric  kernel estimation 
  (see Remark \ref{rem-conv-compar}).
 
In order to complete our argument, we provide a regression-based algorithm for signal estimation which is performed off-line, that is, independently from the trading algorithm. Specifically we decompose the semimatingale price process $P$ in \eqref{ex-framwork} to a martingale and to a finite variation process $A$ which has the interpretation of a trading signal (see e.g. \cite{N-V-2021}). 
We observe that the optimal trading strategy which is used in the algorithm (see \eqref{eq:optimalcontrol_monotone_str}) involves the  conditional  process 
$(t,s)\mapsto  \sE\left[   A_s\mid \cF_t\right]$. As  the conditional distribution of $A$ is in general not observable, we propose a regression-based algorithm to estimate it,
based on observed signal trajectories. 
Since the agent's trading strategy does not affect the signal, 
the signal estimation can be carried out separately from the 
learning algorithm for $(\lambda^\star,G^\star)$. The convergence rate of  this algorithm is derived in Theorem \ref{theorem-signal}.

Our main results which were outlined above significantly extend the work on reinforcement learning for continuous-time parametric models  which were studied by   \cite{basei2022logarithmic, guo2023reinforcement,szpruch2021exploration, gao2022logarithmic,gao2022square, szpruch2024optimal}  among others. We outline our main contributions that correspond to each part of the learning algorithm. 
\medskip \\
\textbf{Non-parametric kernel estimation:} 
{The main component of the exploration phase is to estimate the kernel function 
$G$ in the \emph{non-Markovian} model 
  \eqref{ex-framwork} in non-parametric manner. This 
results in  a   learning problem with  infinite dimensional parameter, input, and output spaces, 
which 
stands in  contrast to existing theoretical works on 
discrete  Markov decision processes  
(see e.g.~\cite{osband2014model, dean2020sample, 
 he2021logarithmic})
 and on continuous time  
parametric Markov processes
\cite{basei2022logarithmic, guo2023reinforcement,szpruch2021exploration, gao2022logarithmic,gao2022square, szpruch2024optimal}. 
The  sample  complexity bounds therein  depend explicitly on the dimensions of the parameter space, the input space, and the output space, and hence cannot be applied in the present infinite dimensional setting.

It is also challenging to apply existing functional linear regression (FLR) frameworks to our model \eqref{ex-framwork}.
Standard FLR frameworks directly estimate the mapping from the input $ Q$ to the response 
$S$  as an unknown regression operator, instead of estimating 
$\lambda$ and $G$ individually. Moreover, most existing FLR works   characterise  the convergence rate of the proposed estimators under the so-called source conditions, which   assume that the unknown regression operator lies in the range of a suitable fractional power of the input covariance operator (see, e.g., \cite[Assumption 4]{benatia2017functional} and also \cite{benatia2017functional, Rastogi2020convergence}).
It is well known that identifying explicit conditions on 
$G$ and $\dot Q$ such that these source conditions are met is challenging, as it typically requires computing the spectral decomposition of the input covariance and the unknown regression operator, which cannot be performed analytically for general $G$ and $\dot Q$.

In this work, 
we   propose a novel method for the convergence rate analysis of the estimator of $G$ which applies even for singular kernels. 
 For a suitable deterministic trading speed $\dot Q$, the proposed method estimates 
$\lambda$   using a classical Monte Carlo estimator and estimates $G$ using a regularised least-squares estimator involving the estimated $\lambda$. 
We further quantify the convergence rate of the estimator for $G$ by using \emph{appropriate source conditions}, instead of the standard source conditions found in \cite{benatia2017functional, blanchard2018optimal, Rastogi2020convergence}. These appropriate source conditions quantify the degree to which the true kernel $G$ violates the assumption of being in the range of the input operator (see \eqref{eq:distance_function}) and were introduced in \cite{hofmann2006approximate} to study deterministic inverse problems. We extend these ideas to the present setting with stochastic observations (see Theorem \ref{thm:estimate_error_high_probability} and Remark \ref{rem-opt}).

In particular, we identify explicit regularity conditions on the kernel $G$ 
and the trading speed $\dot Q$
such that the approximate source conditions hold and optimise the convergence rates of the estimator accordingly. The resulting convergence rates are optimal under the regularity conditions of the kernel $G$
 and are better than the worst-case convergence rates given in the statistical inverse problem literature \cite{blanchard2018optimal, Rastogi2020convergence}; see Remark \ref{rmk:rate_optimal}.
 The  method is general and can be applied for non-parametric estimation in other classes of infinite-dimensional stochastic control problems. 
}
\medskip \\
\textbf{Lipschitz stability of infinite-dimensional optimal control:}  
In \cite{basei2022logarithmic,guo2023reinforcement,szpruch2021exploration}, Lipschitz stability 
of optimal controls for  stochastic control problems was derived,   by showing that the optimiser is continuously differentiable with respect to  finite-dimensional  model parameters, hence establishing Lipschitz continuity.  
The Lipschitz stability of   controls is crucial for 
quantifying the precise performance gap between   controls derived from estimated and true models,
and for
characterizing the  regret
order of learning algorithms.

The performance gap in Theorem \ref{thm:gap} entails proving the Lipschitz  stability of the optimal control with respect to the (infinite-dimensional) kernel function $G$ (see Proposition \ref{prop:u_lipschitz}),
for which   the preceding argument 
developed   for  finite-dimensional parameters 
does not apply. Moreover,  in our setting the non-Markovianity introduced by $G$ turns the problem to infinite-dimensional stochastic control, in contrast to 
the finite-dimensional   control problems studied 
in the aforementioned references.
This major difference is reflected in the ingredients of optimiser. For example standard Riccati equations become operator-valued Riccati equations and solution to a BSDE becomes a solution to infinite dimensional BSDE (see Sections 6.2-6.3 of \cite{AJ-N-2022} for additional details). In this work, we 
establish  uniform boundedness and Lipschitz stability for all 
components of the optimiser 
in suitable norms (see Remark \ref{rmk: stability_estimate}).

Also note that for the control problem studied in this paper the running cost is not strongly concave in the control variable, and is not concave in the state variable (see \eqref{def:objective-str}). Such a (strong) concavity assumption is assumed in the aforementioned references in order to establish the required Lipschitz stability. We overcome this issue by imposing a nearly non-negativeness condition of the estimated kernel $G$ (see Definition \ref{assum1:parameters}) in order to prove stability.

\paragraph{Organisation of the paper:} 
In Section \ref{sec-l-p} we describe the reinforcement learning problem and present our main results on the convergence of the  estimator for propagator and on the regret bounds. Section \ref{sec-num} is dedicated to a numerical implementation of our propagator estimation results. In Section \ref{sec-results} we recall some essential results on the associated optimal liquidation problem. Section \ref{sec-bnd} is dedicated to the proof Theorem \ref{thm:gap}, which provides the bound on the performance gap. Section \ref{sec-est} deals with the analysis of the regularised least-squares estimator. Section \ref{sec-pf-rate} dedicated to the proof of Theorem \ref{thm:regret_general_rate}, which derives regret rate. Sections \ref{sec-pf-prop-u} and \ref{sec-pf-thm-lips} contain proofs for some auxiliary results. Finally in Appendix \ref{app-signal} we provide a regression-based algorithm for signal estimation and derive its convergence rate. 

 \section{Problem formulation and main results} \label{sec-l-p} 
 
This section studies  the optimal liquidation problem with unknown transient price impact kernel and temporary price coefficient $\theta^\star=(\lambda^\star,G^\star)$.  
 The agent's objective is to search for the optimal trading strategy while simultaneously learn the price dynamics, that is to learn $\theta^\star$. We first propose a least-squares estimator for these coefficients and derive its convergence rate. Then we present a phased-based learning algorithm and establish its regret bound.

 \subsection{Episodic learning  for optimal liquidation problems}
\label{sec:episodic_learning_setting}

 \paragraph{Optimal liquidation with known price impacts.}
 We first recall the optimal liquidation framework which was presented in \cite{AJ-N-2022}. \medskip 
 
 Let $T>0$ denote a finite deterministic time horizon and fix a filtered probability space~$(\Omega, \mathcal F,(\mathcal F_t)_{0 \leq t\leq T}, \P  )$ satisfying the usual conditions of right continuity and completeness.  We consider a semimartingale unaffected price process~$P=(P_t)_{0 \leq t\leq T}$ with a canonical decomposition  \begin{equation} \label{eq:P_decomposition}
 P_t = A_t + M_t, \quad 0\leq t \leq T,
\end{equation} 
into a predictable finite-variation signal process $A=(A_t)_{0 \leq t\leq T}$ and 
an independent  martingale $M$ satisfying $\mathbb{E}[M_0]=0$ and 
\be \label{ass:P} 
E \left[ \langle   M  \rangle_T \right] + E\left[\left( \int_0^T |dA_s| \right)^2 \right] < \infty.  
\ee
Let  $\theta^\star=(\lambda^\star,G^\star)\in (0,\infty)\times  L^2([0,T], \mathbb{R})$ be fixed 
coefficients such that 
 for every $f\in L^2\left([0,T],\mathbb R\right)$, 
 \be \label{pos-def_convolution}
\int_{0}^T\int_{0}^T G^\star(|t-s|)f(s)  f(t) dsdt \geq 0. 
\ee

\begin{remark} \label{R:boundedkernels}
Note that in \eqref{pos-def_convolution} we consider a class of non-negative kernels. 
An important subclass of these kernels is the class of bounded 
non-increasing convex functions (see Example 2.7 in \cite{GSS}). 
See  Remark \ref{rmk:G_regularity_example}
for further examples. 
\end{remark}

We consider a trader with an initial position of $q>0$ shares in a risky asset. The number of shares the trader holds at time $t\in [0,T]$ is prescribed as 
    \begin{align} \label{def:Q}
    Q_t^u = q -\int_0^t u_s ds, 
    \end{align}
where $(u_s)_{s \in [0,T]}$ denotes the trading speed which is chosen from the set of admissible strategies
\be \label{def:admissset} 
\mathcal A \triangleq \left\{ u \, : \, u \textrm{ progressively measurable s.t. } \mathbb E\left[ \int_0^Tu_s^2 ds \right] <\infty \right\}.
\ee
We assume that the trader's trading activity causes price impact on the risky asset's execution price. In order to define the price impact effects we introduce some additional definitions. 
For any trading speed $u\in \cA$,
the   price process $S^u$ satisfies the following dynamics: 
 for all $t\in [0,T]$, 
 \be \label{def:S_str}
S^u_{t} := P_{t} - \lam^\star u_t -  Z_t^{\theta^\star, u},
\qquad 
\textnormal{with 
$Z_t^{\theta^\star, u}=   \int_0^t G^\star(t-s)u_s ds$}.
\ee
Note that $\lambda^\star$ is the temporary price impact coefficient and $Z_t^{\theta^\star, u}$ is the transient price impact term, which is associated with the price impact kernel $G^\star$, also known as the propagator. 
 
Consider maximising  the following risk-revenue functional  
  over $u\in \mathcal A$: 
\begin{equation} \label{def:objective-str}
\begin{aligned}
& J^{\theta^\star}(u ) := \mathbb{E} \Bigg[ \int_0^T S^u_t   u_t dt  + Q_T^u P_T   -\phi \int_0^T (Q_t^u)^2 dt - \varrho (Q_T^u)^2 \Bigg].
\end{aligned}
\end{equation}
The first two terms in~\eqref{def:objective-str} represent the trader's terminal wealth;  that is, her final cash position including the accrued revenue  as well as her remaining final risky asset position's book value. The third and fourth terms in~\eqref{def:objective-str} implement a penalty $\phi \geq 0$ and $\varrho \geq 0$ on her running and terminal inventory, respectively. 
Observe that $J(u) < \infty$ for any   strategy $u \in \mathcal A$.

If the agent knows $\theta^\star$, then \eqref{def:S_str}-\eqref{def:objective-str} is a special case of the Volterra stochastic control problem studied in \cite{AJ-N-2022}.
By Proposition 4.5 therein, the optimal trading strategy
 $u^{\theta^\star}$ 
 is given by 
	\begin{align}	\label{eq:optimalcontrol_monotone_str}
	 u^{\theta^\star}_t = a^{\theta^\star}_t + \int_0^t B^{\theta^\star}(t,s)u^{\theta^\star}_s ds, \quad 0\leq t \leq T, 
	\end{align}
 where $a^{\theta^\star}$ is a stochastic process satisfying \eqref{def:admissset}, depending on $A$ but not on $M$ in \eqref{eq:P_decomposition}, and $B^{\theta^\star}$ is a function satisfying 
$
\sup_{t\leq T} \int_0^t(B^{\theta^\star}(t,s))^2ds<\infty,  
$ 
see \eqref{eq:aB} for the precise definition.  
We emphasise the dependence of $u^{\theta^\star} $ in \eqref{eq:optimalcontrol_monotone_str}  by writing
  $$u^{\theta^\star} = 
\textsf{Greedy}\left( A, \theta^\star
\right).$$
{ 
  \paragraph{A note about observables.}
Recall that the visible price process $S^u$ was introduced in \eqref{def:S_str}. In addition to this observable the agent clearly knows her own trading rate $u$, which impacts $S^u$. Recall that the fundamental price process $P$ was defined in \eqref{eq:P_decomposition}. While $P$ is unobserved by the trader, it is a common practice that the short term price predicting signal $A$ (also called \emph{alpha}) is an observable, typically obtained from limit order books real-time data. We briefly survey some well known examples for such signals which impact the price at different time scales. In Section 4 of \cite{Lehalle-Neum18} a detailed statistical analysis of the limit order book imbalance signal was performed. The effect of this signal on future price moves was demonstrated in time intervals of the $10$ next trades. The usage of this signal by high frequency proprietary traders was also proved statistically. The order flow imbalance signal has been extensively studied in the literature, in particular the correlation between the current order flow and the future price move in 10 seconds intervals was studied by \citet{cont14}. More examples of observed trading signals which are used in optimal execution can be found in a practitioners presentation by Robert Almgren \cite{almg-pres}. 
In reality the agent also determines the penalty parameters $\phi, \varrho$ in the quadratic costs \eqref{def:objective-str}, however the parameters  $\theta^\star = (\lam^*,G^*)$ are unknown and are subject to estimation in this paper. 
}
  \paragraph{Optimal liquidation with unknown price impacts.}

 In this work, we consider an agent  who 
  repeatedly liquidates the risky asset in \eqref{def:S_str}
  without knowing the price impact coefficient $\theta^\star$.
    This is often referred to as the episodic  (also known as reset or restart) learning 
framework in the reinforcement learning literature. 
  The agent will
 improve  her knowledge of   $\theta^\star$ through successive learning episodes,
 while  simultaneously optimise the objective \eqref{def:objective-str}.
 In reality the agent knows 
the dynamics of $A$ in \eqref{eq:P_decomposition}, the form of \eqref{def:S_str} (excluding the coefficient $\theta^\star$) and the penalty parameters $\phi, \varrho$ in the quadratic costs \eqref{def:objective-str}. We will therefore assume that these are known features of the model in the following.  
For each episode, the agent observes (a realisation of) the price $S$ and the signal  $A$, but not the      noise $M$.  
Additional   regularity properties of 
  $\theta^\star$
will be assumed in order to optimise the    learning  algorithm    (see Assumptions 
\ref{assume:kernel_assumption} 
and \ref{assumption:theta_bound}).

Mathematically, the learning problem is described as follows. 
Let $(\Omega, \mathcal{F},\mathbb{P})$
be a probability space,
let 
$(A^m)_{m\in \mathbb{N}}$
and $(M^m)_{m\in \mathbb{N}}$   be   mutually  independent copies of $A$ and $M$
 on $(\Omega, \mathcal{F},\mathbb{P})$,
respectively,
and for each $m\in \mathbb{N}$,
let  $P^m=A^m+M^m$.
Here $A^m$ and 
$M^m$ 
correspond  to the observed signal 
and 
unobserved martingale noise
for the $m$-th learning episode, respectively.
For each episode, the agent   interacts with \eqref{def:S_str} by  
choosing  controls   that are adapted to available observations.
These   
 admissible controls  
 and observation filtrations
are defined recursively as follows. 
The   observation before the first episode is given by 
 the $\sigma$-algebra 
$\cF_0=\cN$,
where 
$\cN$ is  the $\sigma$-algebra 
 generated by the $\sP$-null set. 
For the $m$-th  episode with $m\in \sN$,
taking the   $\sigma$-algebra $\cF_{m-1}$,
the agent executes a square-integrable control $u^m$  
that is progressively measurable with respect to the filtration 
$(\cG^m_t)_{t\in [0,T]}$
with 
$\cG^m_t\coloneqq \cF_{m-1}\vee \sigma\{A^m_s\mid  s\in [0,t]\}$,
and  observes a trajectory of the   price process 
$S^{m}$   governed by  the following dynamics (cf.~\eqref{def:S_str}): 
\begin{equation} \label{def:S_m}
S^{m}_{t} = A^m_{t}+M^m_t - \lam^\star u^m_t -     \int_0^t G^\star(t-s)u^m_s ds. 
\end{equation}
The  available information for the agent before the $(m+1)$-th episode 
is   
$\cF_m\coloneqq \cF_{m-1} \vee \sigma\{S^{m}_{t} , A^m_t\mid  t\in [0,T]\}$.

To measure the performance of 
the  controls $(u^m)_{m\in \mathbb{N}}$
(also referred to as an learning algorithm)
 in this setting, one widely adopted criteria is the  regret of learning  \citep{guo2023reinforcement, basei2022logarithmic, gao2022logarithmic,gao2022square}:
 for each $N\in \mathbb{N}$, the regret   of learning up to $N$-th episode is given by
\begin{equation}\label{eq:regret_def}
 R(N) = \sum_{m=1}^N \big(   
 J^{\theta^\star}(u^{\theta^\star}) -J^{\theta^\star}(u^m)\big),
 \end{equation}
  where $ J^{\theta^\star}(u^{\theta^\star})$ is the  optimal value  that agent can achieve knowing the parameter $\theta^\star$,
  and $J^{\theta^\star}(u^m) $ is the 
  expected  performance of the control $u^m$ for the $m$-episode\footnotemark
   \footnotetext{With a slight abuse of notation, we
 denote by $J^{\theta^\star}(\cdot)$
  the performance functional for all  episodes,
  without specifying its   dependence on $m$. 
It is possible as  $P^m$ is independent of $\cF_{m-1}$. 
  }:
\begin{equation} \label{def:objective-m}
\begin{aligned}
& J^{\theta^\star}(u^m) \coloneqq \mathbb{E} \Bigg[
 \int_0^T S^m_t  u^m_t dt  + Q_T^{u^m} P^m_T   -\phi \int_0^T (Q_t^{u^m})^2 dt - \varrho (Q_T^{u^m})^2
 \,\bigg\vert\, \cF_{m-1} \Bigg],
\end{aligned}
\end{equation}
with $Q^{u^m}_t=q-\int_0^t u^m_s d s$ being the corresponding inventory (cf.~\eqref{def:Q}). 
The expectation in \eqref{def:objective-m} is only taken with respect to   $P^m$, and
hence 
$ J^{\theta^\star}(u^m)   $ is a random variable depending on the realisations of the signals 
$(A^n)_{n=1}^{m-1}$
and noises 
$(M^n)_{n=1}^{m-1}$.
Intuitively, 
the regret $ R(N)$ characterises the cumulative expected loss from taking sub-optimal controls up to the $N$-th  episode. 
Agent's aim is to construct a learning algorithm for which the regret
$  R(N)$ 
 grows sublinearly in $N$   in high probability.

Note that  the above setting 
assumes the algorithm runs indefinitely without a prescribed maximal  number of  learning episodes. 
The agent 
will then
derive an \emph{anytime} learning algorithm (see e.g., \cite{lattimore2020bandit,szpruch2021exploration}), i.e., an algorithm whose implementation does
not require advance knowledge of the algorithm termination time and whose performance
guarantee holds for all learning episodes;
see Remark \ref{rmk:phased_based_algorithm} for more details.

 \subsection{A least-squares estimator and its convergence rate}
\label{sec:lse}

In this section we derive the identifiability of  $\theta^\star=(\lambda^\star, G^\star)$ under suitable exploratory strategies. 
We propose a regularised least-squares estimator based on observed trajectories and analyse its finite sample accuracy.   
The estimator will be employed in Section \ref{reg-alg-sec}  to design a regret optimal learning algorithm for \eqref{def:objective-str}. 
By an abuse of notation,  
we will index the observed trajectories for the 
estimator by $m$, 
which   is typically different from 
the number of learning episodes   in Section   \ref{reg-alg-sec}.

More precisely,  let  $A$ and $M$ the    processes
   in \eqref{def:S_str},
and   let   $(A^m)_{m\in \mathbb{N}}$ and  $(M^m)_{m\in \mathbb{N}}$  be mutually  independent copies of 
  $A$ and $M$, respectively,  defined on
the probability space $(\Omega, \mathcal{F},\mathbb{P})$.
 The agent    executes a    trading strategy 
$u^e\in L^2([0,T],\mathbb{R})$, 
and 
estimates  $\theta^\star=(\lambda^\star, G^\star)$ using  the corresponding price  trajectories $(S^m,A^m)_{m\in \mathbb{N}}$,
where for all  $m\in \mathbb{N}$,
$(S^m,A^m)_{t\in [0,T]}$ satisfies 
for all $t\in [0,T]$,
\begin{align} \label{def:S_m_estimation}
\begin{split}
S^m_{t} 
&= A^m_{t}+M^m_t - \lam^\star u^e(t) -     \int_0^t G^\star(t-s)u^e(s) ds
\\
& = A^m_{t}+M^m_t - \lam^\star u^e(t) -   (\boldsymbol{u}^e G^\star)(t),
\end{split}
\end{align} 
with  
 $\boldsymbol{u}^e :  L^2([0,T],\mathbb{R})\to L^2([0,T],\mathbb{R}) $ being the integral   operator 
 defined by 
\begin{equation}
\label{eq:operator_u}
(\boldsymbol{u}^e f)(t)\coloneqq 
\int_0^t u^e({t-s})  f(s) d s,
\quad   f\in L^2([0,T],\mathbb{R}). 
\end{equation} 
Note that in \eqref{def:S_m_estimation}, $G^\star$ plays the role of an  unknown  function instead of a kernel. 
 
 The following regularity condition on $u^e$ is imposed for the identifiability of $\theta^\star$.
Recall that 
  $H^1([0,T],\sR)$ is  the space of absolute continuous functions $f:[0,T]\to 
\sR$   whose    derivative (which exists a.e.) belongs to   $L^2([0,T],\sR)$.
 
\begin{assumption}
\label{assum:learning}
 $u^e\in H^1([0,T],\sR)$ 
 and $u^e(0)\not =0$.
   
\end{assumption}
 
 \begin{remark} 
Assumption \ref{assum:learning} holds for any nonzero constant strategy or classical 
 trading strategies in the Almgren--Chriss framework (see e.g. \cite[Chapter 6]{cartea15book}). 
 Unfortunately, the trajectories of the greedy strategy $u^{\theta^\star}$ in \eqref{eq:optimalcontrol_monotone_str} may not  satisfy Assumption \ref{assum:learning}.
 Indeed, $u^{\theta^\star}_0$ could be zero,  due to the randomness of the signal process $A$. Moreover,   the time regularity of   $u^{\theta^\star}$ relies on   the   regularity of $a_t$ and $B(t,\cdot)$  with respect to $t$, which subsequently depends on   
 the path regularity 
 of the conditional expectations
 of the signal $A$ (cf.~\eqref{eq:aB}). 
 Even for the special case 
  with   $A\equiv 0$, it is still   challenging to obtain explicit conditions for the differentiability of $t\mapsto B(t,\cdot)$ and $t\mapsto a_t$ to ensure that   $u^{\theta^\star}\in H^1([0,T],\sR)$. 
 
 \end{remark} 
 
Under Assumption \ref{assum:learning}, 
 the operator 
  $\boldsymbol{u}^e:   L^2([0,T],\sR)\to   L^2([0,T],\sR) $ is injective
as shown in   Lemma  \ref{lemma:u_operator_property}.
This indicates that $\theta^\star$  can be uniquely identified
based on sufficiently many  trajectories $(S^m,A^m)_{m\in \mathbb{N}}$.
In the sequel, we propose  a regularised least-squares estimator for $\theta^\star$ and analyse its finite sample accuracy. 

 By   \eqref{def:S_str} and $\sE[M_0]=0$,  
 $\lambda^\star u^e(0)=-\sE[S_0-A_0]$. Replacing the expectation by an empirical mean 
yields  the following estimation for $\lambda^\star$: 

\begin{equation}
\label{eq:lse_lambda}
\lambda^N \coloneqq -\frac{1}{Nu^e(0)}\sum_{m=1}^N (S^m_0-A^m_0), \quad \textrm{for } N\in \sN, 
\end{equation}
 which is well-defined as $u^e(0)\not =0$. 
Given the estimators 
$(\lambda^N)_{N\in \sN}$,
we  then introduce   a sequence of projected least-squares estimators
for   the kernel    $G^\star$.
To this end,  
let  
$\mathscr{P}_{[0,T]}$ be the collection of all partitions of $[0,T]$,
and 
let $(\pi_N)_{N\in \sN} \subset \mathscr{P}_{[0,T]}$ 
be  such that 
 $\pi_N=\{0=t^{(N)}_0<\cdots<t^{(N)}_N=T\}$ for all $N\in \sN$
and 
 $\lim_{N\to \infty}|\pi_N|=0$, where 
 $|\pi_N| \coloneqq \max_{i=0,\ldots, N-1}(t^{(N)}_{i+1}-t^{(N)}_i)$ is the mesh size of $\pi_N$. 
For each $N\in \sN$, let $V_N $ be the space of 
piecewise constant functions on  $\pi_N$: 
\begin{align} 
\label{eq:piecewise_constant}
\begin{split}
V_N
&=\left\{ f\in L^2([0,T],\sR)  \,\Big\vert \,
f_t=\sum_{i=0}^{N-1} f_{t_i} \mathds 1_{[t^{(N)}_i,t^{(N)}_{i+1})}(t),
\;
\textnormal{for all $t\in [0,T]$}\right\}.
\end{split}
\end{align}
Then for each 
   regularising weight   $\tau_N>0$,    consider 
   minimising the following  $L^2$-regularised 
   least-squares estimation error over $V_N$ (cf.~\eqref{def:S_m_estimation}): 
   \begin{align}\label{eq:lse_G}
\begin{split}
G^N 
&\coloneqq 
 \argmin_{G\in  V_N}
 \left(
 \frac{1}{N}\sum_{m=1}^N    \|S^m 
-A^m
 +\lambda^N u^e
 +  \boldsymbol{u}^e G\|_{L^2([0,T])} ^2 
+\tau_N \|G\|_{L^2([0,T])}^2
 \right),
 \end{split}
 \end{align} 
which is derived by replacing $\lambda^\star$ in \eqref{def:S_m_estimation}
with $\lambda^N$. 
As $\tau_N>0$, it is easy to see that the quadratic functional in \eqref{eq:lse_G}
has a unique minimum and hence 
 $ {G}^N $ 
is well-defined.
\begin{remark} 
Here, we take  
 $(V_N)_{N\in \sN}$ to be spaces of piecewise constant functions
  for the clarity of presentation, 
 but   the estimator \eqref{eq:lse_G} and its convergence  analysis   can be   extended to any subspaces 
  $(V_N)_{N\in \sN}$   of  $L^2([0,T],\sR)$ such that
 $\ol{\bigcup_{N=1}^\infty V_N} =L^2([0,T],\sR)$; 
 see Section    \ref{reg-alg-sec} for details. 
 \end{remark} 
For notational simplicity, we write $\theta^N=(\lambda^N,G^N)$ 
in \eqref{eq:lse_lambda} and \eqref{eq:lse_G} as 
\begin{equation}
\label{eq:lse_abbreviation} 
\theta^N = 
\textsf{LSE}\left(
(S^m,A^m)_{1\le m\le N}, \tau_N, \pi_N
\right),
\end{equation}
which emphasises  the dependence    on the data $(S^m,A^m)_{1\le m\le N}$, the regularising weight $\tau_N$
and the mesh size of $\pi_N$. 

\begin{remark} \label{rem-t} 
The fact that $\tau_N>0$ is critical for the well-posedness of \eqref{eq:lse_G}. 
Indeed, 
 consider  $\tau_N=0$, $V_N = L^2([0,T],\sR)$ and  $u^e\equiv 1$. 
Then  \eqref{def:S_m_estimation} and  \eqref{eq:lse_G} suggest  that 
\begin{align}
\label{eq:lse_tau=0}
\begin{split}
 G^N(t)   &= -\frac{1}{N}\sum_{m=1}^N (d S^m_t-d A^m_t)
= -\frac{1}{N}\sum_{m=1}^N d M^m_t +   G^\star(t) ,
\quad \forall t\in [0,T].
\end{split}
\end{align}
As a non-constant continuous martingale has infinite variation,  $G^N \in  L^2([0,T],\sR)$ satisfying \eqref{eq:lse_tau=0} does not exist in general. 

The above observation also indicates that 
a proper scaling of      the regularising weight 
$\tau_N $ with respect to the sample size $N$ 
is crucial   for the smoothness and convergence   of $(G^N)_{N\in \sN}$.
Reducing the weight $\tau_N$ too fast essentially 
 fits the time fluctuation of $(M^m)_{m=1}^N$,
and hence  
  leads to an irregular estimate $G^N$. 
This is in contrast to   the regularised least-squares estimator 
  for parametric models  as in \cite{basei2022logarithmic}.
  The regularising weights  $(\tau_N)_{N\in \sN}$  therein can be chosen  as any vanishing sequence 
  such that $\lim\sup_{N\to\infty}\sqrt{N}\tau_N<\infty $.
  
\end{remark} 

The   dependence of
$ \tau_N$ on $N$
 results in a slower convergence of 
$(G^N)_{N\in \sN}$ compared with the 
$\cO(N^{-1/2})$ order for
classical Monte-Carlo methods. 
It is known that 
the optimal choice  of   $(\tau_N)_{N\in \sN}$ 
 depends on the regularity of  the true kernel $G^\star$
 (also known as  the ``source condition" 
in inverse problem literature 
\cite{hofmann2006approximate, blanchard2018optimal}).
We impose the following  regularity conditions on the kernel 
$G^\star$.

\begin{assumption}
\label{assume:kernel_assumption}
$G^\star\in L^2([0,T],\sR)$ 
is differentiable a.e.,  and is one of the two types: 
\begin{enumerate}[(1)]
\item 
\label{label:regular_kernel}
Regular kernel:
$G^\star\in H^1([0,T],\sR)$ and $G^\star(T) \not =0$.
\item 
\label{label:singular_kernel}
Power-type singular kernel:
 there exists $\alpha\in (0,1/2)$, $t_0\in (0,T)$ 
 and $C_0>0$
 such that 
  $|\frac{d}{dt}{G}^\star(t)|\le C_0 t^{-\alpha-1}$
     for a.e.~$t\in (0,t_0)$,
  and  ${G}^\star\in H^1([t_0,T],\sR)$. 
\end{enumerate}
\end{assumption}
 
 \begin{remark}
\label{rmk:G_regularity_example}

Note that 
Assumption \ref{assume:kernel_assumption}\ref{label:regular_kernel}
requires   $G^\star$ to be continuous on $[0,T]$,
due to   Morrey's inequality. 
It is satisfied by 
  the exponential   kernel $G^\star(t)=e^{-\beta t}$ for $\beta>0$  proposed by 
\cite{Ob-Wan2005}
and the truncated power law kernel $G^\star(t)=  (c_0+t)^{-\beta}$  for some $\beta,c_0>0$   studied in 
\cite{bouchaud-gefen, gatheral2010no}. 
On the other hand, 
{Assumption \ref{assume:kernel_assumption}\ref{label:singular_kernel}
allows for a  power-type singularity   at $t=0$.
It includes as a special case 
 the power law kernel  $G^\star(t)= t^{-\beta}$, for any $0<\beta \leq \alpha$ proposed in \cite{gatheral2010no}. Note that the constant $\alpha$, which determines the range of power law singularities allowed in the kernel $G^\star$, is well documented in the literature, both by non-rigorous empirical estimates using historical data (see e.g. \citep{bouchaud-gefen,citeulike:13990970}) and from theoretical arguments (see \cite{gatheral2010no} and Chapter 13 of \cite{bouchaud_bonart_donier_gould_2018}).}
 \end{remark}

 In the sequel, we assume that the agent knows the precise type of $G^\star$ as in Assumption \ref{assume:kernel_assumption},
i.e., $G^\star$ is regular  on $[0,T]$
or 
admits a power-type singularity at zero  
with   known component $\alpha$.
This allows for 
 specifying  the  precise  decay rate   of 
$(\tau_N)_{N\in \sN}$
in \eqref{eq:lse_G}
and  then
establishing the convergence rate of $(G^N)_{N\in \sN}$.
 To quantify the   convergence rate  
of   $(\lambda^N,G^N)_{N\in \sN}$
 in high probability, 
we impose the   following concentration condition on the martingale  process $M$.

 \begin{assumption}
\label{assum:concentration_M}

There exists  $C_M>0$ 
such that for all $N\in \sN$ and $ {\eta}>0$,
$$
\sP\Bigg(
\left|
\frac{1}{N}\sum_{m=1}^N M_0^m\right|^2
+\left\|\frac{1}{N}\sum_{m=1}^N M^m \right\|^2_{L^2([0,T])} 
\ge C^2_M (\log(2 {\eta}^{-1}))^2 N^{-1}\Bigg)
\le  {\eta}.
$$
\end{assumption}  

{
The following lemma provides a sufficient condition for Assumption \ref{assum:concentration_M}.

\begin{lemma}
\label{lemma:bernstein}
  There exists $L,\sigma>0$ such that for all $p\ge 2$,
$
\sE[(|M_0|^2 +\|M\|^2_{L^2([0,T])})^{p/2}]\le \frac{1}{2}p!\sigma^2 L^{p-2}$.  Then Assumption \ref{assum:concentration_M} holds with $C_M=2(L+\sigma)$. 
\end{lemma}

The proof of Lemma \ref{lemma:bernstein} follows by applying  \cite[Proposition A.1]{blanchard2018optimal} to the  random variable $Z\coloneqq (M_0,M)$ taking values  in the Hilbert space $\sR\times L^2([0,T])$.
The moment condition   in Lemma \ref{lemma:bernstein} is commonly referred to as a Bernstein-type assumption. It is often imposed on the observation noise distribution in the statistical inverse problem literature for conducting complexity analysis  \cite{blanchard2018optimal, Rastogi2020convergence}. For martingales   given by stochastic integrals 
with respect to Brownian motions or   Poisson   measures,
this moment  condition  can be verified by 
 Burkholder's inequality as in \cite{guo2023reinforcement}. 
In the sequel, we will directly work with Assumption \ref{assum:concentration_M}, which is   more general than the Bernstein assumption, 
and is sufficient  for obtaining the optimal convergence  rate of 
\eqref{eq:lse_abbreviation} in high probability. 
}

Under Assumptions \ref{assume:kernel_assumption} and \ref{assum:concentration_M},
the following theorem chooses the optimal regularising   weights $(\tau_N)_{N\in \sN}$ and 
mesh sizes $(|\pi_N|)_{N\in\sN}$, 
and quantifies  the convergence rate of  $(\lambda^N,G^N)_{N\in \sN}$ in high probability.  
It follows as a special case 
of Theorem \ref{thm:estimate_error_general} in  
Section \ref{sec-est}.

   \begin{theorem}
  \label{thm:estimate_error_high_probability}
  Suppose that Assumptions \ref{assum:learning} 
  and \ref{assum:concentration_M}
  hold. Let $C\ge 1$. 
  \begin{enumerate}[(1)]
  \item 
  If Assumption \ref{assume:kernel_assumption}\ref{label:regular_kernel} holds,
    then for all $\eta \in (0,1) $, 
       by setting  
$(\tau_N)_{N\in \sN}\subset (0,\infty) $ and $(\pi_N)_{N\in \sN} \subset \mathscr{P}_{[0,T]}$ such that 
for all $N\in \sN$,
 \begin{equation}\label{eq:lse_parameter_regular}
  \textstyle
  \frac{1}{C} \left(\frac{  \log(  {\eta}^{-1})+ \log N  }{\sqrt{N}}
\right)^{\frac{4 }{3}} \le \tau_N \le   C \left(\frac{  \log(  {\eta}^{-1})+ \log N  }{\sqrt{N}}
\right)^{\frac{4 }{3}},
\quad 
|\pi_N| \le C {\tau_N}^{\frac{1 }{2}}, 
\end{equation}
   it holds  with probability at least $1-\eta$ that,
for all    $N\in \sN\cap [2,\infty)$, 
\begin{equation}
\label{eq:lse_conv_regular}
 \textstyle
  |\lambda^N-\lambda^\star| 
  \le C'  \left( \frac{  \log(  {\eta}^{-1})+ \log N  }{\sqrt{N}} \right),
 \quad
 \|G^N -G^\star\|_{L^2([0,T])} 
 \le C'  \left( \frac{  \log(  {\eta}^{-1})+\log N  }{\sqrt{N}} \right)^{\frac{1}{3}}.
\end{equation}

  \item 
  If Assumption \ref{assume:kernel_assumption}\ref{label:singular_kernel} holds,
    then for all $\eta\in (0,1)$, 
       by setting  
$(\tau_N)_{N\in \sN}\subset (0,\infty) $ and $(\pi_N)_{N\in \sN} \subset \mathscr{P}_{[0,T]}$ such that 
for all $N\in \sN$,
  \begin{equation}
  \label{eq:lse_parameter_singular}
  \textstyle
  \frac{1}{C} \left(\frac{  \log(  {\eta}^{-1})+ \log N   }{\sqrt{N}}
\right)^{\frac{4}{3-2\alpha}} \le \tau_N \le   C \left(\frac{  \log(  {\eta}^{-1})+ \log N  }{\sqrt{N}}
\right)^{\frac{4}{3-2\alpha}},
\quad 
|\pi_N| \le C {\tau_N}^{\frac{1 }{2}}, 
\end{equation}
   it holds  with probability at least $1-\eta$ that,
   for all    $N\in \sN\cap [2,\infty)$, 
\begin{equation}
\label{eq:lse_conv_singular}
 \textstyle
    |\lambda^N-\lambda^\star| 
  \le C'  \left( \frac{  \log(  {\eta}^{-1})+\log N   }{\sqrt{N}} \right),
 \quad
 \|G^N -G^\star\|_{L^2([0,T])} 
 \le C'  \left( \frac{  \log(  {\eta}^{-1})+ \log N   }{\sqrt{N}} \right)^{\frac{1-2\alpha }{3-2\alpha}}.
\end{equation}
\end{enumerate}
  The constant $C'>0$ appearing in 
  \eqref{eq:lse_conv_regular}
and 
\eqref{eq:lse_conv_singular}
is  independent of $\eta$ and $N$. 
  \end{theorem}
  The proof of Theorem \ref{thm:estimate_error_high_probability} is given in Section \ref{sec-est}. 
  
 \begin{remark} 
 \label{rem-opt}
 Theorem \ref{thm:estimate_error_high_probability}
 is proved by 
   first  interpreting  
 \eqref{eq:lse_G}
 as Tikhonov
regularisation of   \eqref{def:S_m_estimation},
and then adapting  existing theoretical frameworks of 
Tikhonov regularisation   with deterministic observations
  (see  \cite{groetsch1984theory, hofmann2006approximate} and references therein) to the present setting with random observations.
 The crucial step in our argument is to quantify the distance between the true kernel $G^\star$
   to the range of the operator 
  $\boldsymbol{u}^e$ in \eqref{eq:operator_u}, which is 
  characterised by the  behavior  of a distance function  $R\mapsto \mathscr{D}(R)$ 
  for large $R$
(cf.~\eqref{eq:distance_function}).
We prove 
in Theorem \ref{thm:convergece_specifici_decay} that 
the function $\mathscr{D}$ decays as a power function as $R\to \infty$, whose exponent depends explicitly on Assumption \ref{assume:kernel_assumption}. 
 To the best of our knowledge, 
 such a power-type decay of $\mathscr D$ has only been established 
 in the literature for $u^e\equiv 1 $ and $G^\star \equiv 1$
 (see e.g., \cite{hofmann2006approximate}).
 We further prove in Proposition \ref{prop-opt} that the  exponents of these power functions   
are optimal, i.e., 
they   are   the maximal power-type decay rates 
of $\mathscr D$ under 
Assumption \ref{assume:kernel_assumption}. Specifically, 
  Proposition \ref{prop-opt}  considers      the power law kernel $G^\star(t)= t^{-\alpha}$,
  which satisfies    Assumption \ref{assume:kernel_assumption}\ref{label:regular_kernel} if $\alpha=0$,
  and       Assumption \ref{assume:kernel_assumption}\ref{label:singular_kernel} if $\alpha\in (0,1/2)$.

\end{remark}

{
\begin{remark}
\label{rmk:rate_optimal}
 The convergence  rates given in Theorem  \ref{thm:estimate_error_high_probability} are optimal (up to a logarithmic order in $N$) under Assumption \ref{assume:kernel_assumption}. 
In order to see this, note that the estimation of $G^\star$ can be interpreted as an inverse problem  with the forward operator 
  $\boldsymbol{u}^e$
  based on noisy observations, 
  where the 
   noise level is of the magnitude $\cO(1/\sqrt{N})$ with high probability; see Theorem \ref{thm:estimate_error_general}.
   By the order optimality result  \cite[Proposition 3.15]{engl1996regularization}, under the source condition that $G^\star =((\boldsymbol{u}^e)^*\boldsymbol{u}^e)^\mu w $ for some $w\in L^2([0,T];\sR)$ and $\mu>0$, no estimation algorithm can recover $G^\star$  with a rate faster than $\cO(N^{-\frac{\mu}{2\mu+1}})$ as $N\to \infty$.  
   We then characterise the worst $\mu$ for a kernel $G^\star$ 
   satisfying Assumption \ref{assume:kernel_assumption}. Recall that by \cite[Proposition 3.13]{engl1996regularization}, $G^\star$ satisfies the  source condition with $\mu>0$ if and only if 
   $\sum_{n=1}^\infty \frac{1}{\sigma_n^{4\mu  }}\langle G^\star,\mathfrak{u}_n \rangle^2_{L^2([0,T])}<\infty,
   $ 
   where 
$\sigma_1\ge \sigma_2\ge \ldots >0$  is  the  singular values  of   $\boldsymbol{u}^e $,
and $(\mathfrak{u}_n)_{n\in \sN} $ is the orthonormal system  of eigenfunctions of $(\boldsymbol{u}^e)^*\boldsymbol{u}^e$. 
 Now   consider  the power law kernel $G^\star(t)= t^{-\alpha}$, which   satisfies    Assumption \ref{assume:kernel_assumption}\ref{label:regular_kernel} if $\alpha=0$,
  and       Assumption \ref{assume:kernel_assumption}\ref{label:singular_kernel} if $\alpha\in (0,1/2)$.
 By \eqref{eq:singular_sum}
 and the above criterion, 
 the power law kernel 
   $G^\star(t)= t^{-\alpha}$ satisfies the source condition for   
  $\mu< \frac{1}{2}(\frac{1}{2}-\alpha)$ (but not $\mu= \frac{1}{2}(\frac{1}{2}-\alpha)$).
This suggests that under  Assumption \ref{assume:kernel_assumption}, the optimal rate is not greater  than 
   $\cO(N^{-\frac{1}{2}\frac{1-2\alpha }{3-2\alpha}})$ as $N\to \infty$.
This lower   rate of convergence matches the  
 upper   rate of convergence    in Theorem \ref{thm:estimate_error_high_probability}
   (up to a logarithmic term), 
which indicates the parameter choices in Theorem \ref{thm:estimate_error_high_probability} are order optimal under  Assumption \ref{assume:kernel_assumption}.

Note that the convergence rates in Theorem \ref{thm:estimate_error_high_probability} are better than the lower rates for general statistical inverse problems with random input and output variables \cite{blanchard2018optimal, Rastogi2020convergence}.
This is due to the usage   of  a deterministic trading strategy $u^e$ in \eqref{def:S_m_estimation}, which allows for applying deterministic inverse problem theory to achieve an improved rate.
 
\end{remark}
}

 \subsection{Phased-based learning algorithm and   its regret bound} \label{reg-alg-sec}

By leveraging Theorem \ref{thm:estimate_error_high_probability}, in
this section we propose a phased-based algorithm 
for learning \eqref{def:S_str}-\eqref{def:objective-str}.
The algorithm 
  alternates between exploration and exploitation phases, 
 and  achieves sublinear regrets with high probability.
%

\paragraph{Admissible estimated models.}

We first  introduce a class of estimated models based on which the greedy policies are constructed during the learning process. 
To facilitate the regret analysis, 
we assume that the agent knows the order of magnitude of the true parameter $\theta^\star$, from heuristic estimations using historical data (see \cite{bouchaud-gefen}, \cite[Chapter 13]{bouchaud_bonart_donier_gould_2018},   \cite[Chapter 6.2]{cartea15book} and \cite{ulrich_19} among others). Note that the constant $L$ which is defined below is a known parameter of the problem along with $\phi,\varrho$ in \eqref{def:objective-m}. 

\begin{assumption}
\label{assumption:theta_bound}
    There exists a known constant $L>0$ such that 
    $L^{-1}<\lambda^\star< L$
    and $\|G^{\star}\|_{L^2([0,T])}<L$.
\end{assumption}


\begin{definition} [Class of admissible parameters $ \Xi_{\eps}$]  \label{assum1:parameters}
  Let $L>0$ be the constant in Assumption \ref{assumption:theta_bound}.
For each   $\eps\in (0,L^{-1}/2)$,
   define
  $ \Xi_{\eps}$ 
to be  the set 
containing all   $(  \lambda, G) \in \sR\times L^2([0,T];\sR) $
such that
$L^{-1}\le   \lambda \le L$,
$ \|G\|_{L^2([0,T])} \le L$,
and  
\be \label{eps-pos-def}
\int_{0}^T\int_{0}^T G(|t-s|)  f(s)  f(t) dsdt 
\geq -\eps \|f\|^2_{L^2([0,T])}, \quad \textrm{for all }  f\in L^2([0,T],\sR). 
\ee
\end{definition}

\begin{remark}
Recall that 
by Theorem \ref{thm:estimate_error_high_probability},
the estimator \eqref{eq:lse_G} only approximates the true kernel 
$G^\star$ in the $L^2$ sense,
and hence may not be non-negative definite.
Thus, 
 Definition \ref{assum1:parameters} only 
 requires the estimated kernel $G$ to be nearly non-negative definite
relative to the estimated $\lambda$, as reflected by  \eqref{eps-pos-def} and $\eps\in (0,L^{-1}/2)$.
Since  
$\lambda^\star>0$ and 
$G^\star$ is 
  non-negative definite (see
 \eqref{pos-def_convolution}),
this condition can be satisfied 
by   estimated models  
with sufficiently many samples, 
as shown in Lemma \ref{lemma:initial_exploration}.  
\end{remark}

Definition \ref{assum1:parameters} ensures that 
the   greedy policy $u^\theta$ is well-defined 
 for any admissible model
$\theta\in \Xi_{\eps}$. 
 Moreover, Theorem \ref{thm:gap} shows that 
  the performance gap of the  
  greedy policy  of an estimated model depends quadratically on the model error.
 The proof of Theorem \ref{thm:gap} is given in Section \ref{sec-bnd}.
 
 \begin{theorem}\label{thm:gap}
 Let  $\eps\in (0,L^{-1}/2)$
 with $L>0$ as in Assumption \ref{assumption:theta_bound}.
For each $\theta \in \Xi_{\eps} $,
let $u^{\theta}=\textsf{Greedy}\left( A, \theta 
\right)$ be defined by \eqref{eq:optimalcontrol_monotone_str}. Then there exists a constant $C> 0$, depending on $L$ and $\eps$,
such that $$
 |J^{\theta^\star}(u^{\theta^\star}) -J^{\theta^\star}(u^{\theta}) |\leq C\left(|\lambda^{\star}-\lambda|^2  +\|G^{\star}-G\|^2_{L^2([0,T])} \right), \quad \textrm{for all } \theta,\theta' \in  \Xi_{\eps}.  
$$
\end{theorem}
 
\begin{remark}
\label{rmk: stability_estimate}
The performance gap   in Theorem \ref{thm:gap} relies   on the Lipschitz stability   of the optimal control $u^\theta$ with respect to the parameter $\theta$, in particular,    the kernel function $G$ (see Proposition \ref{prop:u_lipschitz}). 
In \cite{basei2022logarithmic,guo2023reinforcement,szpruch2021exploration},  
Lipschitz stability of
optimal controls has been derived 
for   finite-dimensional parametric control problems, 
where the parameter varies in a compact subset of a finite-dimensional   space. 
The stability analysis for 
Theorem \ref{thm:gap}
is more technically involved, 
as $G$ takes value in the infinite-dimensional space $L^2([0,T]; \sR)$,
and the control problem \eqref{def:objective-str} 
is non-Markovian due to the kernel $G$.
For instance, 
one can no longer 
prove the Lipschitz continuity of $u^\theta$
by  simply arguing a continuous differentiability  of   $\theta\mapsto u^\theta$  as in 
\cite{basei2022logarithmic,szpruch2021exploration},
since a bounded subset of  $L^2([0,T]; \sR)$
may not be compact. 

To overcome these difficulties,
we exploit an explicit representation of $u^\theta$ given in \eqref{eq:aB},
and establish   uniform   bounds  of  
  $a$ and $B$ in \eqref{eq:optimalcontrol_monotone_str} with suitable norms 
over all $\theta\in \Xi_{\eps}$. 
These  a-priori bounds further allow for proving the Lipschitz stability of the non-Markovian controls $u^\theta$. Note that the explicit form of $a$ and $B$ is given in \eqref{eq:aB}, and its main ingredients are operators and stochastic processes. 
We point out that as the running cost in  
\eqref{def:objective-str} is not strongly concave with respect to the control variable, 
the control problem is an  indefinite linear-quadratic control problem (see   
\cite[Chapter 6]{yong1999stochastic} and references therein).
Hence the   condition \eqref{eps-pos-def}
is essential for the well-definedness and the stability  of $u^\theta$. 

\end{remark}

%

\paragraph{Phased-based   algorithm and its regret.}
The algorithm goes as follows.
{The input of the algorithm includes $\eps>0$ which satisfies Definition \ref{assum1:parameters}, and a deterministic exploration strategy $u^e\in L^2([0,T],\sR)$ as in Assumption \ref{assum:learning}.}
The algorithm starts with an initial exploration phase, where the agent  exercises 
 $u^e$ for $\mathfrak{m}^e_0$ episodes,
 and forms an  
  estimate  $ {\theta}^0$  of $\theta^\star$
according to 
\eqref{eq:lse_abbreviation}: 
\begin{equation}
\label{eq:lse_initial} 
\theta^0 = 
\textsf{LSE}\left(
(S^m,A^m)_{1\le m\le \mathfrak{m}^e_0}, \tau_{\mathfrak{m}^e_0},
 \pi_{\mathfrak{m}^e_0}
\right).
\end{equation}
{Here $\mathfrak{m}^e_0 \in \sN $  is a prescribed  number 
  such that $ {\theta}^0$ is guaranteed to be  in  $ \Xi_{\eps}$
  (with high probability). Note that such initial exploration phase has a constant weight in the regret bounds, hence it has no impact on the results established in Theorem \ref{thm:regret_general_rate} and Corollary \ref{cor-reg-bnd}. One can alternatively use estimators based on historical datasets in order to get $ {\theta}^0$ which is in $ \Xi_{\eps}$. Some references for these preliminary estimates are available in \citep{bouchaud-gefen,citeulike:13990970} and in Chapter 13 of \cite{bouchaud_bonart_donier_gould_2018}.
Additional approach for getting $ {\theta}^0$ from historical trading data, uses the offline learning approach which was developed in \cite{off-prop}. Indeed a continuous version of Theorem 2.10 therein allows us to get a candidate for $G$ which satisfies the properties in Definition \ref{assum1:parameters}. Finally, as shown in \cite[Example 2.7]{GSS}, 
any  bounded, non-increasing and convex $G$  satisfies \eqref{eps-pos-def} with $\eps=0$. 
If one assumes the true kernel 
$G^\star$
  satisfies these shape constraints (as suggested by the empirical studies    in \cite{bouchaud-gefen, Ob-Wan2005}), then one can   enforce these constraints in the estimated kernels by minimising \eqref{eq:lse_G} over shaped constrained functions. This approach could potentially avoid the initial exploration phase and improve the sample complexity. 
  These shape-constrained estimators have been analysed for   discrete-time propagator models in \cite{off-prop}, and extending the analysis therein to   continuous-time propagator models is left for future work. } 

   After this initial exploration, the algorithm  then
   operates in cycles, and each cycle consists of exploitation and exploration  phases.
The exploitation phase of the $k$-th cycle, $k\in \sN$,
contains  $\mathfrak{n}(k)$ 
consecutive episodes for some prescribed $\mathfrak{n}(k)\in \sN$. 
At each exploitation episode,
the agent 
 executes the optimal strategy 
\eqref{eq:optimalcontrol_monotone_str}
defined using   the current estimate $ {\theta}^{k-1}$
and   the    signal trajectory observed in this   episode.  
  During the exploration phase of the $k$-th cycle, 
 the agent  exercises  the exploration strategy $u^e$   for one episode, and 
constructs an updated  estimate   $ {\theta}^{k}$  
by \eqref{eq:lse_abbreviation}
 using data from    previous exploration episodes:
\begin{equation}
\label{eq:lse_k_cycle} 
\theta^k = 
\textsf{LSE}\left(
(S^m,A^m)_{m\in \cE_k}, \tau_{\mathfrak{m}^e_0+k},
\pi_{\mathfrak{m}^e_0+k}
\right),
\end{equation}
 where 
 $\cE_k=\{1,\ldots, \mathfrak{m}^e_0\}\cup \{ 
 \mathfrak{m}^e_0 +\sum_{i=1}^j  \mathfrak{n}(i)+j \mid 1\le  j \le  k\}
$ is the indices   of  all    exploration episodes up to the $k$-th cycle.
This   parameter $\theta^k$ will be used in the exploitation phase of   the $(k+1)$-th cycle. 
The algorithm is summarised as follows.

\begin{algorithm}[H]
\label{alg:phased}
\DontPrintSemicolon
\SetAlgoLined

  \KwInput{ $ \eps>0$,
$u^e\in L^2([0,T],\sR)$,
   $\mathfrak{m}^e_0 \in \sN $ 
and 
$\mathfrak{n} : \sN \to \sN$.
  }

 {
Execute   $u^e$ for $\mathfrak{m}^e_0$ episodes,
and set $\theta^0\in \Xi_\eps $ as in \eqref{eq:lse_initial}.
 }\;

 \For{$k = 1, 2,\ldots$}
 {
 {
 $\mathfrak{L}(k-1)=\mathfrak{m}^e_0 +\sum_{i=1}^{k-1}  \mathfrak{n}(i)+k-1$. \Comment*[r]{last  episode's index}}
  \For{$m= \mathfrak{L}(k-1)+1,  \ldots,   \mathfrak{L}(k-1)+    \mathfrak{n}(k)$}
  {
	  {Execute the greedy strategy  
	  $u^m
	 = \textsf{Greedy}\left( A^m, \theta^{k-1}
\right)$. }\;

}

 {Execute    $u^e$ for one episode, and set $\theta^k\in \Xi_\eps $ as in \eqref{eq:lse_k_cycle}.}\;

%
 }
 \caption{Phased-based  learning algorithm}
\end{algorithm}

\begin{remark}
\label{rmk:phased_based_algorithm}
Algorithm \ref{alg:phased} is an anytime algorithm,
as it does not restrict   the maximum number of learning episodes
(see the last paragraph of Section \ref{sec:episodic_learning_setting}).
It
distributes the exploration episodes  over the whole learning process according to the   schedulers   $\mathfrak{m}^e_0   $ and
  $\mathfrak{n}$, which are chosen  to optimise the regret order   for all   episodes.
This should be in contrast to the setting 
where the  algorithm termination time  
  is fixed and known by the agent.
  In this case, the agent can put all   exploration episodes   at the beginning, 
  whose   number  depends explicitly on the prescribed     maximal episode number (see  \cite[Chapter 6]{ lattimore2020bandit}). 


 Compared with the algorithm in  \cite{szpruch2021exploration}, 
 Algorithm \ref{alg:phased} introduces an initial exploration step, and 
 updates   $(\theta^k)_{k\ge 0 }$  
using   trajectories    generated by a  fixed  exploration strategy $u^e$.
This allows for applying Theorem   \ref{thm:estimate_error_high_probability}
to ensure that, with high probabilities,  $(\theta^k)_{k\ge 0}$ stay in $\Xi_{\eps}$ 
(without an explicit projection
as in \cite{szpruch2021exploration})
and converge to $\theta^\star$ as $k\to \infty$. 
\end{remark}

The following theorem chooses 
the learning schedulers $\mathfrak{m}^e_0 \in \sN $ and 
$\mathfrak{n}: \sN \to \sN$
 such that Algorithm  \ref{alg:phased} achieves sublinear regrets 
 in high probability.
These hyper-parameters are optimised depending on the   convergence rate   
of the regularised least-squares estimator  \eqref{eq:lse_abbreviation} in Theorem \ref{thm:estimate_error_high_probability}.
The proof of Theorem \ref{thm:regret_general_rate}
is given in Section \ref{sec-pf-rate}.

\begin{theorem}
\label{thm:regret_general_rate}
  Suppose that the parameters $(\tau_N,\pi_N)_{N\in \sN}$ 
   for \eqref{eq:lse_abbreviation} are chosen such that 
   $(\theta^N)_{N\in \sN}$ satisfies the following error estimate: 
    there exists $\tilde{C}>0$ and 
  $\kappa\in (0,1)$
  such that for all $\eta\in (0,1)$, it holds with probability at least $1-\eta$ that,
  for all $N\in \sN\cap [2,\infty)$,
\begin{equation}
\label{eq:lse_error_regret}
\textstyle
  |\lambda^N-\lambda^\star| 
  \le \tilde{C}\left( \frac{  \log(  {\eta}^{-1})+\log N    }{\sqrt{N}} \right),
  \quad 
\|G^N-G^\star\|_{L^2([0,T])}
\le \tilde{C}\left( \frac{  \log(  {\eta}^{-1})+\log N    }{\sqrt{N}} \right)^{\kappa}.
\end{equation}
Assume further that  Assumption \ref{assumption:theta_bound} holds and let $\eps\in (0,L^{-1}/2)$ as in Definition \ref{assum1:parameters}.
  Then there exists $C_0>0$ such that   for all $\eta\in (0,1)$ and $C\ge C_0$,
  if one sets $\mathfrak{m}^e_0 =\lceil  {C}( \log(  {\eta}^{-1})^2+1)\rceil$,
  and $\mathfrak{n}:\sN\to \sN$   such that 
  $\mathfrak{n}(k)=[ k^\kappa]$ for all $k\in \sN$,
   then  with probability at least $1-\eta$,
the regret of Algorithm  \ref{alg:phased} satisfies 
$$
R(N)\le C' \left(
  N^{\frac{ 1}{1+\kappa}  }  \left(  {  \log(  {\eta}^{-1})+\log N  }  \right)^{2\kappa}
  +\log(  {\eta}^{-1})^2 
\right), \quad 
 \textnormal{for all $N\in \sN\cap [2,\infty)$},
$$
where $C'>0$ is a constant independent of $\eta$ and $N$.
 
\end{theorem}

As Algorithm \ref{alg:phased} operates in cycles,  
for each $N\in \sN$, the regret of learning $R(N)$ up to $N$ episodes 
can be upper bounded by the accumulated regret at the end of the $K$-th cycle, 
with $K=\min\{k\in \sN\cup \{0\}\mid 
 \mathfrak{L}(k) =
 \mathfrak{m}^e_0 +\sum_{i=1}^k  \mathfrak{n}(i)+k\ge N\}$.


\begin{remark} \label{rem-conv-compar} 
If 
\eqref{eq:lse_error_regret} holds with $\kappa =1$,
then 
Theorem \ref{thm:regret_general_rate} recovers the
square-root regret bound   in \cite{szpruch2021exploration}
for linear-convex RL problems with finite-dimensional unknown parameters. 
However, in the present non-parametric setting, 
\eqref{eq:lse_error_regret} typically holds with $\kappa <1$
  (see Theorems \ref{thm:estimate_error_high_probability}),
  and this leads to a worse sublinear regret bound.
 Indeed, classical results   for Tikhonov regularization indicate that $\kappa=2/3$ is the best rate 
 one can expect,
  even for a smooth  kernel $G^\star$
  (see \cite[Section 3.2]{groetsch1984theory}).
  Employing other regularisation approaches to 
  improve the sample efficiency of the kernel estimation  
is left for future research. 
\end{remark}

By combining   Theorems \ref{thm:estimate_error_high_probability} and \ref{thm:regret_general_rate},
the following corollary optimises the regret bounds of Algorithm \ref{alg:phased}
depending on 
  the regularity of the true kernel $G^\star$.

\begin{corollary} \label{cor-reg-bnd} 
 Suppose that Assumptions \ref{assum:learning}, \ref{assum:concentration_M}
 and    \ref{assumption:theta_bound}
  hold. Let $\eps\in (0,L^{-1}/2)$ as in Definition \ref{assum1:parameters}.
  \begin{enumerate}[(1)]
  \item
      If Assumption \ref{assume:kernel_assumption}\ref{label:regular_kernel} holds,
then there exists $C_0>0$ such that  for all   $\eta\in (0,1)$ and $C\ge C_0$,
by setting $(\tau_N,\pi_N)_{N\in \sN} $ 
as   \eqref{eq:lse_parameter_regular} for \eqref{eq:lse_abbreviation},
 $\mathfrak{m}^e_0 =\lceil C( \log(  {\eta}^{-1})^2+1)\rceil$,
  and 
$\mathfrak{n}:\sN\to \sN$   such that 
  $\mathfrak{n}(k)=[ k^{\frac{1}{3}}]$ for all $k\in \sN$,
    then  with probability at least $1-\eta$,
the regret of Algorithm  \ref{alg:phased} satisfies for all $N\in \sN\cap [2,\infty)$, 
\begin{equation}\label{eq:regret_bound_regular}
R(N)\le C' \left(
  N^{\frac{ 3}{4 }  }  \left(  {  \log(  {\eta}^{-1})+\log N  }  \right)^{\frac{2}{3}}
  +\log(  {\eta}^{-1})^2 
\right).
\end{equation}
  \item
      If Assumption \ref{assume:kernel_assumption}\ref{label:singular_kernel} holds,
then there exists $C_0>0$ such that  for all $\eta\in (0,1)$ and    $C\ge C_0$, 
by setting $(\tau_N,\pi_N)_{N\in \sN} $ 
as   \eqref{eq:lse_parameter_singular} for \eqref{eq:lse_abbreviation},
 $\mathfrak{m}^e_0 =\lceil C( \log(  {\eta}^{-1})^2+1)\rceil$,
  and 
$\mathfrak{n}:\sN\to \sN$   such that 
  $\mathfrak{n}(k)=[ k^{\frac{1-2\alpha }{3-2\alpha}}]$ for all $k\in \sN$,
    then  with probability at least $1-\eta$,
the regret of Algorithm  \ref{alg:phased} satisfies for all $N\in \sN\cap [2,\infty)$, 
\begin{equation}\label{eq:regret_bound_singular}
R(N)\le C' \left(
  N^{\frac{ 3-2\alpha}{4-4\alpha }  }  \left(  {  \log(  {\eta}^{-1})+\log N  }  \right)^{\frac{2(1-2\alpha) }{3-2\alpha}}
  +\log(  {\eta}^{-1})^2 
\right).
\end{equation}
  \end{enumerate}
   The constant $C'>0$ appearing in 
  \eqref{eq:regret_bound_regular}
and 
\eqref{eq:regret_bound_singular}
is  independent of $\eta$ and $N$. 
 
\end{corollary}

{
\section{Numerical implementation} \label{sec-num}
In this section, we 
numerically examine the performance of the least-squares estimator \eqref{eq:lse_abbreviation} 
  developed in Section \ref{sec:lse}, which is the key    component of the 
  phased-based learning algorithm (Algorithm \ref{alg:phased}).
  We focus on 
estimating singular power law propagators, which are  extensively used by practitioners (see e.g. Chapter 13 of [10]). 

More precisely, let 
  $T>0$ and consider 
  $\lambda^\star >0$ to be an unknown temporary price impact coefficient,
and $G^\star \in L^2([0,T],\sR) $ to be an unknown transient impact kernel.
 For each $n\in \mathbb{N}$, consider the following price process (cf.~\eqref{def:S_m}):
\begin{align}
\label{eq:S_experiment}
S^n_t = A^n_t + M^n_t - \lambda^\star u^e_t -\int_0^t G^\star(t-s) u^e_s d s,
\quad
t\in [0,T],
\end{align}
where
 $u^e \in L^2([0,T], \mathbb{R})$ be a   trading strategy specified by the agent, 
$(A^n)_{n\in \mathbb{N}}$ are observed  signals, and  $(M^m)_{m\in \mathbb{N}}$ are unobserved  zero-mean   noises.  
The agent   estimates $(\lambda^\star, G^\star)$ based on the observed trajectories $(S^n,A^n)_{n\in \mathbb{N}}$ and the   trading strategy $u^e$.

Given sample trajectories   $(S^n,A^n)_{n=1}^N$, we estimate the parameter $\lambda^\star$  by 
$\lambda^N$ defined in 
\eqref{eq:lse_lambda},
and estimate the kernel  $G^\star$ by  
\begin{align}
\label{eq:lse_G_H}
\begin{split}
  G^N
  \coloneqq 
 & \argmin_{G\in V_N} \bigg(\frac{1}{N}\sum_{n=1}^N
 \int_0^T \left|
 S^n_t
-A^n_t
 +\lambda^N u^e_t
 + \int_0^t    u^e(t-s) G_s d s\right|^2 d t
 \\
&\quad +\tau_N \int_0^T (G_t- H_t)^2 dt\bigg).
\end{split}
\end{align}
The estimator   \eqref{eq:lse_G_H}
 extends \eqref{eq:lse_G} by allowing for  a (non-zero)  initial guess $H$ of $G^\star$
in the regularisation term.
It satisfies the convergence rates in Theorem \ref{thm:estimate_error_high_probability}
 provided  that  
  $G^\star -H$  satisfies Assumption \ref{assume:kernel_assumption} 
  (see e.g., \cite{Rastogi2020convergence}). 
  The estimator \eqref{eq:lse_G_H} can be   numerically
approximated by  
\begin{align}
\label{eq:lse_G_H_K}
\begin{split}
  G^N
  =& \argmin_{(G_k)_{k=0}^{K-1}} \bigg( \frac{1}{N}\sum_{n=1}^N
 \sum_{j=1}^{K}\left|
 S^n_{t_j}
-A^n_{t_j}
 +\lambda^N u^e_{t_j}
 +
 \sum_{k=0}^{j-1}
   u^e_{t_{j-k}} G_k \Delta t\right|^2 \Delta t
   \\
   &\quad 
+\tau_N  \sum_{k=0}^{K-1}
 (G_k-H_k)^2 \Delta t
 \bigg),
 \end{split}
\end{align}
where $\Delta t = {T}/{K}$ and 
  $t_i=iT/K$  for a  sufficiently large $K\in \sN$. Note that 
the minimiser of 
   \eqref{eq:lse_G_H_K} can be computed analytically by a first-order condition.

For our numerical experiments, we fix 
$\lambda^\star =0.5 $ and $G^\star (t) =t^{-\alpha}$ for some $\alpha \in (0,1/2)$.
For each $N\in \sN$, 
we  generate  observed trajectories $(S^m-A^m)_{m=1}^N $ 
according to \eqref{eq:S_experiment}
with  $T=1$,  $u^e\equiv 1$ and  
 $M^n_t=0.5 (B^n_t+\iota_0)$, $t\in [0,T]$,
where $(B^n)_{n=1}^N$ are independent Brownian motions and $\iota_0$ is an independent standard normal random variable. Using these trajectories, we evaluate 
$\lambda^N$ as in  \eqref{eq:lse_lambda}
and   $G^N$ as in \eqref{eq:lse_G_H_K} with $H\equiv 1$, $K=10^3$
and $\tau_N =N^{-{2}/{(3-2\alpha)}}$ as suggested by \eqref{eq:lse_parameter_singular}.
To estimate statistical properties of the estimators, we
carry out the experiments   for 10 independent runs, where among different executions, the observed state trajectories are simulated based on independent noises.
In the sequel, we only report  the performance of    $G^N$, as $G^\star$ is more challenging to estimate than $\lambda^\star$.

 Figure \ref{fig:power_law} illustrates the effectiveness of 
$G^N$ for estimating the power law kernel  $G^\star(t)=t^{-\alpha}$ with  different   $\alpha\in \{0.1,0.4\}$ and different sample sizes $N$.
Figure \ref{fig:true_kernel} plots the true kernels, showing that the degree of singularity at $t=0$ increases for larger 
$\alpha$. 
Figures \ref{fig:estimated_N2_10} and \ref{fig:estimated_N2_16} display the estimated kernels, where the solid lines represent the mean and the shaded areas indicate the extremes over 10 repeated experiments.
One can see that the estimator $G^N$ successfully recovers the overall behavior of the true kernels, even with a small sample size. The regularisation $\tau_N$ prevents oscillation in the estimated kernel, in contrast to an unregularised estimator as discussed in Remark \ref{rem-t}. However, the estimator is not able to accurately capture the singularities of the true kernels without a sufficient number of samples.
By increasing the sample sizes, these singularities can be better captured by the estimator.

 \begin{figure}[H]
\centering
\begin{subfigure} {0.45\textwidth} 
\centering
 \includegraphics[trim=18 5 30 25, clip, width=\textwidth]{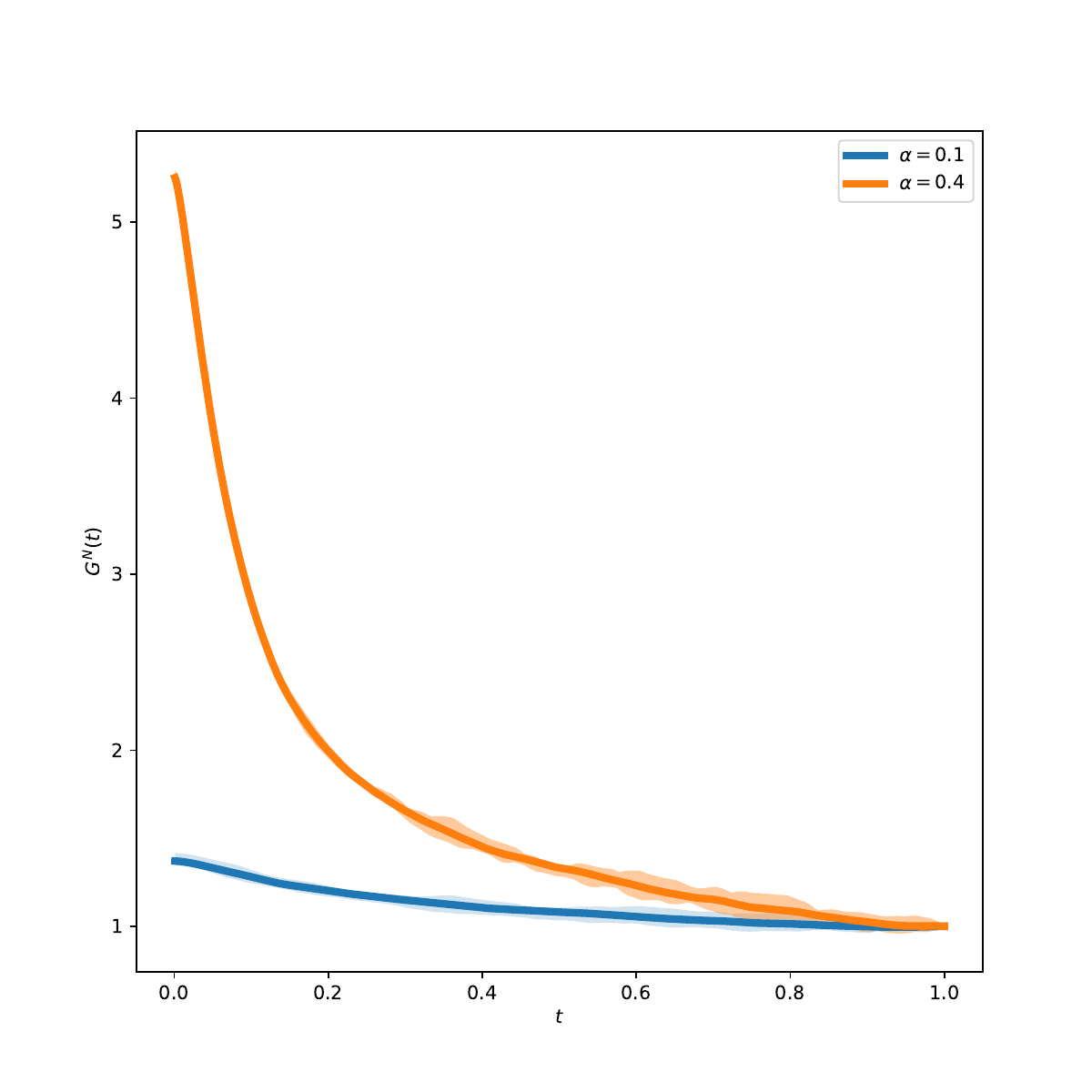}
\caption{Estimated kernels with $N=1024$} 
\label{fig:estimated_N2_10}
\end{subfigure}
\hfill
\begin{subfigure}{0.45\textwidth} 
\centering
\includegraphics[trim=20 5 30 25, clip,  width=\textwidth]
 {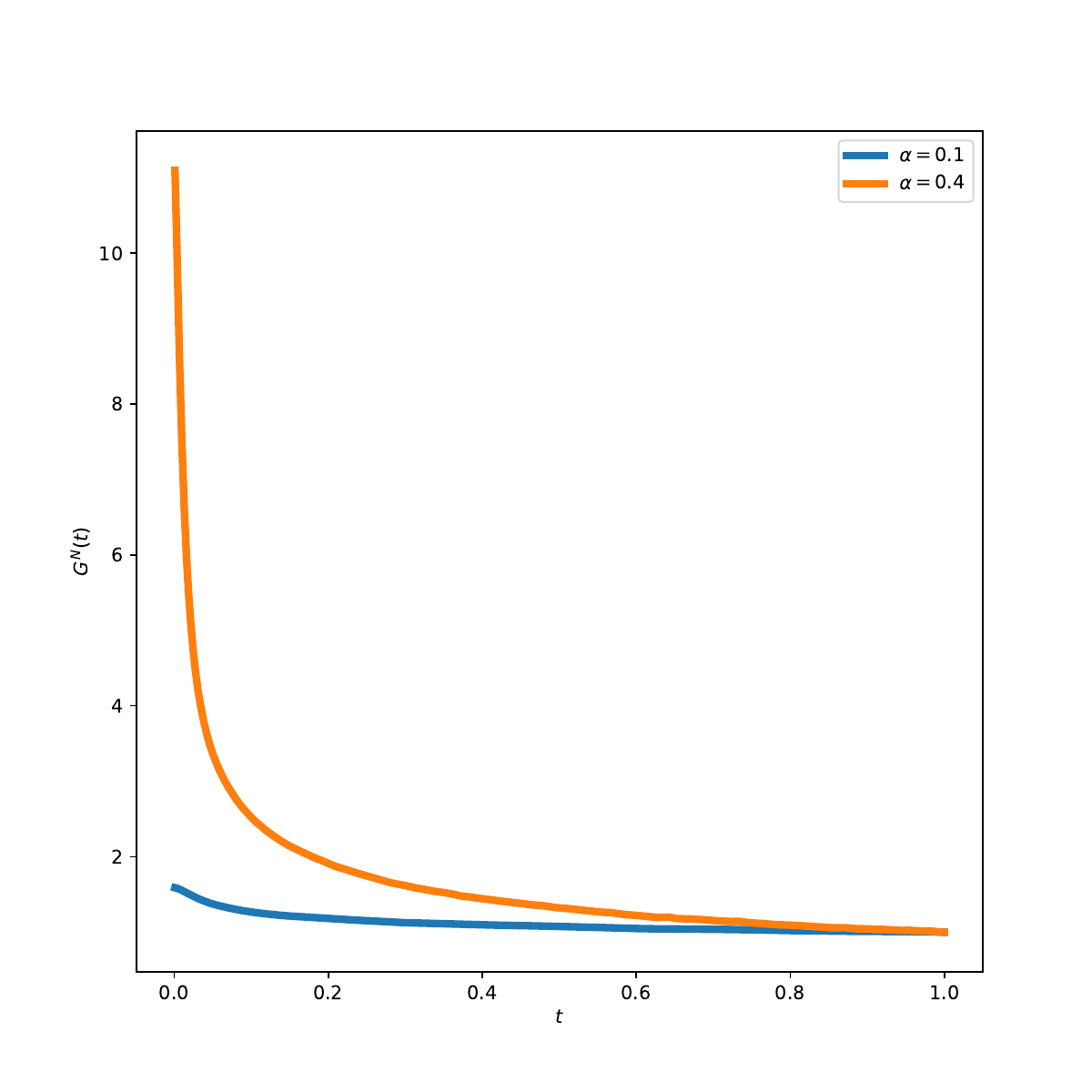}
\caption{Estimated kernels with $N=65526$}  
\label{fig:estimated_N2_16}
\end{subfigure}
\vspace{10pt}

\begin{subfigure}{0.45\textwidth} 
\centering
\includegraphics[trim=18 5 30 25, clip, width=\textwidth]{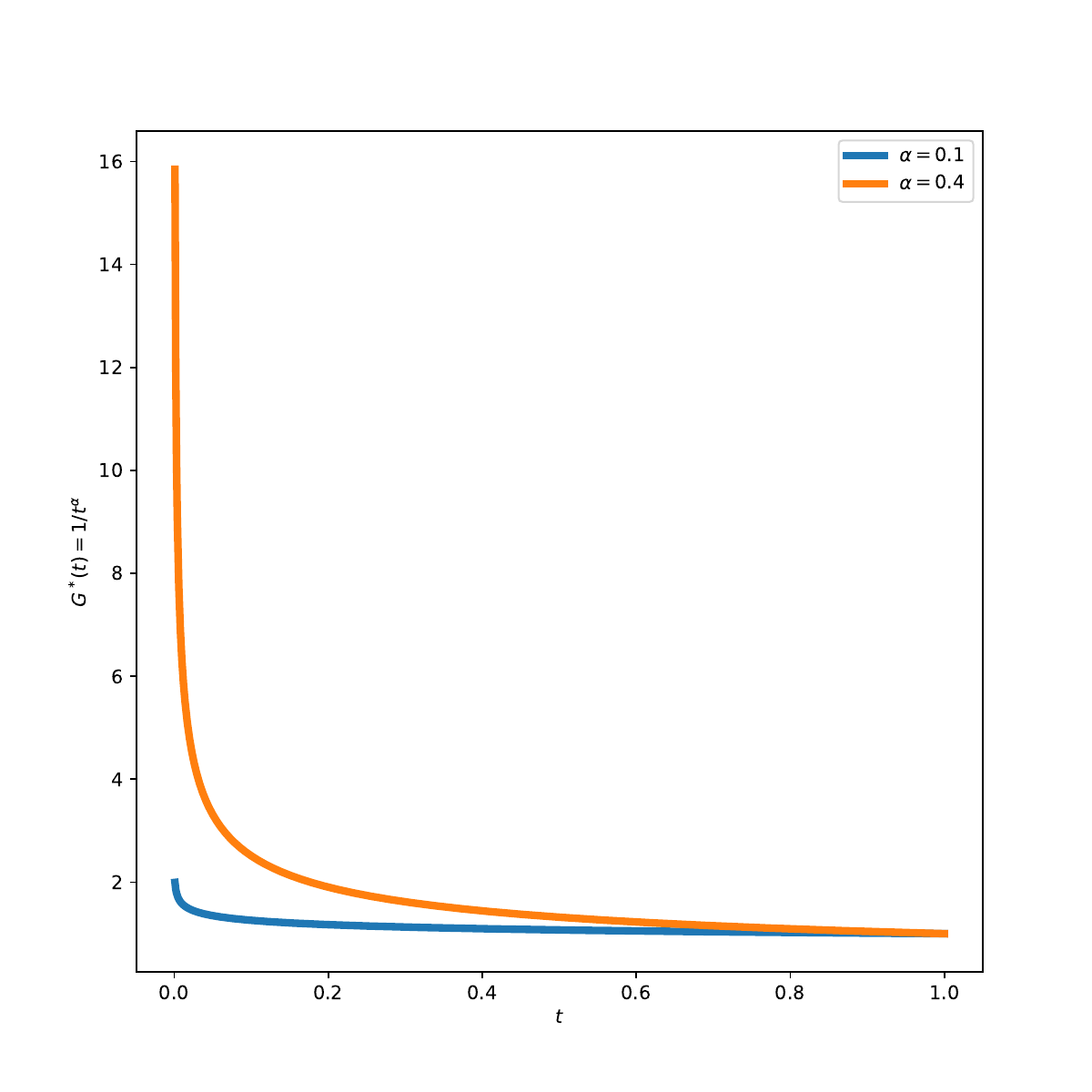}
\caption{True kernels} 
\label{fig:true_kernel}
\end{subfigure} 
\caption{Comparison between the true power law kernels $G^\star(t)=t^{-\alpha}$, with $\alpha =0.1$ (in blue) and $\alpha =0.4$ (in orange),  and the estimated kernels with different sample sizes $N$.}
\label{fig:power_law} 
\end{figure}

We further demonstrate the decay of the relative error $\|G^N-G^\star\|_{L^2([0,T])}/ \|G^\star\|_{L^2([0,T])}$
in Figure \ref{fig:error} for sample sizes  $N=2^n$ with $n \in \{10,11,\ldots, 16\}$. 
The estimator makes larger errors for 
$\alpha=0.4$
 compared to 
 $\alpha=0.1$,
 due to the more severe singularity of the true kernel. 
 A linear regression on the logarithms of relative errors and sample sizes reveals that the convergence rate of 
 $(G^N)_{N\in \sN}$   is of the order 
$\cO(N^{-0.22})$ for   $\alpha=0.1$
and 
$\cO(N^{-0.2})$ for   $\alpha=0.4$,
 which   are better than the theoretical upper bounds presented in Theorem \ref{thm:estimate_error_high_probability}.
 Note that this improved rate does not contradict the theorem, as the convergence rates in Theorem \ref{thm:estimate_error_high_probability} pertain to the worst-case scenario over all kernels 
$G^\star$  satisfying Assumption \ref{assume:kernel_assumption} and all realizations of the noises 
$M$ satisfying Assumption \ref{assum:concentration_M}. As already pointed out in \cite[Section 3.2]{engl1996regularization}, for particular realised sample trajectories, the error of the estimator \eqref{eq:lse_G} (or \eqref{eq:lse_G_H_K}) may be smaller than the bounds provided in Theorem \ref{thm:estimate_error_high_probability}.

\begin{figure}[H]
\centering
\includegraphics[trim=20 5 30 25, clip,  width=0.8\textwidth]{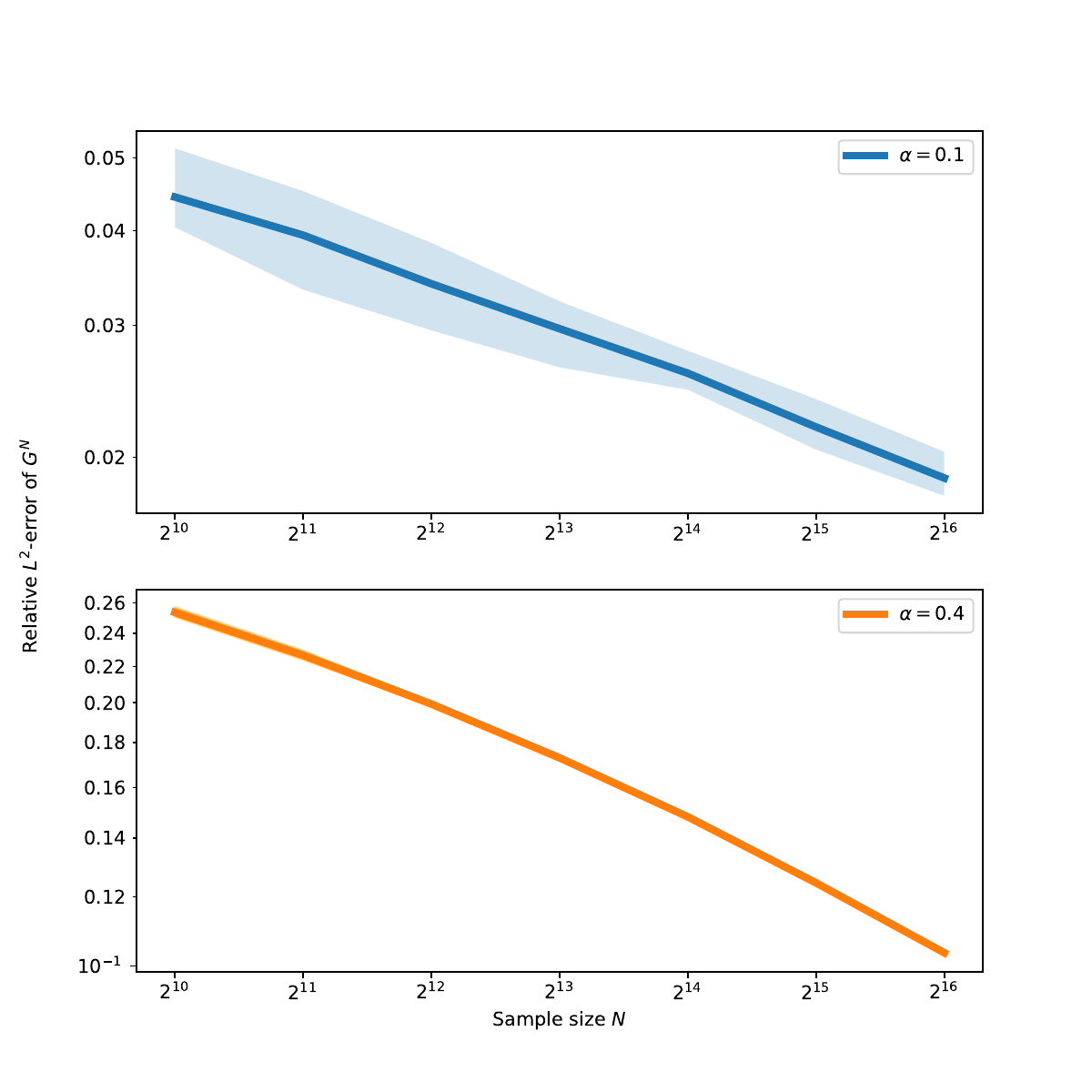}     
\caption{Mean relative errors of $G^N$ for different sample sizes plotted in solid lines and the intervals containing the errors denoted by the lighter regions (the plot is in a log-log  scale). The true power law kernels $G^\star(t)=t^{-\alpha}$, with $\alpha =0.1$ (upper panel) and $\alpha =0.4$ (lower panel).}
\label{fig:error}
\end{figure}

It is worth noting that in this experiment, we focus on the challenging task of estimating singular kernels, for which the proposed estimators are expected to produce larger estimation errors compared to more regular kernels. However, as shown in Theorem \ref{thm:gap}, the error of the greedy policy for an estimated model depends quadratically on the kernel estimation error. Therefore, a rough estimate of the kernel 
$G$ is often sufficient to design a nearly optimal trading strategy. In particular, as shown in Figure \ref{fig:error}, 
even for the most singular kernel with $\alpha=0.4$, the relative error of the resulting greedy policy is expected to be around 6\% with a sample size of 
$N=1024$.

} 

\section{Analytic solution to the control problem} \label{sec-results} 
In this section, we recall the explicit form on the optimiser of \eqref{def:objective-str} from \cite{AJ-N-2022}. Before stating this result we introduce some essential definitions of function spaces, integral operators and stochastic processes.   
\subsection{Function spaces, integral operators} 
Let $T>0$. We denote by $\langle \cdot, \cdot \rangle_{L^2}$ the inner product on $L^2([0,T], \R^2)$, that is 
\be \label{in-prod} 
\langle f, g\rangle_{L^2} = \int_0^T f(s)^{\top} g(s) ds, \quad f,g\in L^2\left([0,T],\mathbb R^2\right). 
\ee
We define $L^2\left([0,T]^2,\mathbb R^{2 \times 2}\right)$ to be the space of measurable kernels $\Sigma:[0,T]^2 \to \R^2$ such that 
\begin{align*}
\int_0^T \int_0^T |\Sigma(t,s)|^2 dt ds < \infty.
\end{align*}
The notation $|\cdot|$ stands for a matrix norm, and in particular we have, 
$$
\int_0^T \int_0^T |\Sigma_{i,j}(t,s)|^2 dt ds < \infty, \quad \textrm{for all } i,j =1,2. 
$$
For any $\Sigma, \Lambda \in  L^2\left([0,T]^2,\mathbb R^{2\times 2}\right)$ we define the $\star$-product as follows,
\begin{align}\label{eq:starproduct}
(\Sigma \star \Lambda)(s,u) = \int_0^T \Sigma(s,z) \Lambda(z,u)dz, \quad  (s,u) \in [0,T]^2,
\end{align}
which is a well-defined kernel in $L^2\left([0,T]^2,\mathbb R^{2\times 2}\right)$ due to Cauchy-Schwarz inequality.  For any  kernel $\Sigma \in L^2\left([0,T]^2,\mathbb R^{2\times 2}\right)$, we denote by {$\boldsymbol \Sigma$} the integral operator   induced by the kernel $\Sigma$ that is 
\begin{align} \label{sigma-def} 
({\boldsymbol \Sigma} g)(s)=\int_0^T \Sigma(s,u) g(u)du,\quad g \in L^2\left([0,T],\mathbb R^2\right).
\end{align}
$\boldsymbol \Sigma$ is a linear bounded operator from  $L^2\left([0,T],\mathbb R^2 \right)$ into itself. 
For $\boldsymbol{\Sigma}$ and $\boldsymbol{\Lambda}$ that are two integral operators induced by the kernels $\Sigma$ and $\Lambda$  in $L^2\left([0,T]^2,\mathbb R^{2\times 2}\right)$, we denote by $\boldsymbol{\Sigma}\boldsymbol{\Lambda}$ the integral operator induced by the kernel $\Sigma\star \Lambda$.

We denote by $\Sigma^*$ the adjoint kernel of $\Sigma$ for $\langle \cdot, \cdot \rangle_{L^2}$, that is 
\begin{align} \label{adj-def} 
\Sigma^*(s,u) &= \; \Sigma(u,s)^\top, \quad  (s,u) \in [0,T]^2,
\end{align}
and by $\boldsymbol{\Sigma}^*$ the corresponding adjoint integral operator. 

We recall that an operator $\boldsymbol{\Sigma}$ as above is said to be non-negative definite if $\langle \boldsymbol{\Sigma}f,f\rangle_{L^{2}} \geq0$ for all $f\in L^2\left([0,T],\mathbb R^2\right)$. It is said to be positive definite if $\langle\boldsymbol{\Sigma}f,f\rangle_{L^{2}}> 0$ for all $f\in L^2\left([0,T],\mathbb R^2\right)$ not identically zero. 

\subsection{Essential operators and processes} 
\paragraph{The $\Gamma^{-1}_t$ operator:}
We define  
\be \label{h-tilde} 
   \tilde G(t,s) = \left(2 \varrho  + G(t-s)\right) \mathds{1}_{\{s<t\}}, \quad 0\leq s, t \leq T. 
\ee
and $\boldsymbol{\tilde G}_t$ as the operator induced by the kernel 
\be \label{tilde-G} 
\tilde G_t(s,u) = \tilde G(s,u)\mathds 1_{\{u \geq t\}}. 
\ee
 We introduce 
\begin{align}\label{eq:schur}
    \boldsymbol{D}_t := 2\lambda \id +  (\boldsymbol{\tilde G}_t + \boldsymbol{\tilde G}^*_t) + 2\phi  \boldsymbol{1}^*_t \boldsymbol{1}_t,
\end{align}
where $\id$ is the idendity operator, i.e.~$(\id f)(t)=f(t)$, $\boldsymbol{1}_t$ is the integral operator induced by the kernel 
\be \label{op-one} 
\mathds 1_t(u,s) := \mathds 1_{\{u \geq  s\}} \mathds 1_{\{s \geq t\}} .
\ee
In Lemma 3.1 of \cite{AJ-N-2022} it was proved that $\boldsymbol{D}_t$ is invertible if $\lambda>0$ and $\varrho, \phi \geq 0$. This will be necessary for upcoming definitions.  We define the operator $\boldsymbol{\Gamma}_t^{-1}$ by 
 \be \label{op-gam-inv} 
\begin{aligned}
\boldsymbol {\Gamma}_t^{-1} =  \left(\begin{matrix}
 \boldsymbol{D}_t^{-1}  &-2 \phi \boldsymbol{D}_t^{-1} \boldsymbol{1}^*_t \\
-2\phi \boldsymbol{1}_t \boldsymbol{D}_t^{-1} & - {2\phi} \id + 4\phi^2  \boldsymbol{1}_t \boldsymbol{D}_t^{-1} \boldsymbol{1}^*_t
\end{matrix}\right).
\end{aligned}
\ee 
We recall that $\boldsymbol \Gamma^{-1}$ is associated with a solution to an operator Riccati equation (see Lemmas 6.1 and 6.2 in \cite{AJ-N-2022}).

\paragraph{The process $\Theta$:} 
For convenience we introduce the following notation, 
\be \label{ind-def} 
\mathds{1}_t(s)= \mathds{1}_{\{ s \geq t\}}. 
\ee
 For $A$ as in \eqref{eq:P_decomposition}, we define $\Theta= \{\Theta_{t}(s) : t\in [0,s],\, s\in [0,T]\}$ as follows,  
\be \label{th-sol} 
\Theta_t(s)=- \left(\boldsymbol {\Gamma}_t^{-1}  \mathds{1}_t \E\left[A_{\cdot}- A_T \Mid \mathcal F_t\right] e_1 \right)(s),
\quad \textnormal{with $e_1=(1,0)^\top$}.
\ee
Note that $\Theta$ solves the $L^2$-valued BSDE (see Proposition 6.3 in \cite{AJ-N-2022}). 
 
 \paragraph{The optimal control $u^\star$:} 
Proposition 4.5 in \cite{AJ-N-2022} states that the optimiser of \eqref{def:objective-str}, $u^\star$ is a solution to the following equation,  
\be \label{volt-u} 
 u^\star_t = a_t + \int_0^t B(t,s)u_s^\star ds,
 \ee
where the process $\{a_t\}_{t\in [0,T]}$ and the kernel $B$ which are given by
	\begin{equation} \label{eq:aB}
	\begin{aligned} 
a_t &= \frac 1{2\lambda} \left(\E\left[ A_t - A_T \mid \mathcal F_t\right]  + 2\varrho q + \langle \Theta_t , K_t \rangle_{L^2} +  \langle \boldsymbol {\Gamma}_t^{-1} K_t  , \mathds 1_t  (- 2\varrho q, q )^\top \rangle_{L^2} \right), \\
B(t,s) &=  \mathds 1_{\{s<t\}}\frac 1{2\lambda} \left( \langle \boldsymbol {\Gamma}_t^{-1} K_t  , \mathds 1_t  ( \tilde G(\cdot,s), -1 )^\top \rangle_{L^2}   - \tilde G(t,s) \right). 
	\end{aligned}  
	\end{equation}
Here the function $K_t$ is defined as follows 
\be \label{K-t}
 K_t(s) = (\tilde G(t,s), - \mathds{1}_{\{s\leq t\}})^\top.
 \ee

 \section{Bound on the performance gap} \label{sec-bnd} 
In this section we prove Theorem \ref{thm:gap}. Recall that the parameter space  $\Xi_{\eps} $ was defined in \eqref{eps-pos-def} and that  $u^{\theta^\star}$ is the maximiser of $J^{\theta^\star}$ in \eqref{def:objective-str}. In the following proposition we derive an upper bound on $J^{\theta^\star}(u^{\theta^\star}) -J^{\theta^\star}(u)$ for any admissible strategy $u$. 

For any admissible strategy $u$ as in \eqref{def:admissset}  we define 
\be \label{h-norm} 
\|u\|_{\mathcal H^2} = \left(\E\left[ \int_0^Tu_s^2 ds \right] \right)^{1/2}.
\ee
\begin{proposition} \label{lem-perfor} Let $\theta^\star\in \Xi_{\eps} $, then there exits a constant $C>0$ such that 
$$
 0\leq  J^{\theta^\star}(u^{\theta^\star}) -J^{\theta^\star}(u)   \leq    C  \|u^{\theta^\star}- u\|_{\mathcal H^2}^2, \quad \textrm{for all } u \in \mathcal A. 
$$
\end{proposition}
 We start with characterising the optimal strategy by using a variational approach. Note that for any $u \in \mathcal{A}$ the map $u \mapsto J(u)$ in~\eqref{def:objective-str} is strictly concave, which can be easily shown by repeating the same lines as in the proof of Lemma 9.1 of \cite{N-V-2021} . Therefore, it admits a unique maximizer characterized by the critical point at which the G\^ateaux derivative
\begin{equation} \label{def:gateaux}
    \langle (J^{\theta^\star}(u)', \alpha \rangle \triangleq \lim_{\varepsilon \rightarrow 0} \frac{J^{\theta^\star}(u + \varepsilon \alpha) - J^{\theta^\star} (u)}{\varepsilon}
\end{equation}
of the functional $J^{\theta^\star}$ vanishes for any direction $\alpha = (\alpha_t)_{0 \leq t \leq T} \in \mathcal{A}$; see, e.g.,~\cite{EkelTem:99}. The G\^ateaux derivative in~\eqref{def:gateaux} can be derived along the same line as in \cite{NeumanVoss:20}. The result is given in the following lemma. 
\begin{lemma}  \label{lem:gateaux}
Let $  J^{\theta^\star}$ as in \eqref{def:objective-str}. Let $u \in \mathcal{A}$, then we have
\begin{equation}
\begin{aligned} \label{eq:gateaux}
   \langle J^{\theta^\star}(u)', \alpha \rangle  &= \mathbb{E} \Bigg[\int_0^T \alpha_s \left( P_s -   Z_s^{\theta^\star,u} -  \int_s^T G^\star(t-s)  u_t dt - 2 \lambda^{\star} u_s \right. \Bigg. \\
    & \hspace{75pt} \Bigg. \Bigg. + 2 \phi \int_s^T Q^u_t dt + 2 \varrho Q^u_T - P_T \Bigg) ds \Bigg],
\end{aligned}
\end{equation}
for any $\alpha \in \mathcal{A}$.
\end{lemma} 

Since we are maximizing the strictly concave functional $u \mapsto J^{\theta^\star}(u)$ over $\mathcal A$, a necessary and sufficient condition for the optimality of $u^{\theta^\star} \in \mathcal A$ is given by
\begin{equation} \label{gat-0}
     \langle J^{\theta^\star}( u^{\theta^\star})', \alpha \rangle = 0 \quad  \text{for all } \alpha \in \mathcal A;
\end{equation}
see e.g.,~\cite{EkelTem:99}.  Now we are ready to prove Proposition \ref{lem-perfor}. 

\begin{proof} [Proof of Proposition \ref{lem-perfor}]
From \eqref{def:S_str} and Fubini's theorem we get that  
 \begin{equation} \label{y-u-dif} 
\begin{aligned}
&\int_0^T \big( Z_t^{\theta^\star,u} u_t - Z^{\theta^\star, u^{\theta^\star}}_t u^{\theta^\star}_t  \big)dt \\
&= \int_0^T (u_t-u_t^{\theta^\star})(Z_t^{\theta^\star,u}-Z^{\theta^\star, u^{\theta^\star}}_t ) dt \\ 
&\quad + \int_0^T (u_t-u_t^{\theta^\star}) Z^{\theta^{\star},u^{\theta^\star}}_t dt + \int_0^Tu_t^{\theta^\star}(Z^{\theta^\star, u}_t-Z^{\theta^\star, u^{\theta^\star}}_t )dt. 
\\
&= \int_0^T (u_t-u_t^{\theta^\star})(Z^{\theta^\star, u}_t-Z^{{\theta^\star, u^{\theta^\star}}}_t ) dt +\int_0^T (u_t-u_t^{\theta^\star}) Z^{{u^{\theta^\star, \theta^\star}}}_t dt \\
&\quad  +  \int_0^T (u_s-u_s^{\theta^\star})\int_s^T  G^\star(t-s) u^{\theta^\star}_t dt ds. 
\end{aligned}
\end{equation}
 
From \eqref{def:S_str} and application of Cauchy-Schwarz inequality twice we get  
 \be  \label{y-diff} 
\begin{aligned}
 &\mathbb{E} \left [ \Big| \int_0^T \big(u_t - u^{\theta^\star}_t  \big) (Z^{\theta^\star, u}_t-Z^{\theta^\star, u^{\theta^\star}}_t )dt \Big| \right] \\
  &\leq \left(  \mathbb{E} \left [  \int_0^T \big( Z^{\theta^\star, u}_t-Z^{\theta^\star, u^{\theta^\star}}_t  \big)^2dt  \right] \right)^{1/2}  \left(\mathbb{E} \left [ \int_0^T \big(u_t - u^{\theta^\star}_t  \big)^2dt \right] \right)^{1/2} \\
 &\leq   \|u^{\theta^\star}- u\|_{\mathcal H^2} \left(\mathbb{E} \left [  \int_0^T \left( \int_0^t G^\star(t-s)  \big( u_s - u^{\theta^\star}_s \big)ds   \right)^2    dt\right] \right)^{1/2} \\
  &\leq C\left(  \|u^{\theta^\star}- u\|_{\mathcal H^2} \right)^2 \left(   \int_0^T(G^\star(t-s))^2  ds   \right)^{1/2}    \\
  &\leq  C \|u^{\theta^\star}- u\|_{\mathcal H^2}^2, 
  \end{aligned}
 \ee
where we used Definition \ref{assum1:parameters} in the last inequality.

Using \eqref{def:Q} and Jensen's inequality we get that  
\be \label{x-diff} 
\begin{aligned}
 \mathbb{E} \big[\big(Q_T^u-Q_T^{u^{\theta^\star}} \big)^2 \big] & \leq C \mathbb{E} \left[ \int_{0}^T (u_t^{\theta^\star}- u_t)^2dt \right] \\
  & \leq C\|u^{\theta^\star}- u\|_{\mathcal H^2}^2.
 \end{aligned}
 \ee
 Similarly we have 
 \be \label{q-diff} 
\E \left[ \int_0^T\big(Q_t^u-Q_t^{u^{\theta^\star}} \big)^2 dt \right] \leq \|u^{\theta^\star}- u\|_{\mathcal H^2}^2.
 \ee
From \eqref{def:objective-str} and \eqref{y-u-dif} we therefore get for any $u \in \mathcal{A}$ that 
\begin{equation} 
\begin{aligned}
    & J^{\theta^\star}(u) -J^{\theta^\star}(u^{\theta^\star}) \\
    & =   \mathbb{E} \Bigg[\int_0^T P_t (u_t-u_t^{\theta^\star}) dt - \int_0^T \big( Z_t^{\theta^\star, u} u_t - Z^{\theta^\star, u^{\theta^\star}}_t u^{\theta^\star}_t  \big)dt \\
    &\qquad  - \lambda \int_0^T\big( u_t^2 - (u^{\theta^\star}_t)^2 \big) dt   -\phi \int_0^T\big( (Q_t^u)^2-(Q_t^{u^{\theta^\star}})^2 \big) dt  +(Q_T^u-Q_T^{u^{\theta^\star}})P_T \\
    &\qquad  - \varrho \big((Q_T^u)^2-(Q_T^{u^{\theta^\star}})^2 \big)\Bigg] \\
         & =   \mathbb{E} \Bigg[  - \int_0^T \big( Z_t^{\theta^\star, u} - Z^{\theta^\star,u^{\theta^\star}}_t \big)\big(u_t - u^{\theta^\star}_t  \big)dt  - \lambda \int_0^T\big( u_t^{\theta^\star}- u_t\big)^2 dt \\
     &\qquad  -\phi \int_0^T\big(Q_t^u-Q_t^{u^{\theta^\star}} \big)^2 dt - \varrho \big(Q_T^u-Q_T^{u^{\theta^\star}} \big)^2\Bigg] \\
        & \qquad +  \mathbb{E} \Bigg[ \int_0^T P_t (u_t-u_t^{\theta^\star}) dt 
        -\int_0^T (u_t-u_t^{\theta^\star})Z^{{\theta^\star,u^{\theta^\star}}}_t  dt \\
        &\qquad \qquad -   \int_0^T (u_s-u_s^{\theta^\star})\int_s^TG^\star(t-s)  u^{\theta^\star}_t dt ds 
 -2 \lambda \int_0^Tu_t^{\theta^\star} \big( u_t - u_t^{\theta^\star}  \big) dt \\
&\qquad \qquad + (Q_T^{u^{\theta^\star}}- Q_T^u) P_T  +2\phi \int_0^TQ_t^{u^{\theta^\star}} \big( Q_t^u-Q_t^{u^{\theta^\star}} \big) dt \\
      &\qquad \qquad+ 2\varrho Q_T^{u^{\theta^\star}} \big(Q_T^u-Q_T^{u^{\theta^\star}} \big)\Bigg].
\end{aligned}
\end{equation}
Since $ u_t - u_t^{\theta^\star} \in \mathcal{A}$ it follows from \eqref{def:Q}, \eqref{eq:gateaux} and \eqref{gat-0} that 
\begin{equation} 
\begin{aligned}
    & J^{\theta^\star}(u) -J^{\theta^\star}(u^{\theta^\star}) \\
           & =   \mathbb{E} \Bigg[  - \int_0^T \big( Z_t^{\theta^\star, u} - Z^{\theta^\star, u^{\theta^\star}}_t \big)\big(u_t - u^{\theta^\star}_t  \big)dt  - \lambda \int_0^T\big( u_t^{\theta^\star}- u_t\big)^2 dt \\
     &\qquad   -\phi \int_0^T\big(Q_t^u-Q_t^{u^{\theta^\star}} \big)^2 dt - \varrho \big(Q_T^u-Q_T^{u^{\theta^\star}} \big)^2\Bigg]. 
 \end{aligned}
\end{equation}
Together with \eqref{y-diff} and \eqref{x-diff} we get that 
\begin{equation*} 
\begin{aligned}
  | J^{\theta^\star}(u^{\theta^\star}) -J^{\theta^\star}(u)| & \leq  C \|u^{\theta^\star}- u\|^2_{\mathcal H^2},
  \end{aligned}
  \end{equation*}
which completes the proof.   
\end{proof} 
The following proposition proves the  stability of $u^{\theta^\star}$ with respect to the parameter $\theta^\star$.  The proof of Proposition \ref{prop:u_lipschitz} is postponed to Section \ref{sec-pf-prop-u}.
 \begin{proposition}
\label{prop:u_lipschitz}
For each $\theta \in \Xi_{\eps} $ let $u^{\theta}$ be defined as in \eqref{eq:optimalcontrol_monotone_str}. Then, there exists a constant $C> 0$
such that 
$$
\|u^{\theta}-u^{\theta'}\|_{\mathcal H^2}
\le 
C\left(
|\lambda-\lambda'|  +\|G-G'\|_{L^2([0,T])}\right), \quad \textrm{for all } \theta,\theta' \in  \Xi_{\eps}. 
$$
 \end{proposition}

Now we are ready to prove Theorem \ref{thm:gap}. 
\begin{proof}[Proof of Theorem \ref{thm:gap}]
The proof of Theorem \ref{thm:gap} follows immediately from Propositions \ref{lem-perfor} and \ref{prop:u_lipschitz}. 
\end{proof}

\section{Analysis of  regularised least-squares estimators}  \label{sec-est} 

This section 
 quantifies   the convergence rate of  $(\theta^N)_{N\in \sN}$ 
 in \eqref{eq:lse_abbreviation} 
   in high probability,
   and hence proves Theorem \ref{thm:estimate_error_high_probability}.
To simplify the notation, for each $f\in H^1([0,T],\sR)$,   we denote by $\dot{f}$ the derivative of $f$.

 The following lemma proves the injectivity of    $\boldsymbol{u}^e   $
in \eqref{eq:operator_u},
and further characterises the range of the adjoint operator
$(\boldsymbol{u}^e )^*$. 

 \begin{lemma}
 \label{lemma:u_operator_property}
 Suppose   Assumption \ref{assum:learning} 
 holds. Then 
  $\boldsymbol{u}^e:   L^2([0,T], \sR)\to   L^2([0,T],\sR) $ is injective and compact.
Moreover,   the range of the adjoint operator $(\boldsymbol{u}^e)^*$ is given by 
  $\mathscr{R}((\boldsymbol{u}^e)^*) =\{f\in H^1([0,T],\sR)\mid f(T)=0\} $.
  \end{lemma}
 
 \begin{proof}
 The compactness of  $\boldsymbol{u}^e   $ follows from the fact that 
 it is a  Hilbert-Schmidt integral operator.
  Next we show that $\boldsymbol{u}^e$ is injective.
   Let $f\in  L^2([0,T],\sR)$ such that 
   $(\boldsymbol{u}^e f)(t)= 0$ for a.e.~$t\in [0,T]$.
 By 
\eqref{eq:operator_u} and
  the Leibniz integral rule,
for   {a.e.~$t\in [0,T]$},
\begin{align*}
0&=\frac{d }{d t}\left(\int_0^t  u^e({t-s}) f(s)  d s\right)
=u^e(0)f(t)+\int_0^t  \dot{u}^e({t-s}) f(s)  d s
\\
&=\left((u^e(0)\id +\boldsymbol{\dot{u}}^e)f\right)(t), 
\end{align*}
where   $\boldsymbol{\dot{u}}^e$ is   the Volterra operator 
induced by the square-integrable   kernel
   $[0,T]^2\ni (t,s)\mapsto  \dot{u}^e({t-s})  \mathds  1_{\{s<t\}}\in \sR$.
This along with 
  \cite[Chapter 9, Corollary 3.16]{gripenberg1990volterra}
   and    $u^e(0)\not =0$
   yields 
 $f =   0$. This proves  
 the injectivity of  $\boldsymbol{u}^e$.

It remains to characterise the range of $(\boldsymbol{u}^e)^*$.
By \eqref{eq:operator_u} and Fubini's theorem,
 the adjoint 
of $\boldsymbol{u}^e$ satisfies 
for all $f\in L^2([0,T],\sR)$ and a.e.~$t\in [0,T]$, 
$\left((\boldsymbol{u}^e)^* f\right)(t) 
=\int_t^T u^e({s-t}) f(s) d s$,
and hence $\left((\boldsymbol{u}^e)^* f\right)(T)=0$.
By   Assumption \ref{assum:learning} 
and  the Leibniz integral rule,
for all   $f\in L^2([0,T],\sR)$ and 
for a.e.~$t\in [0,T]$,
\begin{align}
\label{eq:u^e_derivative}
\begin{split}
\frac{d }{d t}\left((\boldsymbol{u}^e)^* f\right)(t)
& =
 \frac{d }{d t}\left(\int_t^T u^e({s-t}) f(s) d s \right)
=- u^e(0)f(t) -\int_t^T   \dot{u}^e({s-t}) f(s)  d s
\\
&= -\left((u^e(0) \id+\boldsymbol{\dot{u}}^e)^*  f\right)(t),
\end{split}
\end{align}
where $(\boldsymbol{\dot{u}}^e)^*$ is the adjoint of   $\boldsymbol{\dot{u}}^e$. 
This along with the integrability of $f$ implies that 
$\frac{d }{d t}\left((\boldsymbol{u}^e)^* f\right)\in L^2([0,T],\sR)$
and hence
$\mathscr{R}\left((\boldsymbol{u}^e)^*\right)\subset \{f\in H^1([0,T],\sR)\mid f(T)=0\} $. 
On the other hand,
let $f\in H^1([0,T],\sR)$ be such that $f(T)=0$.
Then by \eqref{eq:u^e_derivative},  \cite[Chapter 9, Corollary 3.16]{gripenberg1990volterra}
   and    $u^e(0)\not =0$,
   there exists $g\in L^2([0,T],\sR)$ such that 
$\frac{d }{d t}\left((\boldsymbol{u}^e)^* g\right)= \dot{f}$,
which along with $f(T)=\left((\boldsymbol{u}^e)^* g\right)(T)=0$ implies that 
$f\equiv  (\boldsymbol{u}^e)^* g $.
This shows 
$\{f\in H^1([0,T],\sR)\mid f(T)=0\} \subset \mathscr{R}\left((\boldsymbol{u}^e)^*\right)$
and finishes the proof.
 \end{proof} 

To  analyse the convergence rate of \eqref{eq:lse_abbreviation},
we    
  interpret \eqref{eq:lse_G}  as a Tikhonov regularisation for \eqref{def:S_m_estimation}.
We first  adapt   
the   theoretical  framework of Tikhonov regularisation   for (deterministic) linear inverse problems to the present  setting. 
%
Let $(X,\|\cdot\|_X)$ 
be a Hilbert space,
and let $K:X\to X$ be an injective compact linear operator with a non-closed range $\mathscr{R}(K)$.
Let $x_0, y_0\in X$ be such that 
$y_0=Kx_0$, 
let 
    $y^\delta\in X$ be a noisy observation  
 of $y_0$,
 and   consider   approximations of 
 $y_0$ via the regularised Ritz approach.  
Specifically, 
let $(V_m)_{m\in \sN}\subset X$ be a  sequence of finite-dimensional subspaces  such that 
  $\ol{\bigcup_{m=1}^\infty V_m} =X$.
For each $\alpha>0$ and $m\in \sN$, 
let   $x^{\delta}_{\alpha,m}$ be the 
unique  minimiser of 
the following Tikhonov functional:
\begin{equation}\label{eq:tikhonov_projection}
x^{\delta}_{\alpha,m} =
\argmin_{x\in V_m}\left(
\|Kx-y^\delta\|^2_Y+\alpha \|x\|^2_X
\right)
=(K^*_mK_m+\alpha \id)^{-1}K^*_m y^\delta,
\end{equation}
where $K_m =KP_m$ and
  $P_m$ is the orthogonal projection 
of $X$ onto $V_m$.

The following lemma quantifies   $\|x^{\delta}_{\alpha,m}-x_0\|_X $ in terms of 
   $\alpha, m$ and  the   error 
$\|y^\delta -y_0\|_X$.
The proof essentially combines the results in    \cite{groetsch1984theory, hofmann2006approximate},
and is presented below for the reader's convenience. 


\begin{lemma}
\label{lemma:general_error}
For each  $  \alpha>0$ and $m\in \sN$,
let $x^{\delta}_{\alpha,m }\in X$ be defined by 
\eqref{eq:tikhonov_projection},
and 
for each $R>0$, let
 $\mathscr{D}(R)=\inf\{\|x_0-K^* v\|_X\mid v\in X, \|v\|_X\le R\}$.
 Then for all $ \alpha,R>0$ and $m\in \sN$,
\begin{align*}
\|x^{\delta}_{\alpha,m }-x_0\|_X
& \le 
\frac{\|y_0-y^\delta\|_X }{2 \sqrt{\alpha}}
+\gamma_m \sqrt{1+\frac{\gamma^2_m}{\alpha}}
\left(R+\frac{1}{2\sqrt{\alpha} }\mathscr{D}(R) \right)
+\mathscr{D}(R)+\frac{\sqrt{\alpha} }{2}R,
\end{align*}
where
   $\gamma_m=\|  (\id-P_m)K^* \|_{\rm op}$. \end{lemma}

\begin{proof}
For each $m\in \sN$ and $\alpha>0$, 
let 
$x_{\alpha, m }=(K^*_mK_m+\alpha \id)^{-1}K^*_m y_0 $ 
and 
$x_{\alpha}=(K^*K+\alpha \id)^{-1}K^* y_0 $. 
Then for all   $m\in \sN$ and $ \alpha>0$, 
\be \label{eq-tyk} 
\begin{aligned}
\|x^{\delta}_{\alpha,m }-x_0\|_X
& \le 
\|x^{\delta}_{\alpha,m }-x_{\alpha, m }\|_X
+
\|x_{\alpha, m } -x_{\alpha }\|_X
+
\|x_{\alpha}-x_0\|_X
\\
&\le 
\frac{\|y_0-y^\delta\|_X }{2 \sqrt{\alpha}}
+\sqrt{1+\frac{\gamma^2_m}{\alpha}}
\|(\id-P_m)x_\alpha\|_X
+\mathscr{D}(R)+\frac{\sqrt{\alpha} }{2}R,
\end{aligned}
\ee
where for the second inequality in \eqref{eq-tyk} is derived as follows: the first 
term used the spectral inequality 
$\|(A^*A+\alpha \id)^{-1}A^*\|_{\rm op}
\le \sup_{\lambda \ge 0}\frac{\sqrt{\lambda}}{\lambda+\alpha}
\le
\frac{1}{2\sqrt{\alpha}}$ 
for any compact operator $A$ (e.g., \cite[ p.~45, equation (2.48)]{engl1996regularization}),
the   second term used 
  \cite[Lemma 4.2.8]{groetsch1984theory}, and the third term 
used    \cite[Lemma 1]{hofmann2006approximate}.

Finally, by 
$K^*(KK^*+\alpha \id)
= (K^*K+\alpha \id) K^*$,
 $x_\alpha 
= (K^*K+\alpha \id)^{-1}K^*y_0
= K^*(KK^*+\alpha \id)^{-1}y_0$. 
Thus by  the definition of $\gamma_m$,
\begin{align*}
\|(\id -P_m)x_\alpha\|_X
&=
\|(\id -P_m)K^*(KK^*+\alpha \id)^{-1}K x_0\|_X
\le \gamma_m \|(KK^*+\alpha \id)^{-1}Kx_0 \|_X.
\end{align*}
For each $R>0$, let $(v^R_n)_{n\in \sN}\subset   X$ such that 
$\lim_{n\to \infty}\|x_0-K^*v^R_n\|_X=\mathscr{D}(R)$ and $\|v^R_n\|_X\le R$ for all $n\in\sN$.  
Then  for all $n\in \sN$, by spectral inequalities (see \cite[ p.~45]{engl1996regularization}),
\begin{align*}
 \|(KK^*+\alpha \id)^{-1}Kx_0 \|_X
 &= 
  \|(KK^*+\alpha \id)^{-1}K(K^* v^R_n+x_0 -K^* v^R_n)\|_X
  \\
  &\le 
    \| v^R_n\|_X+\frac{1}{2\sqrt{\alpha} }\|x_0 -K^* v^R_n\|_X
    \le R+\frac{1}{2\sqrt{\alpha} }\|x_0 -K^* v^R_n\|_X,
\end{align*}
from which by passing $n\to\infty$ yields
$ \|(KK^*+\alpha \id)^{-1}Kx_0 \|_X
    \le R+\frac{1}{2\sqrt{\alpha} }\mathscr{D}(R)$.  
Combining the above estimate with \eqref{eq-tyk} leads to the desired result.  
\end{proof}

We now apply Lemma \ref{lemma:general_error} in order to analyse  
\eqref{eq:lse_abbreviation} under a general regularity assumption of the kernel $G^\star$.
To this end, 
let $(V_N)_{N\in\sN}\subset L^2([0,T],\sR)$ be a given family of    finite-dimensional subspaces
such that  $\ol{\bigcup_{N=1}^\infty V_N} =L^2([0,T],\sR)$,
and consider a slight generalisation of   \eqref{eq:lse_abbreviation},
where   $G^N$, $N\in \sN$, is defined as the minimiser  of    \eqref{eq:lse_G}
over   $V_N$.
Then, by the first-order condition to 
\eqref{eq:lse_G},  for all $N\in \sN$, 
\begin{equation}
\label{eq:G^N_analytic}
 {G}^N 
 = \left((\boldsymbol{u}^e_N)^*\boldsymbol{u}^e_N+\tau_N \id \right)^{-1} 
(\boldsymbol{u}^e_N)^* \left(-\frac{1}{N}\sum_{m=1}^N (S^m-A^m+\lambda^N u^e)\right),
\end{equation}
with $\boldsymbol{u}^e_N \coloneqq  \boldsymbol{u}^e P_N$, where $P_N$
  is the   projection 
of $L^2([0,T],\sR)$ onto $V_N$.
For each $N\in \sN$, let 
\be \label{gamma-n} 
\gamma_N=\|  (\id-P_N)(\boldsymbol{u}^e)^* \|_{\rm op}.
\ee 
Note that if $(V_N)_{N\in\sN}$ is the space of piecewise constant functions as in \eqref{eq:piecewise_constant},
then 
there exists  $C>0$, depending   on $u^e$,
such that 
$\gamma_N \le C|\pi_N|$ for all $N\in \sN$.
Indeed,  for all $N\in \sN$ and $f\in L^2([0,T],\sR)$,
 \begin{equation}
 \label{eq:piecewise_constant_error}
\|  (\id-P_N)(\boldsymbol{u}^e)^* f\|_{L^2([0,T])}
\le |\pi_N|\|\tfrac{d}{d t}((\boldsymbol{u}^e)^* f)\|_{L^2([0,T])}
\le C|\pi_N|\|f\|_{L^2([0,T])},
\end{equation}
where the first and second inequalities used 
\cite[Theorem 6.1]{schumaker2007spline}
and 
  \eqref{eq:u^e_derivative}, respectively.
We further 
introduce the following    function
 $\mathscr{D}:(0,\infty)\to [0,\infty)$
to  measure the  regularity    of   $G^\star$:
 for all $R>0$,
   \begin{equation}\label{eq:distance_function}
   \mathscr{D}(R)=\inf\{\|G^\star-(\boldsymbol{u}^e)^* v\|_{L^2([0,T])}\mid v\in L^2([0,T],\sR), \|v\|_{L^2([0,T])}\le R\}.
   \end{equation}
 Note that  $G^\star \in  \mathscr{R}((\boldsymbol{u}^e)^*)$ if and only if $\mathscr{D} (R)=0$ for all large $R>0$.
 As most commonly used kernels are not in  $\mathscr{R}((\boldsymbol{u}^e)^*)$ (see Remark \ref{rmk:G_regularity_example} and Lemma \ref{lemma:u_operator_property}), 
 the subsequent analysis     focuses  
on the case with
   $\mathscr{D}(R)>0$ for all  $R>0$.

   The following theorem presents a general version of Theorem \ref{thm:estimate_error_high_probability},
   and 
  quantifies the accuracy of 
 $(\theta^N)_{N\in \sN}$  
 under a power-type decay rate of the function $\mathscr{D}$.

  \begin{theorem}
  \label{thm:estimate_error_general}
  Suppose Assumptions \ref{assum:learning} 
  and \ref{assum:concentration_M}
  hold.
Assume further that  there exists    $\beta\in (0,1]$ such that 
  $\lim\sup_{R\to\infty}\mathscr{D}(R)R^{\beta } <\infty $.
  Let $ {C}\ge 1$. 
 Then 
 for all $\eta\in (0,1)$, 
by setting 
$(\tau_N)_{N\in \sN}\subset (0,\infty)$ and $(V_N)_{N\in \sN}\subset L^2([0,T],\sR)$
such that 
$$
\textstyle
\frac{1}{ {C}}
\left(\frac{  \log(  {\eta}^{-1})+\log N    }{\sqrt{N}}
\right)^{\frac{2(\beta+1)}{2\beta+1}}
\le
\tau_N \le  {C} 
\left(\frac{  \log(  {\eta}^{-1})+\log N    }{\sqrt{N}}
\right)^{\frac{2(\beta+1)}{2\beta+1}},
\quad 
\gamma_{N} \le   {C}  \tau_N^{1/2}, 
$$
it holds 
  with probability at least $1-\eta$ that,
  for all $N\in \sN\cap [2,\infty)$,
$$\textstyle
  |\lambda^N-\lambda^\star| 
  \le C'\left( \frac{  \log(  {\eta}^{-1})+\log N    }{\sqrt{N}} \right),
  \quad 
\|G^N-G^\star\|_{L^2([0,T])}
\le C'\left( \frac{  \log(  {\eta}^{-1})+\log N    }{\sqrt{N}} \right)^{\frac{\beta}{2\beta+1}}.
$$
for   some constant $C'>0$ independent of $\eta$ and $N$.

  \end{theorem}

\begin{proof}
   Throughout this proof,
let  $\eta\in (0,1)$ be fixed, 
and 
    let $C'$ be a generic constant, which is 
    independent of $  N$ and $\eta$, and may take a different value at each occurrence.  
  
  For each $N\in \sN$, by \eqref{def:S_m_estimation} and \eqref{eq:lse_lambda},    
 \begin{equation}
 \label{eq:lambda_difference}
  \lambda^N-\lambda^\star= 
 -\frac{1}{Nu^e(0)}\sum_{m=1}^N(M^m_0-\lambda^\star u^e(0))-\lambda^\star=-\frac{1}{Nu^e(0)}\sum_{m=1}^NM^m_0.
 \end{equation}
Let 
$y_0= - (\sE[S^1-A^1] +\lambda^\star u^e)$ and for each $N\in \sN$, let 
$y_N= -\frac{1}{N}\sum_{m=1}^N (S^m-A^m+\lambda^N u^e)$,
and let  $\delta_N=\|y_N-y_0\|_{L^2([0,T])}$. 
Then by \eqref{def:S_m_estimation} and $\sE[M^m_t]=0$ for all $t\in [0,T]$ and $m\in \sN$, 
\begin{align}
\label{eq:y_difference}
\begin{split}
y_0-y_N
&=
\frac{1}{N}\sum_{m=1}^N (S^m-A^m) -\sE[S^1-A^1]
+( \lambda^N-\lambda^\star) u^e
\\
&= 
\frac{1}{N}\sum_{m=1}^N M^m
+
( \lambda^N-\lambda^\star) u^e.
\end{split}
\end{align}

Observe that $\boldsymbol{u}^e G^\star=y_0$ and the choice of 
 $(\tau_N, V_N)_{N\in \sN}$ ensures that  
$\gamma_{N} \le C'\sqrt{\tau_N}$ for all $N\in \sN$.
Moreover, 
  $\lim\sup_{R\to\infty}\mathscr{D}(R)R^{\beta } <\infty $
  implies that there exists $C>0$ such that 
$\mathscr{D}(R)\le CR^{-\beta } $ for all sufficiently large $R>0$. 
  Then by 
  Lemma \ref{lemma:general_error}
  (cf.~\eqref{eq:tikhonov_projection}
  and \eqref{eq:G^N_analytic})
  and  the decay condition of $\mathscr{D}(R)$,
  for all $N\in \sN$ and $  R>0$, 
\begin{align*}
\|G^N-G^\star\|_{L^2([0,T])}
& \le 
\frac{\delta_N }{2 \sqrt{\tau_N}}
+C'\sqrt{\tau_N} 
\left(R+\frac{1}{ R^{\beta}\sqrt{\tau_N} } \right)
+\frac{1}{R^\beta}+\frac{R\sqrt{\tau_N} }{2}.
\end{align*}
from which  by setting  $R= \tau_N^{-1/2(\beta+1)}$,
it holds   for all $N\in \sN$, 
\begin{align}
\label{eq:G_error_delta_N}
\|G^N-G^\star\|_{L^2([0,T])}
& \le 
\frac{\delta_N }{2 \sqrt{\tau_N}}
+C' \tau_N^{\frac{\beta}{2(\beta+1)}}.
\end{align}

We now estimate  $(\delta_N)_{N\in \sN}$  in  high probability. 
For each $N\in \sN$, by applying  
Assumption \ref{assum:concentration_M} with $\widetilde{\eta}=\eta/N^2$, 
with probability at least  $1- {\eta}/N^2$, 
\begin{equation*}
\left|\frac{1}{N}\sum_{m=1}^N M_0^m\right|
+\left\|\frac{1}{N}\sum_{m=1}^N M^m\right\|_{L^2([0,T])} 
\le C'\left(\log(2 {\eta}^{-1})+\log N  \right) N^{-\frac{1}{2}},
\end{equation*}
which along with {$\sum_{N=2}^\infty\frac{1}{N^2}<1$} implies that 
with probability at least  $1-  \eta$, 
  \begin{equation}
\label{eq:concentration_eta_N}
\left|\frac{1}{N}\sum_{m=1}^N M_0^m\right|
+\left\|\frac{1}{N}\sum_{m=1}^N M^m\right\|_{L^2([0,T])} 
\le C'\frac{  \log(  {\eta}^{-1})+\log N   }{\sqrt{N}},
\quad \textrm{for all } N\in \sN\cap[2,\infty).
\end{equation}
Consider the event where \eqref{eq:concentration_eta_N} holds.
Then, by \eqref{eq:lambda_difference} and \eqref{eq:y_difference},
 \begin{equation}
\label{eq:lambda_error_probability}
| \lambda^N-\lambda^\star |
+ \| y_N-y_0 \|_{L^2([0,T])}
\le C' \frac{ \log( {\eta}^{-1})+\log N  }{\sqrt{N}},
\quad \textrm{for all } N\in \sN\cap[2,\infty).
\end{equation}
Substituting the   bound of
$\| y_N-y_0 \|_{L^2([0,T])}$ and the choice of $\tau_N$ 
 into 
\eqref{eq:G_error_delta_N} yields 
    the    estimate of  
$\|G^N-G^\star\|_{L^2([0,T])}$. 
  \end{proof}

Based on   Theorem \ref{thm:estimate_error_general},
it suffices  to establish the precise decay rate of $\mathscr{D}:(0,\infty)\to [0,\infty)$, in order to conclude 
  Theorem \ref{thm:estimate_error_high_probability}
  (recall the bounds of $(\gamma_N)_{N\in \sN}$ following from \eqref{gamma-n} and \eqref{eq:piecewise_constant_error}).

\begin{theorem}
\label{thm:convergece_specifici_decay}
  Suppose Assumption  \ref{assum:learning} 
  holds.
  \begin{enumerate}[(1)]
  \item 
  \label{item:decay_regular}
    If Assumption \ref{assume:kernel_assumption}\ref{label:regular_kernel} holds,
then  
$\lim\sup_{R\to \infty}\mathscr{D}(R) R<\infty$.
Consequently, Theorem \ref{thm:estimate_error_general} holds with $\beta =1$.
\item 
 \label{item:decay_singular}
 If Assumption \ref{assume:kernel_assumption}\ref{label:singular_kernel} holds,
then 
$\lim\sup_{R\to \infty}\mathscr{D}(R) R^{\frac{1-2\alpha}{1+2\alpha}}<\infty$.
Consequently, Theorem \ref{thm:estimate_error_general} holds with $\beta =\frac{1-2\alpha}{1+2\alpha}$.
\end{enumerate}
\end{theorem}

\begin{proof}
Throughout this proof $C,C',\wt C>0$ are generic constants that 
may take a different value at each occurrence,
 and are   independent of $R$. 
The proof relies on constructing specific sequence $(v^R)_{R>0}$ in $L^2([0,T],\sR)$ such that $\|v^R\|_{L^2([0,T])}\le CR$
for all large $R$,
and quantifying the decay rate of $\|G^\star-(\boldsymbol{u}^e)^* v^R\|_{L^2([0,T])}$ as $R\to\infty$.

To prove Item \ref{item:decay_regular},
let $R_0=T^{-1/2}>0$.
Then for   each $R>R_0$, define   $G_R\in  H^1([0,T],\sR)$
such that $G_R(t)=G^\star(t)$ for all $t\in [0,T-R^{-2}]$
and $G_R(t)= (T-t)R^2 G^\star(T-R^{-2})$  for all 
$t\in [T-R^{-2},T]$.
Lemma \ref{lemma:u_operator_property} implies that  $(G_R)_{R>R_0}\subset   \mathscr{R}((\boldsymbol{u}^e)^*)$.
For each $R>R_0$,  let $v^R\in L^2([0,T],\sR) $ be such that    $G_R = (\boldsymbol{u}^e)^* v^R$, which along with  \eqref{eq:u^e_derivative} implies that 
$\dot{G}_R= -(u^e(0) \id+\boldsymbol{\dot{u}}^e)^*  v^R$.
As $(u^e(0) \id+\boldsymbol{\dot{u}}^e)^* $ has a bounded inverse on $L^2([0,T],\sR)$
(see \cite[Corollary 9.3.16, p~238]{gripenberg1990volterra}), 
$  \|v^R\|_{L^2([0,T])}\le C
 \| \dot{G}_R\|_{L^2([0,T])}$ for all $R>R_0$.
Observe that  $\dot{G}_R(t)= \dot{G}^\star(t)$ for all $t\in [0,T-R^{-2}]$ and $\dot{G}_R(t) = -R^2 G^\star(T-R^{-2})$ 
for all 
$t\in [T-R^{-2},T]$.
This along with $G^\star \in H^1([0,T],\sR)$ implies that 
$ \| \dot{G}_R\|_{L^2([0,T])}\le CR $ for all $R>R_0$, and hence $\|v^R\|_{L^2([0,T])}\le CR$ for all $R>R_0$.
To estimate $\|G^\star-(\boldsymbol{u}^e)^* v^R\|_{L^2([0,T])}$,
note that by the definition of $G_R$
and the fact that $G^\star\in H^1([0,T],\sR)$ with $G^\star(T)\not =0$,
\begin{align*}
 \| G^\star - {G}_R\|^2_{L^2([0,T])} 
 &=\int_{T-R^{-2}}^T |G^\star(t)-(T-t)R^2 G^\star(T-R^{-2})|^2 d t
 \\
& \le \wt C\left( R^{-2} +R^4 \int_{T-R^{-2}}^T(T-t)^2 d t\right)\le \wt CR^{-2}.
  \end{align*}
 Recalling \eqref{eq:distance_function}, this shows that
there exists   $C,\wt C>0$ such that 
 $$\mathscr{D}(CR) \le  \| G^\star - {G}_R\|_{L^2([0,T])} \le \wt C R^{-1}, \quad \textrm{for all } R>R_0.$$
Rescaling the inequality yields Item \ref{item:decay_regular}.

Next, we prove Item \ref{item:decay_singular}.
As   $|\dot{G}^\star(t)|\le C_0 t^{-\alpha-1}$
for all $t< t_0$, we have
 for all $t< t_0$, 
 \begin{equation}
 \label{eq:G_singularity}
| {G}^\star(t)|= 
\left|{G}^\star(t_0)-\int_{t}^{t_0}\dot{G}^\star(s) d s \right|
 \le |{G}^\star(t_0)|+
  \frac{C}{\alpha }\left(t^{-\alpha}-t_0^{-\alpha}\right)
  \le Ct^{-\alpha}.
 \end{equation}
 Let $R_0>1$ be such that 
$R_0^{-2}<R_0^{-2/(2\alpha+1)}< \min(t_0,T-t_0)$.
 For each $R>R_0$, define  $G_R\in  H^1([0,T],\sR)$
such that 
$$
G_R(t)
=\begin{cases}
G^{\star}(R^{-2/(2\alpha+1)}),  & t\in \cI_1\coloneqq (0,R^{-2/(2\alpha+1)}), 
\\
G^\star(t), & t\in \cI_2\coloneqq [R^{-2/(2\alpha+1)},T-R^{-2}],
\\
(T-t)R^2 G^\star(T-R^{-2}), & t\in \cI_3\coloneqq  [T-R^{-2},T].
\end{cases}
$$
Lemma \ref{lemma:u_operator_property} implies that  $(G_R)_{R>R_0}\subset   \mathscr{R}((\boldsymbol{u}^e)^*)$.
For each $R>0$,  let $v^R\in L^2([0,T],\sR) $ be such that    $G_R = (\boldsymbol{u}^e)^* v^R$. Then by similar arguments as above,   $  \|v^R\|_{L^2([0,T])}\le C
 \| \dot{G}_R\|_{L^2([0,T])}$ for all $R>R_0$,   
 where for   each $R>R_0$,
 $$
\dot{G}_R(t)
 = 0,  \;   t\in \cI_1;
 \quad 
\dot{G}_R(t) = \dot{G}^\star(t), 
\;  t\in \cI_2;
\quad 
\dot{G}_R(t) =
-R^2 G^\star(T-R^{-2}), \;  t\in \cI_3.
$$
Then by $G^\star \in H^1([t_0,T],\sR)$, 
$|G^\star(t)|\le C$ for all  $t\in [T-R_0^{-2},T]$.
Hence by \eqref{eq:G_singularity}, for all $R>R_0$, 
 \begin{align*}
 \| \dot{G}_R\|^2_{L^2([0,T])}
&  \le     \int_{R^{-\frac{2}{2\alpha+1}}}^{t_0}
 (\dot{G}^{\star}(t))^2 d t
+  
 \int_{t_0}^{T-R^{-2}}
( \dot{G}^{\star}(t))^2 d t
 +R^4\int_{T-R^{-2}}^T (G^\star(T-R^{-2}))^2 d t 
\\
&
\le  
C\left(  R^2 -t_0^{-2\alpha-1}
+ \| \dot{G}^\star \|^2_{L^2([t_0,T])}
+R^2
\right)
\le CR^2.
  \end{align*}
 This implies that   $\|v^R\|_{L^2([0,T])}\le CR$ for all $R> R_0$.
 To estimate $\|G^\star-(\boldsymbol{u}^e)^* v^R\|_{L^2([0,T])}$,
 for all $R>R_0$, 
\begin{align*}
& \| G^\star - {G}_R\|^2_{L^2([0,T])}  
 \le C\left( \| G^\star  \|^2_{L^2(\cI_1)}
+ \| G^\star  \|^2_{L^2(\cI_3)}
 +\| G_R  \|^2_{L^2(\cI_1)}
 +\| G_R  \|^2_{L^2(\cI_3 )}
\right)
  \\
  & 
  \le C\bigg(   \int_0^{R^{-\frac{2}{2\alpha+1}}}
 (G^{\star}(t))^2 d t
+  R^{-2}
 +
 \int_0^{R^{-\frac{2}{2\alpha+1}}}
 \big(G^{\star}\big( R^{-\frac{2}{2\alpha+1}}\big)\big)^2 d t
 +R^4\int_{T-R^{-2}}^T (T-t)^2 d t 
\bigg)
\\
& 
\le  
C  R^{-\frac{2(1-2\alpha)}{2\alpha+1}},
  \end{align*}
  where the last inequality used 
\eqref{eq:G_singularity}
 and
 $R^{-\frac{2(1-2\alpha)}{2\alpha+1}}\ge R^{-2}$
 due to 
  $\alpha\in (0,1/2)$.
 This shows that
there exists   $C,\wt C>0$ such that 
   $\mathscr{D}(CR) \le  \| G^\star - {G}_R\|_{L^2([0,T])}\le  \wt C R^{-\frac{ 1-2\alpha}{2\alpha+1}}$ for all sufficiently large $R>0$.  
       Rescaling the inequality yields Item \ref{item:decay_singular}.
\end{proof}

Now we are ready to prove Theorem \ref{thm:estimate_error_high_probability}. 

\begin{proof}[Proof of Theorem \ref{thm:estimate_error_high_probability}]
By combining the result of Theorem \ref{thm:estimate_error_general} with the decay rate of $\mathscr{D}:(0,\infty)\to [0,\infty)$ in Theorem \ref{thm:convergece_specifici_decay} and the bounds on $(\gamma_N)_{N\in \sN}$ following from \eqref{gamma-n} and \eqref{eq:piecewise_constant_error} we get the result. 
\end{proof}

  The following proposition shows that   the    decay  rates
  of $\mathscr{D} $
   in Theorem \ref{thm:convergece_specifici_decay} are optimal, i.e., 
they   are   the maximal power-type decay rates under 
Assumption \ref{assume:kernel_assumption}. Specifically, 
 it  considers      the power law kernel $G^\star(t)= t^{-\alpha}$,
  which satisfies    Assumption \ref{assume:kernel_assumption}\ref{label:regular_kernel} if $\alpha=0$,
  and       Assumption \ref{assume:kernel_assumption}\ref{label:singular_kernel} if $\alpha\in (0,1/2)$.
  \begin{proposition} \label{prop-opt} 
  Let $\alpha\in [0,1/2)$,  
  let 
  $G^\star\in L^2([0,T],\sR)$ be such that 
  $G^\star(t)= t^{-\alpha}$
   for $t\in (0,T]$,
   and 
let $u^e \in H^1([0,T],\sR)$ be such that 
      $u^e(t)=1$ for all $t\in [0,T]$. 
      Then $\lim\sup_{R\to \infty}\mathscr{D}(R) R^{\frac{1-2\alpha}{1+2\alpha}} <\infty$ 
and 
for all $\eps>\frac{1-2\alpha}{1+2\alpha}$,
 $\lim\sup_{R\to \infty}\mathscr{D}(R) R^{\eps} =\infty$.
   \end{proposition}
  
  \begin{proof}
In the sequel, we focus on the case with $\alpha\in (0,1/2)$,
as the result for $\alpha=0$ has been proved in \cite{hofmann2006approximate}. 
By \cite[Remark 1]{hofmann2006approximate}, 
the decay rate of $\mathscr{D}$ can be characterised in terms of the singular system of $\boldsymbol{u}^e$.
More precisely, let  
$\sigma_1\ge \sigma_2\ge \ldots >0$  be  the ordered singular values  of   $\boldsymbol{u}^e $
with $\lim_{n\to \infty}\sigma_n =0$,
and $(\mathfrak{u}_n)_{n\in \sN} $ be the orthonormal eigensystems of   
$(\boldsymbol{u}^e)^*\boldsymbol{u}^e$. 
Then
 for any $G^\star\in L^2([0,T],\sR)$ with
 \begin{equation}
\label{eq:max_decay_rate}
 \kappa_{0 }\coloneqq  {\sup} \left\{ \kappa>0
\,\bigg\vert\,
\sum_{n=1}^\infty \frac{1}{\sigma_n^{2 \kappa }}\langle G^\star,\mathfrak{u}_n \rangle^2_{L^2([0,T])} <\infty
 \right\}\in (0,1),
\end{equation}
and for any $\kappa\in (0,1)$, $\lim\sup_{R\to \infty} \mathscr{D}(R)R^{\frac{\kappa}{1-\kappa}}<\infty$
if and only if $0<\kappa\le \kappa_0$. 
In the sequel, we compute the value  $\kappa_0$ for $G^\star(t)=t^{-\alpha}$
and $u^e(t)=1$, $t\in [0,T]$.

The fact that  $u^e\equiv 1$ implies that $(\boldsymbol{u}^e f)(t)=\int_0^t f(s)ds$ for all $f\in L^2([0,T],\sR)$ and $t\in [0,T]$. 
Let  $\alpha\in (0,1/2)$ and $G^\star(t)=t^{-\alpha}$ for all $t\in (0,T]$. 
A direct computation shows that for all $n\in \sN$,   $\sigma_n= \frac{ T}{\pi (n-\frac{1}{2})}$,
$\mathfrak{u}_n (t)= \sqrt{\frac{2}{T}}\cos((n-\frac{1}{2})\pi \frac{t}{T})$ for all $t\in [0,T]$,
and
\begin{align}
\begin{split}
\label{eq:fourier_coefficient_G^star}
 \langle G^\star,\mathfrak{u}_n \rangle_{L^2([0,T])}
&= \sqrt{\frac{2}{T}} \int_0^T t^{-\alpha} \cos\left(  \tfrac{n-\frac{1}{2}}{T}\pi t \right) d t
=
\sqrt{\frac{2}{T}}\left( \tfrac{n-\frac{1}{2}}{T}\right)^{\alpha-1} 
 \int_0^{n-\frac{1}{2}} t^{-\alpha} \cos ({ \pi t}) d t.
\end{split}
\end{align}

To estimate $\langle G^\star,\mathfrak{u}_n \rangle_{L^2([0,T])}$ for large $n$, 
we   prove that  
$\lim_{R\to \infty}\int_0^{R} t^{-\alpha} \cos ({ \pi t}) d t
\in (0,\infty)$.
 Indeed,  for each $0<R_1<R_2<\infty$,
consider the    closed contour 
$\gamma=(\gamma_1,\gamma_2,\gamma_3,\gamma_4)\subset \sC$,
where
$\gamma_1(\theta)=R_1e^{i  (\frac{\pi}{2}- \theta)}$ for
$ \theta\in [0, \frac{\pi}{2}]$,
$\gamma_2(r)=r  $ for
$ r\in [R_1, R_2]$,
$\gamma_3(\theta)=R_2e^{i   \theta }$ for
$ \theta\in [0, \frac{\pi}{2}]$,
and 
$\gamma_4 (r) = i(R_1+R_2 -r) $
for
$ r\in [R_1, R_2]$.
As $\alpha>0$,  
the function 
$\sC\setminus \{0\}\ni z\mapsto 
z^{-\alpha} e^{i\pi z} \in \sC $  is analytic.
Hence by 
  the Cauchy-Goursat theorem
  and the definition of   contour integral,
  for each $0<R_1<R_2<\infty$,  
\begin{align}
\label{eq:contour_integral}
0=\int_{\gamma }z^{-\alpha} e^{i\pi z} d z 
= 
I_{R_1}+
\int_{R_1 }^{R_2}
r ^{-\alpha} e^{i\pi r} d r
+
I_{R_2}
+
 \int_{R_2 }^{R_1}
(ir)^{-\alpha }
  e^{ i \pi (ir) }   i d r,
\end{align}
where $I_{R_1}$
and $I_{R_2}$
 are  the integrals along the curves $\gamma_1$ and 
 $\gamma_3$, respectively:
\begin{align*}
 I_{R_1}
 & =\int_{0 }^{\frac{\pi}{2}}
  \gamma_1(\theta)^{-\alpha}  
 e^{i\pi \gamma_1(\theta)} 
 R_1e^{i  (\frac{\pi}{2}- \theta)} (-i) d \theta,
 \quad  
 I_{R_2}   =\int_{0 }^{\frac{\pi}{2}}
 \gamma_3(\theta)^{-\alpha}  
 e^{i\pi \gamma_3(\theta)}
   R_2e^{i   \theta } i d \theta.
\end{align*}
By the definition of $\gamma_1$,  
for all $\theta\in (0,\pi/2)$
and $R_1>0$,
$|e^{i\pi \gamma_1(\theta)} |= 
 e^{-\pi R_1\sin(\frac{\pi}{2}-\theta)}\le 1
 $.
As
$|\int_{0 }^{\frac{\pi}{2}} f(\theta) d \theta|\le 
\int_{0 }^{\frac{\pi}{2}} |f(\theta)| d \theta$
for any integrable  $f: [0, \frac{\pi}{2}]\to \sC$,
$$
\lim_{R_1\to 0^+} | I_{R_1}|\le
\lim_{R_1\to 0^+} \int_{0 }^{\frac{\pi}{2}}
   R_1^{1-\alpha}  d \theta 
   \le  \lim_{R_1\to 0^+}  \frac{\pi}{2} R_1^{1-\alpha}=0,
$$
where the last identity used 
  $\alpha \in (0,1/2)$.
   Similarly,   
by the definition of   $\gamma_3$,  
for all $\theta\in (0,\pi/2)$
and $R_2>0$,
 $|e^{i\pi \gamma_3(\theta)} | = 
 e^{-\pi R_2\sin(\theta)}  $ 
 and  
$  |I_{R_2}|  \le
R_2^{1-\alpha}   
\int_{0 }^{\frac{\pi}{2}}
 e^{-\pi R_2\sin(\theta)}
  d \theta$.
 As $\sin\theta\ge \frac{2}{\pi }\theta$ for all    $\theta\in (0,\pi/2)$, 
$$
\lim_{R_2\to \infty}   |I_{R_2}|  \le
\lim_{R_2\to \infty}
\left(
R_2^{1-\alpha}   
\int_{0 }^{\frac{\pi}{2}}
 e^{-2 R_2 \theta } d \theta
     \right)
     =
   \lim_{R_2\to \infty}
\left(
 \frac{1}{2}R_2^{-\alpha} (1-e^{- R_2 \pi})\right)
  =0.
     $$
  Thus,   letting $R_1\to 0^+$ and $R_2\to \infty$ in  \eqref{eq:contour_integral} yields 
$$
\lim_{R_2 \to \infty}\int_{0 }^{R_2}
r ^{-\alpha} e^{i\pi r} d r
= 
i^{1-\alpha}
\lim_{R_2\to \infty}\int_{0 }^{R_2}
r^{-\alpha }
 e^{-\pi r }  d r.
$$
The   upper bound
$\lim_{R\to \infty}\int_0^{R} t^{-\alpha} \cos ({ \pi t}) d t
<\infty$ 
  follows from 
$\lim_{R_2\to \infty}
\int_{0 }^{R_2}
r^{-\alpha }
 e^{-\pi r } d r<\infty$
 and 
$
 \int_0^{R} t^{-\alpha} \cos( \pi t) d t$ is 
 the real part of 
$\int_0^{R} t^{-\alpha} e^{i \pi t} d t$
for all $R>0$.
 For   the lower bound, consider the sequence 
$
\left(a_n
 \right)_{n\in \sN}
 $ with 
$a_n\coloneqq \int_{0 }^{2n}
t^{-\alpha} \cos ({ \pi t})  d t$ for all $n\in \sN$.
As $t\mapsto t^{-\alpha}$ is decreasing on $(0,\infty)$,
$a_{n+1}\ge a_n>0$ for all $n\in \sN$,
and hence 
$\lim_{R\to \infty}\int_0^{R} t^{-\alpha} \cos ({ \pi t})  d t= \lim_{n\to \infty}\int_0^{2n} t^{-\alpha} \cos ({ \pi t}) d t
\ge a_1>0$.

  Therefore, 
by \eqref{eq:fourier_coefficient_G^star},
 for each $\alpha\in (0,1/2)$, 
 there exists $C>0$, depending   on $\alpha$ and $T$, such that 
 $
 \frac{1}{C}n^{2\alpha -2 }
 \le 
  \langle G^\star,\mathfrak{u}_n \rangle^2_{L^2([0,T])}
\le Cn^{2\alpha -2 }$ for all $n\in \sN$,
and for all $\kappa>0$, 
\begin{equation}
\label{eq:singular_sum}
\frac{1}{C} \sum_{n=1}^\infty \frac{1}{n^{2-2\kappa  -2\alpha }}
\le 
\sum_{n=1}^\infty \frac{1}{\sigma_n^{2\kappa  }}\langle G^\star,\mathfrak{u}_n \rangle^2_{L^2([0,T])} 
\le C \sum_{n=1}^\infty \frac{1}{n^{2-2\kappa  -2\alpha }}.
\end{equation}
Hence 
$\sum_{n=1}^\infty \frac{1}{\sigma_n^{2\kappa }}\langle G^\star,\mathfrak{u}_n \rangle^2_{L^2([0,T])} <\infty$ if and only if $\kappa\in (0, \frac{1}{2}-\alpha)$.  
This implies that  $\kappa_0=   \frac{1}{2}-\alpha\in (0,1)$ in \eqref{eq:max_decay_rate}, and subsequently proves that
 for all $\kappa\in ( \frac{1}{2}-\alpha,1)$, $\lim\sup_{R\to \infty} \mathscr{D}(R)R^{\frac{\kappa}{1-\kappa}}=\infty$. 
The  desired statement follows from 
the fact that $\left\{\frac{\kappa}{1-\kappa} \mid \kappa\in \left( \frac{1}{2}-\alpha,1\right)\right\}= \left(\frac{1-2\alpha }{1+2\alpha},\infty \right)$. 
 \end{proof}

 \section{Proof of Theorem \ref{thm:regret_general_rate}} \label{sec-pf-rate} 
 
 Throughout this section, 
%
   we denote by $C>0$ a generic constant,
    which is independent of $\eta\in (0,1)$ and $N\in \sN$, 
     and may take a different value at each occurrence.
  
  The following lemma proves that 
  for sufficiently large  $\mathfrak{m}^e_0 \in \sN $,
  the estimators $(\theta^k)_{k\ge 0 }$ from Algorithm   \ref{alg:phased}
  are in the admissible parameter set $\Xi_{\eps}$ from Definition \ref{assum1:parameters}.
  
  \begin{lemma}
  \label{lemma:initial_exploration}
  Assume that condition \eqref{eq:lse_error_regret} holds.
  Let $(\theta^k)_{k\ge 0 }$ be generated from Algorithm \ref{alg:phased}.
  Then there exists $C_0>0$  
  such that for all $\eta\in (0,1)$,
  and $\mathfrak{m}^e_0 \ge C_0( \log(  {\eta}^{-1})^2+1)$, the following hold with probability at least $1-\eta$.
  \begin{enumerate}[(1)]
  \item 
\label{item:convergence_lse_admissible}
  $\theta^k=(\lambda^k,G^k) \in \Xi_{\eps}$   for all  $  k\in \sN\cup\{0\}$,
  \item
  \label{item:convergence_lse_m0}
  There  exists   $C>0$, independent of 
  $\eta$, 
  such that for all 
  $k\in \sN\cup\{0\}$, 
$$
 |\lambda^k-\lambda^\star| 
+\|G^k-G^\star\|_{L^2([0,T])}
\le C\left( \frac{  \log(  {\eta}^{-1})+\log (  \mathfrak{m}^e_0 +k)  }{\sqrt{ \mathfrak{m}^e_0 +k}} \right)^{\kappa}.
$$
\end{enumerate}
  \end{lemma}

  \begin{proof}
  Observe that 
  for all $G, f\in L^2([0,T],\sR)$,
    by the Cauchy-Schwarz inequality, 
  \begin{align*}
&\int_{0}^T\int_{0}^T G(|t-s|)  f(s)  f(t) dsdt 
= 2  \int_{0}^T\int_{0}^T G(t-s ) \mathds{1}_{\{s\le t\}} f(s)  f(t) dsdt 
\\
&\quad \le 2\|f\|_{L^2([0,T])}  \int_{0}^T\left(\int_{0}^T (G(t-s ) )^2\mathds{1}_{\{s\le t\}} ds \right)^{1/2}  f(t) dt  
\\
&\quad
 \le 2\|f\|^2_{L^2([0,T])} \left( \int_{0}^T \int_{0}^T (G(t-s ) )^2\mathds{1}_{\{s\le t\}} ds   dt  \right)^{1/2} \\
&\quad \le  2\sqrt{T}\|f\|^2_{L^2([0,T])} \|G\|_{L^2([0,T])}.
  \end{align*}
  Consequently, by \eqref{pos-def_convolution} and Definition \ref{assum1:parameters}
  there exists $\eps'>0$ 
  such that 
  $\sB(\theta^\star, \eps') \coloneqq  \{(\lambda,G) \in \sR \times L^2([0,T], \sR)\mid 
|\lambda-\lambda^\star|+\|G-G^\star\|_{L^2([0,T])}\le \eps'\}\subset  \Xi_{\eps}$.  

Let $\tilde  C>0$ be the constant in  the condition \eqref{eq:lse_error_regret},
and assume without  loss of generality that $\eps'/ \tilde C<1$.
Observe that for all $x,y>0$,
$\frac{\partial }{\partial y}\left(
\frac{x+\log y}{\sqrt{y}}
\right)=
\frac{2-\log y-x }{2y^{3/2}}$,
and hence 
$y\mapsto \frac{x+\log y}{\sqrt{y}}$ decreases on $[e^2,\infty)$ for all $x>0$.
Let  $C_0\ge e^2$ be a constant such that
$\frac{x+\log(C_0(x^2+1))}{\sqrt{C_0(x^2+1)}}
\le \left(\frac{\eps'}{\tilde C}\right)^{{1}/{\kappa}}$
 for all $x>0$.
Then  for any $\eta\in (0,1)$ and $N \ge C_0( \log(  {\eta}^{-1})^2+1)$, 
  \begin{equation*}
  \label{eq:log_eta_N_kappa}
  \frac{  \log(  {\eta}^{-1})+\log N    }{\sqrt{N}}\le 
  \left( \frac{  \log(  {\eta}^{-1})+\log N    }{\sqrt{N}} \right)^{\kappa}\le \frac{\eps'}{\tilde C},
  \end{equation*}
 where the first inequality used $\kappa \in (0,1)$.
  Now let 
  $\mathfrak{m}^e_0\ge C_0( \log(  {\eta}^{-1})^2+1)$.
 By the definitions of $(\theta^k)_{k\ge 0 }$ in
  \eqref{eq:lse_initial} and \eqref{eq:lse_k_cycle},
  for each $k\in \sN\cup \{0\}$,
 the estimator $ \theta^k$ is computed using $k+\mathfrak{m}^e_0$ samples,
 which along with 
    \eqref{eq:lse_error_regret}
    shows that 
    $\theta^k\in \sB(\theta^\star, \eps') $ for all 
    $k\ge 0$ in Item \ref{item:convergence_lse_admissible}.
    The convergence rate of $(\theta^k)_{k\ge 0 }$ 
    in Item \ref{item:convergence_lse_m0}
follows from  \eqref{eq:lse_error_regret}
with $N=k+\mathfrak{m}^e_0$.
  \end{proof}

We are now ready to prove Theorem \ref{thm:regret_general_rate}.

\begin{proof}[Proof of Theorem \ref{thm:regret_general_rate}]
Throughout this proof, let 
 $\cE=\{1,\ldots, \mathfrak{m}^e_0\}\cup \{ 
 \mathfrak{m}^e_0 +\sum_{j=1}^k  \mathfrak{n}(j)+k \mid k\in \sN\} $ 
 be the indices of exploration episodes,
 let  $\mathfrak{L}(k) =\mathfrak{m}^e_0 +\sum_{j=1}^k  \mathfrak{n}(j)+k$, $k\in \sN\cup\{0\}$ be the index of the last episode in the $k$-th cycle
 (cf.~Algorithm  \ref{alg:phased}), 
  and 
  let $\mathfrak{c}(m)$, $m\in \sN$, be the corresponding cycle for the $m$-th episode,
  i.e., $\mathfrak{c}(m)=\min\{k\in \sN\cup\{0\}\mid   \mathfrak{L}(k)\ge m\}$.
  Let 
  $\mathfrak{n}(k)=\lfloor k^\delta\rfloor$, $k\in \sN$, with some   $\delta\in (0,1)$
  to be determined later.
 Observe that 
for all $m> \mathfrak{m}^e_0$, 
$\mathfrak{L}(\mathfrak{c}(m)-1)<m\le \mathfrak{L}(\mathfrak{c}(m))$,
  which along with 
the inequality that 
$ 
k^\delta-1\le \lfloor k^\delta\rfloor
\le k^\delta
$ for all $k\in \sN$
implies that 
$$
\sum_{j=1}^{\mathfrak{c}(m)-1}  (j ^\delta-1)+\mathfrak{c}(m)-1
<m-\mathfrak{m}^e_0
\le 
\sum_{j=1}^{\mathfrak{c}(m)}   j^\delta+\mathfrak{c}(m).
$$
Hence   there exists $\overline{c},\underline{c}>0$, depending on $\delta$, such that  for all $m> \mathfrak{m}^e_0$, 
\be \label{cyc-bnd} 
\underline{c} (\mathfrak{c}(m)-1)^{\delta+1} \le  m-\mathfrak{m}^e_0\le \overline{c} (\mathfrak{c}(m))^{\delta+1}
\ee
In the following, we optimise the growth rate of  $R(N)$ over $\delta$.

Let $\eta\in (0,1)$ be fixed,
and let 
$\mathfrak{m}^e_0 =\lceil C( \log(  {\eta}^{-1})^2+1)\rceil$
  with $C\ge C_0$ in Lemma \ref{lemma:initial_exploration},
  and consider the event (which  holds   with probability at least $1-\eta$)
such that  $\theta^k \in \Xi_{\eps}$ for all $k\in \sN\cup\{0\}$.
For each $N\in \sN$, by \eqref{eq:regret_def}
and Algorithm  \ref{alg:phased},
\begin{align}\label{eq:regret_decomposiition}
\begin{split}
 R(N) &= \sum_{m\in [1,N]\cap \cE} \big(   
 J^{\theta^\star}(u^{\theta^\star})
 -
 J^{\theta^\star}(u^m)\big)
 +\sum_{m\in [1,N]\setminus \cE} \big(   J^{\theta^\star}(u^{\theta^\star}) -
 J^{\theta^\star}(u^m)\big)
 \\
 &=\sum_{m\in [1,N]\cap \cE} \big(   J^{\theta^\star}(u^{\theta^\star}) -J^{\theta^\star}(u^e)\big)
 +\sum_{m\in [1,N]\setminus \cE} \big(   J^{\theta^\star}(u^{\theta^\star}) -J^{\theta^\star}(u^{\theta^{\mathfrak{c}(m)-1}})\big).
 \end{split}
 \end{align}
 If $N\le  \mathfrak{m}^e_0 $, 
    the fact that  $ J^{\theta^\star}(u^{\theta^\star}), J^{\theta^\star}(u^e)<\infty$ implies that 
  $R(N)\le  \mathfrak{m}^e_0 |    J^{\theta^\star}(u^{\theta^\star}) -J^{\theta^\star}(u^e)|\le C  ( \log(  {\eta}^{-1})^2+1) $.
Now consider    $N>  \mathfrak{m}^e_0 $.
The   number of exploration episodes up to the $N$-th episdoe is bounded by 
$ \mathfrak{m}^e_0+\mathfrak{c}(N)$.  
Moreover, by Lemma \ref{lemma:initial_exploration}
and Theorem \ref{thm:gap}, 
$|J^{\theta^\star}(u^{\theta^\star}) -J^{\theta^\star}(u^{\theta^{\mathfrak{c}(m)-1}})|
\le C(|\lambda^{\mathfrak{c}(m)-1}-\lambda^\star|^2+\|G^{\mathfrak{c}(m)-1}-G^\star\|^2_{L^2([0,T])})$ 
for some constant $C>0$. 
Hence,
 by Lemma \ref{lemma:initial_exploration} and \eqref{eq:regret_decomposiition},
  \begin{align*}
\begin{split}
 R(N) &\le C
 ( \mathfrak{m}^e_0+\mathfrak{c}(N))
 +
C  \sum_{k=1}^{\mathfrak{c}(N)} \mathfrak{n}(k)\left(|\lambda^k-\lambda^\star|^2+\|G^k-G^\star\|^2_{L^2([0,T])} \right)
\\
&\le 
C\left(\log(  {\eta}^{-1})^2 +1+ \mathfrak{c}(N)\right)
 +
C  \sum_{k=1}^{\mathfrak{c}(N)} \mathfrak{n}(k) 
\left( \frac{  \log(  {\eta}^{-1})+\log (  \mathfrak{m}^e_0 +k)  }{\sqrt{ \mathfrak{m}^e_0 +k}} \right)^{2\kappa}
\\
&\le 
C\left(\log(  {\eta}^{-1})^2 +1+ \mathfrak{c}(N)\right)
 +
C  \sum_{k=1}^{\mathfrak{c}(N)} k^{\delta -\kappa} 
\left(  {  \log(  {\eta}^{-1})+\log N  }  \right)^{2\kappa}.
\\
&\le 
C\left(\log(  {\eta}^{-1})^2 +1+ \mathfrak{c}(N)\right)
 +
C \frac{ \mathfrak{c}(N)^{\delta -\kappa+1}}{\delta -\kappa+1} 
\left(  {  \log(  {\eta}^{-1})+\log N  }  \right)^{2\kappa}.
 \end{split}
 \end{align*}
Then from \eqref{cyc-bnd} it follows that $\mathfrak{c}(N)\le C(N-\mathfrak{m}^e_0)^{1/(1+\delta)}$, which implies that 
 for all $N>\mathfrak{m}^e_0$, 
$$
R(N)\le C \left(\log(  {\eta}^{-1})^2 +N^{\frac{1}{1+\delta} }+
 \frac{ 1}{\delta -\kappa+1} N^{\frac{\delta -\kappa+1}{1+\delta}  }  \left(  {  \log(  {\eta}^{-1})+\log N  }  \right)^{2\kappa}
\right).
$$
Hence it is clear that the growth rate of $(R(N))_{N\in \sN}$ is optimised at $\delta =\kappa$. 
This proves  the desired estimate. 
\end{proof}

\section{Proof of Proposition \ref{prop:u_lipschitz}}  \label{sec-pf-prop-u} 
The proof of Proposition \ref{prop:u_lipschitz} will follow by proving the stability of the coefficients $a^{\theta}$ and $B^\theta$ in \eqref{eq:optimalcontrol_monotone_str} with respect to $\theta$. Note that by \eqref{eq:aB} the stability of these coefficients will depend on the stability of the operator ${\Gamma}_t^{-1}$ in \eqref{op-gam-inv} and of $\{\Theta_t(s)\}_{\{t\in [0,s], \, s\in [0,T]\}}$ in \eqref{th-sol} with respect to $\theta$. We emphasise the dependence of the ingredients of $a^{\theta}$ and $B^\theta$ in $\theta$ in the following. Throughout this section we assume that $\phi >0$ in \eqref{def:objective-str}, where the proof of Proposition \ref{prop:u_lipschitz} for the case where $\phi =0$ follows the same lines but is much simpler as $a^{\theta}$ and $B^\theta$ considerably simplify. 
 
We recall that for any operator $\boldsymbol{G}$ from $ L^2\left([0,T],{\R^2}\right)$ to itself we define the operator norm, 
\be \label{op-norm} 
 \|\boldsymbol{G}\|_{\rm {op}}= \sup_{f \in L^2([0,T],\R^2), f\not = 0} \frac{\|\mathbf G f\|_{L^2([0,T])}}{\|f\|_{L^2([0,T])}}.
\ee

The following Lemma and Propositions are essential ingredients for the proof of Proposition \ref{prop:u_lipschitz}.  Recall that  $\Xi_{\eps}$ was defined in Definition \ref{assum1:parameters}.  
 \begin{lemma} \label{lemma-bnd-g-inv} 
 Let $(\boldsymbol {\Gamma}_t^\theta)^{-1}$ be defined as in \eqref{op-gam-inv}, then we have,
 $$
  \sup_{\theta \in \Xi_{\eps} } \sup_{t \in[0,T]}  \|( \boldsymbol {\Gamma}_t^\theta)^{-1} \|_{\rm {op}} < \infty. 
 $$ 
 \end{lemma} 
The proof of Lemma \ref{lemma-bnd-g-inv} is postponed to Section \ref{sec-pf-thm-lips}.

\begin{proposition}[Stability of $(\boldsymbol{\Gamma}_t^{\theta})^{-1}$] \label{thm:Psi_lipschitz1}
There exists a constant $C\ge 0$ such that
$$
\sup_{ t\in [0,T]}
\|\boldsymbol {(\Gamma}_t^{\theta})^{-1}-(\boldsymbol {\Gamma}_t^{\theta'})^{-1}\|_{\rm op}\le C\left(|\lambda-\lambda'| +\|G-G'\|_{L^2([0,T])}\right), \quad \textrm{for all } \theta,\theta'\in \Xi_{\eps}. 
$$
 \end{proposition}
The proof of Proposition \ref{thm:Psi_lipschitz1} is postponed to Section \ref{sec-pf-thm-lips}.

For any $\theta \in \Xi_{\eps}$ we define $\Theta^{\theta}= \{\Theta^{\theta}_{t}(s) : t\in [0,s],\, s\in [0,T]\}$ as in \eqref{th-sol}, that is,   
\be \label{th-th} 
\Theta^{\theta}_t(s)=- \left((\boldsymbol {\Gamma}_t^\theta)^{-1}  \mathds{1}_t \E\left[A_{(\cdot)}- A_T \Mid \mathcal F_t\right] e_1 \right)(s).
\ee

\begin{proposition}[Stability of $\Theta$] \label{prop-st-th} 
There exists a constant $C> 0$ such that for all $t\in [0,T]$ and $\theta,\theta'\in \Xi_{\eps}$ we have,
 \begin{align*}
&\| \Theta^\theta_t(\cdot)-  \Theta^{\theta'}_t(\cdot)\|_{L^2([0,T])} \leq C\left(|\lambda-\lambda'|  +\|G-G'\|_{L^2([0,T])}\right)
\|\sE[A_{(\cdot)}-A_{T}\mid \cF_t]\|_{L^2([0,T])}.
\end{align*}
\end{proposition}

\begin{proof}
Throughout this proof
let $(\om,t)\in \Om\t [0,T]$ and define $$f_t(\cdot)=\mathds{1}_t(\cdot)  \sE[A_{(\cdot)}-A_{T}\mid \cF_t](\om).$$
From \eqref{th-th} it follows that  
\begin{equation} \label{th-comp} 
\begin{aligned} 
 \Theta^\theta_t(s)- \Theta^{\theta'}_t(s)  = -   \left( \left((\boldsymbol {\Gamma}^\theta_t)^{-1}  -(\boldsymbol {\Gamma}^{\theta'}_t)^{-1}  \right) \mathds{1}_t \E\left[A_{(\cdot)}- A_T \Mid \mathcal F_t\right] e_1 \right)(s)  .
\end{aligned} 
\end{equation} 
From \eqref{th-comp} and Theorem \ref{thm:Psi_lipschitz1} we get
 \begin{align*}
\label{eq:B_g}
\left\|   \Theta^\theta_t(\cdot)-   \Theta^{\theta'}_t(\cdot)  
\right\|_{L^2([0,T])} &\leq
C\left\| (\boldsymbol {\Gamma}^\theta_t)^{-1}  -(\boldsymbol {\Gamma}^{\theta'}_t)^{-1}  
\right\|_{\rm op }\|f_t(\cdot)e_1\|_{L^2([0,T])}
\\
& \le 
C\left(|\lambda-\lambda'| +\|G-G'\|_{L^2([0,T])}\right)\|f_t(\cdot)\|_{L^2([0,T])},
 \end{align*} 
 where $C> 0$ is a constant depending only on
$T,L$ and $\eps$, but  independent of $t,s\in [0,T]$, $\theta, \theta' \in \Xi_{\eps}$ and $\om \in \Om$. 
\end{proof}

\begin{proof} [Proof of Proposition \ref{prop:u_lipschitz}]
Inspired by \eqref{eq:aB} we define for any $\theta \in \Xi_{\eps}$, 
\begin{equation} \label{b-a-th} 
	\begin{aligned} 
a^\theta _t &= \frac 1{2\lambda} \left(\E\left[ A_t - A_T  \mid \mathcal F_t\right]  + 2\varrho q + \langle \Theta^\theta_t , K^\theta_t \rangle_{L^2} +  \langle \boldsymbol {(\Gamma}^\theta_t)^{-1} K^\theta_t  , \mathds 1_t  (- 2\varrho q, q )^\top \rangle_{L^2} \right), \\
B^\theta(t,s) &=  \mathds 1_{\{s<t\}}\frac 1{2\lambda} \left( \langle( \boldsymbol {\Gamma}^\theta_t)^{-1} K^\theta_t  , \mathds 1_t  ( \tilde G^{\theta}(\cdot,s), -1 )^\top \rangle_{L^2}   - \tilde G^{\theta}(t,s) \right). 
\end{aligned}  
\end{equation}
	
From Definition \ref{assum1:parameters} we have $\lam, \lam' \ge L^{-1}$ hence 
\be \label{dec-I}
\begin{aligned} 
|a^{\theta}_t- a_t^{\theta'}| &\leq \bigg| \frac 1{2\lambda}  \left(  \langle \Theta^\theta_t , K^\theta_t \rangle_{L^2} +  \langle (\boldsymbol {\Gamma}^\theta_t)^{-1} K^\theta_t  , \mathds 1_t  (- 2\varrho q, q )^\top \rangle_{L^2} \right) \\
&\quad  - \frac 1{2\lambda'}  \left(  \langle \Theta^{\theta'}_t , K^{\theta'}_t \rangle_{L^2} +  \langle (\boldsymbol {\Gamma}^{\theta'}_t)^{-1} K^{\theta'}_t  , \mathds 1_t  (- 2\varrho q, q )^\top \rangle_{L^2} \right)  \bigg| \\
 &\leq C_1 | \lambda-\lambda'|  \left|  \langle \Theta^\theta_t , K^\theta_t \rangle_{L^2} +  \langle( \boldsymbol {\Gamma}^\theta_t)^{-1} K^\theta_t  , \mathds 1_t  (- 2\varrho q, q )^\top \rangle_{L^2} \right| \\
&\quad  + C_2\big  |\langle \Theta^{\theta'}_t , K^{\theta'}_t \rangle_{L^2} -\langle \Theta^{\theta}_t , K^{\theta}_t \rangle_{L^2}\big |  \\
 &\quad+C_3 \big| \langle (\boldsymbol {\Gamma}^{\theta'}_t)^{-1} K^{\theta'}_t-( \boldsymbol {\Gamma}^{\theta}_t)^{-1} K^{\theta}_t  , \mathds 1_t  (- 2\varrho q, q )^\top \rangle_{L^2} \big|  \\
 &=: \sum_{j=1}^3 I^{\theta,\theta'}_j(t).
 \end{aligned}
 \ee 
 From \eqref{h-tilde}, \eqref{K-t} and Definition \ref{assum1:parameters} we get  
 \be \label{gf1} 
  \sup_{\theta \in \Xi_{\eps} } \sup_{t \in[0,T]}  \| K_t^\theta \|_{L^2([0,T])} < \infty. 
 \ee
From \eqref{ass:P}, Lemma \ref{lemma-bnd-g-inv} and \eqref{th-th} we get,  
\be \label{gf2} 
  \sup_{\theta \in \Xi_{ \eps} } \sup_{t \in[0,T]} \E\left[\int_0^T (\Theta_t^\theta(s))^2ds \right]< \infty. 
\ee
From \eqref{dec-I}, \eqref{gf1}, \eqref{gf2}, Lemma \ref{lemma-bnd-g-inv} and Cauchy-Schwarz inequality it follows that 
 \be \label{gf2.2} 
\sup_{t \in [0,T]} \E \big[ (I^{\theta,\theta'}_1(t)^2)\big] \leq C| \lambda-\lambda'|^2. 
\ee
 Note that, 
 \be \label{gf3} 
 \begin{aligned} 
|  I^{\theta,\theta'}_2(t)|  &\leq C_1 \big  |\langle \Theta^{\theta'}_t , K^{\theta'}_t -K^{\theta}_t \rangle_{L^2}| +C_2| \langle \Theta^{\theta}_t- \Theta^{\theta'}_t  , K^{\theta}_t \rangle_{L^2([0,T])}\big |. 
 \end{aligned}
\ee
 From \eqref{h-tilde}, \eqref{K-t} and Definition \ref{assum1:parameters} we get  
\be \label{k-lip} 
\|K^{\theta'}_t -K^{\theta}_t\|_{L^2} \leq C\|G - G'\|_{L^2([0,T])}. 
\ee
From Proposition \ref{prop-st-th}, \eqref{gf1}, \eqref{gf2}, \eqref{gf3} and \eqref{k-lip} it follows that 
\be \label{gf4} 
\sup_{t \in [0,T]}E[ (I^{\theta,\theta'}_2(t))^2]  \leq  C\|G - G'\|_{L^2([0,T])}^2. 
\ee
Following similar steps as in \eqref{gf3}--\eqref{gf4}, only using Proposition \ref{thm:Psi_lipschitz1} instead of Proposition \ref{prop-st-th} we get, 
\be \label{gf5} 
\sup_{t \in [0,T]} I^{\theta,\theta'}_3(t) \leq  C\|G - G'\|_{L^2([0,T])}. 
\ee
Plugging in \eqref{gf2.2}, \eqref{gf4} and \eqref{gf5} into \eqref{dec-I} we get, 
\begin{equation}\label{gf6} 
\sup_{t \in [0,T]}\E[(a^{\theta}_t- a_t^{\theta'})^2] \leq  C\left(|\lambda-\lambda'|^2 +\|G-G'\|^2_{L^2([0,T])}\right), \quad \textrm{for all } \theta,\theta'\in \Xi_{\eps}. 
\end{equation}
By repeating a similar argument leading to \eqref{gf6}, using \eqref{gf1}, \eqref{k-lip}, Lemma \ref{lemma-bnd-g-inv}  and Proposition \ref{thm:Psi_lipschitz1} on \eqref{b-a-th} we obtain for all $\theta,\theta'\in \Xi_{\eps}$, 
\begin{equation} \label{gf7} 
\sup_{t \in [0,T]}\|B^{\theta}(t,\cdot)- B^{\theta'}(t,\cdot)\|_{L^2([0,T])} \leq  C\left(|\lambda-\lambda'| +\|G-G'\|_{L^2([0,T])}\right),
\end{equation}
and 
\be \label{B-l-2} 
\sup_{\theta \in \Xi_{\eps}}  \sup_{t \in [0,T]}\|B^{\theta}(t,\cdot)\|_{L^2([0,T])}< \infty. 
\ee
The following bound can be obtained from \eqref{volt-u}, \eqref{gf6},\eqref{B-l-2} and Gronwall's lemma by using standard arguments,
\be \label{unif-u} 
 \sup_{\theta \in \Xi_{\eps} }\sup_{t\in [0,T]} \E \left[ (u_t^\theta)^2 \right] < \infty. 
\ee  
Since in the following we use a similar argument to derive the stability of $u^\theta$, we omit the details in order to avoid unnecessary repetition. 

Recall that the $\mathcal H^2$-norm was defined in \eqref{h-norm}. From \eqref{volt-u}, \eqref{gf6}, \eqref{gf7} and Cauchy-Schwarz inequality we therefore get 
\be  
\begin{aligned} 
\E[ (u^\theta_t - u^{\theta'}_t)^2 ]  &\leq C\Bigg( \E[(a^\theta_t-a^{\theta'}_t)^2] +\E\left[ \left(\int_0^t (B^\theta(t,s) - B^{\theta'}(t,s))u_s^\theta ds \right)^2\right] \\
& \quad +\E\left[ \left(\int_0^t B^\theta(t,s) (u_s^\theta-u_s^{\theta'}) ds \right)^2\right]  \Bigg)  \\ 
& \leq   C_1\left(|\lambda-\lambda'|^2 +\|G-G'\|^2_{L^2([0,T])}\right)(1+\|u^\theta \|_{\mathcal H^2})\\
& \quad +C_2 \sup_{\theta \in \Xi_{ \eps} }  \sup_{t \in [0,T]}\|B^{\theta}(t,\cdot)\|_{L^2([0,T])}\int_0^t \E\left[ (u_s^\theta-u_s^{\theta'})^2\right] ds.  
\end{aligned} 
 \ee
 Together with \eqref{B-l-2} and \eqref{unif-u} it follows that there exist $C_i>0$, $i=1,2$ not depending on $\theta, \theta'  \in \Xi_{ \eps}$ such that, 
\be  
\begin{aligned} 
\sup_{r\in [0,t]} \E[ (u^\theta_r - u^{\theta'}_r)^2 ]   & \leq   C_1\left(|\lambda-\lambda'|^2 +\|G-G'\|^2_{L^2([0,T])}\right) \\
& \quad +C_2 \int_0^t \sup_{r\in [0,s]}\E\left[ (u_s^\theta-u_s^{\theta'})^2\right] ds. 
\end{aligned} 
 \ee
Then from Gronwall's lemma it follows that 
$$
\sup_{r\in [0,T]} \E[ (u^\theta_r - u^{\theta'}_r)^2 ]  \le C\left(|\lambda-\lambda'|^2 +\|G-G'\|^2_{L^2([0,T])}\right), 
$$
and we get the result. 
\end{proof} 

\section{Proofs of Lemma \ref{lemma-bnd-g-inv} and Proposition \ref{thm:Psi_lipschitz1}} \label{sec-pf-thm-lips} 
As we mentioned at the beginning of Section \ref{sec-pf-prop-u}, we assume that $\phi >0$ in \eqref{def:objective-str}, where the case of $\phi =0$ is much simpler and is left  to reader. 
From \eqref{op-gam-inv} we note that for $\phi >0$, $(\boldsymbol \Gamma^\theta_t)^{-1}$ is the inverse of the operator 
\be \label{op-gamma} 
\begin{aligned}
\boldsymbol{\Gamma}_t^\theta
&= \left(\begin{matrix}
\boldsymbol{D}^\theta_t-2\phi  \boldsymbol{1}^*_t \boldsymbol{1}_t  & -\boldsymbol{1}^*_t \\
-\boldsymbol{1}_t & - \frac 1{2\phi} \id 
\end{matrix}\right). 
\end{aligned}
\ee
The proofs of Lemma \ref{lemma-bnd-g-inv} and Proposition \ref{thm:Psi_lipschitz1} will use the following auxiliary lemmas.
 \begin{lemma} \label{lemma-lip-gam} 
There exists a constant $C>0$ such that,  
 $$ \sup_{t \in[0,T]}  \| \boldsymbol {\Gamma}_t^\theta -  \boldsymbol {\Gamma}_t^{\theta'} \|_{\rm {op}} \le C\left(|\lambda-\lambda'| +\|G-G'\|_{L^2([0,T])}\right), \quad \textrm{for all } \theta,\theta'\in \Xi_{\eps}. 
$$
 \end{lemma} 
 \begin{proof}  
 From \eqref{op-gamma} it follows that it is enough to prove that there exists a constant $C>0$ such that,  
 \be \label{g1} 
 \sup_{t \in[0,T]}  \| \boldsymbol {D}_t^\theta -  \boldsymbol {D}_t^{\theta'} \|_{\rm {op}} \le C\left(|\lambda-\lambda'| +\|G-G'\|_{L^2([0,T])}\right), \quad \textrm{for all } \theta,\theta'\in \Xi_{\eps}.
\ee
From \eqref{eq:schur} we get,     
$$
 \boldsymbol {D}_t^\theta -  \boldsymbol {D}_t^{\theta'} := 2(\lambda -\lambda') \id +  (\boldsymbol{\tilde G}^\theta_t -\boldsymbol{\tilde G}^{\theta'}_t + (\boldsymbol{\tilde G}^\theta_t)^*-(\boldsymbol{\tilde G}^{\theta'}_t)^*)  \boldsymbol{1}_t.
$$
Note that 
for $\theta=(\lambda, G)$,
using \eqref{h-tilde} and \eqref{tilde-G},   the kernel of $\boldsymbol{\tilde G}^{\theta}_t$ is given by, 
\be \label{t-g-dec} 
\tilde G^{\theta}_t(s,u)= (2\varrho +G(s-u))\mathds{1}_{\{u<s\}} \mathds{1}_{\{u > t\}}.   
\ee
Hence 
\be \label{d-r1} 
\begin{aligned} 
(\boldsymbol {D}_t^\theta -  \boldsymbol {D}_t^{\theta'})f(s)&= 2(\lambda -\lambda')f(s)+ \int_0^T  (G (s-u)-G'(s-u)) \mathds{1}_{\{s>u>t\}}f(u)du \\ 
&\quad  + \int_0^T  \big(G(u-s)-G'(u-s)\big)\mathds{1}_{\{u>s>t\}}f(u)du. 
\end{aligned} 
\ee
From \eqref{d-r1} and \eqref{op-norm} and an application of Cauchy-Schwarz inequality we get \eqref{g1}. 
 \end{proof} 
  
\begin{lemma} \label{lemmma-inv-op} 
For any $\theta \in \Xi_{\eps}$ the operator $\boldsymbol{D}^{\theta}_t$ is {positive} definite, self-adjoint, invertible and moreover we have 
\be \label{D-pd} 
 \sup_{\theta \in \Xi_{\eps}} \sup_{t \in[0,T]} \langle f, \boldsymbol {D}_t^\theta f \rangle \ge  (L^{-1}-2\eps) \|f\|^2_{L^2([0,T])} , \quad \textrm{for all } f\in L^2([0,T], \mathbb{R}). 
\ee
\end{lemma} 

  \begin{proof} 
  We first note that from \eqref{eq:schur} it follows that $\boldsymbol{D}^{\theta}_t$ is a self-adjoint operator. We will prove \eqref{D-pd}, which under the condition of Definition \ref{assum1:parameters} implies that $\boldsymbol D^{\theta}_t$ is positive definite, hence it is invertible. 
 
 Recall that $\tilde G^{\theta}$ was defined in \eqref{h-tilde} and that $\boldsymbol{\tilde G}^{\theta}_t$ is the operator induced by the kernel $\tilde G^{\theta}(s,u)\mathds 1_{\{u \geq t\}}$. Since $\lam>0$, the operator $\lam \id$ is positive definite and the operator $\boldsymbol{1}^*_t \boldsymbol{1}_t$ is nonnegative definite. It follows from \eqref{eq:schur} that in order to prove that $\boldsymbol D^{\theta}_t$ satisfies \eqref{D-pd} we need to derive a lower bound on $ (\boldsymbol {\tilde G}^{\theta}_t +  (\boldsymbol {\tilde G}_t^{\theta})^*)$. 
Let $f\in L^2\left([0,T],\mathbb R\right)$. Repeating the same steps as in the proof of Lemma 4.1 in \cite{AJ-N-2022} we get 
\be \label{pd-1}
\begin{aligned}  
&\int_0^T\int_0^T\big( G^{\theta}_t(s,u)+(G^{\theta}_t)^*(s,u)\big) f(s)f(u)ds du  \\
 &= \int_0^T\int_0^T 2G (|s-u|) f_t(s)f_t(u)ds du \\
&\geq - 2\eps \|f\|^2_{L^2([0,T])}, \quad \textrm{for all } t\in [0,T], \  \theta \in\Xi_{\eps}, 
\end{aligned} 
\ee
where we used Definition \ref{assum1:parameters} in the last inequality. 

Moreover, we have
\be \label{pd-2} 
\begin{aligned}
&  \int_{t}^T\int_{u}^T f(s)f(u)dsdu  + \int_{0}^T\int_{0}^T \mathds 1_{\{t\leq s \leq u\}}f(u)f(s)du ds \\
&=   \int_{t}^T\int_{u}^T f(s)f(t)dsdt  + \int_{t}^T\int_{t}^u f(u)f(s)ds du \\
&=   \int_{t}^T\int_{t}^T f(s)f(u)dsdu  = \left( \int_{t}^T f(s)ds\right)^2  \geq 0.
\end{aligned}
\ee
From \eqref{t-g-dec}, \eqref{pd-1} and \eqref{pd-2} it follows that $\boldsymbol{\tilde G}^{\theta}_t$ satisfies, 
\be \label{c-eq} 
 \sup_{\theta \in \Xi_{\eps}} \sup_{t \in[0,T]} \langle f,(\boldsymbol {\tilde G}^{\theta}_t +  (\boldsymbol {\tilde G}_t^{\theta})^*)f \rangle \ge  -2\eps \|f\|^2_{L^2([0,T])} , \quad \textrm{for all } f\in L^2([0,T], \mathbb{R}). 
\ee
Since $\boldsymbol{1}^*_t \boldsymbol{1}_t$ is nonnegative definite and by Definition \ref{assum1:parameters}, $\lambda \geq L^{-1}> 2\eps$, \eqref{D-pd} follows from \eqref{eq:schur} and \eqref{c-eq}.  
\end{proof} 

 Now we are ready to prove Lemma \ref{lemma-bnd-g-inv} and Proposition \ref{thm:Psi_lipschitz1}. 
 \begin{proof} [Proof of Lemma \ref{lemma-bnd-g-inv} ]
 Note that from \eqref{op-gam-inv} it follows that it enough to show that 
\be \label{d-inv-bnd}
  \sup_{\theta \in \Xi_{\eps} } \sup_{t \in[0,T]}  \|( \boldsymbol {D}_t^\theta)^{-1} \|_{\rm {op}} < \infty. 
 \ee
Considering \eqref{D-pd}, we choose $\dl \in (0,L^{-1}-2\eps)$ and define the operator 
\be \label{s-op} 
\boldsymbol S^{\theta}_t := (2\lambda -\dl) \id  + (\boldsymbol{\tilde G}^{\theta}_t + (\boldsymbol{\tilde G}^{\theta}_t)^*) + 2\phi  \boldsymbol{1}^*_t \boldsymbol{1}_t, \quad 0\leq t \leq T. 
\ee
From Lemma \ref{lemmma-inv-op}  it follows that there exists $\bar \dl>0$ such that 
\be \label{S-pd} 
 \sup_{\theta \in \Xi_{\eps}} \sup_{t \in[0,T]} \langle f, \boldsymbol {S}_t^\theta f \rangle > \bar \dl \|f\|^2_{L^2([0,T])} , \quad \textrm{for all } f\in L^2([0,T], \mathbb{R}). 
\ee
and in particular $\boldsymbol {S}_t^\theta $ is positive definite, invertible, self-adjoint and compact with respect to the space of bounded operators {on $L^{2}([0,T])$}, equipped with the operator norm given in \eqref{op-norm}. From Theorem 4.15 in \cite{Porter1990} it follows that  $\boldsymbol S^{\theta}_t$ admits a spectral decomposition in terms of a sequence of positive eigenvalues $(\mu_{t,n})_{n=1}^{\infty}$ and an orthonormal sequence of eigenvectors $(\varphi^{\theta}_{t,n})_{n=1}^{\infty}$ in $L^{2}([0,T])$, such that 
$$
\boldsymbol S^{\theta}_t= \sum_{k}\mu^{\theta}_{t,k} \langle \varphi^{\theta}_{t,k}, \cdot \rangle_{L^2} \varphi^{\theta}_{t,k}. 
$$
By application of Cauchy Schwarz and the fact that $\boldsymbol S^{\theta}_t$ is self-adjoint we get for any $\theta \in \Xi_{\eps}$, 
\begin{align*}
 \sup_{t\leq T}\sum_{k}(\mu_{t,k}^{\theta})^2 & \leq  C\left((2\lambda -\dl)^2 + \sup_{t\leq T} \int_0^T \left(({\tilde G^{\theta}}_t + ({\tilde G}^{\theta}_t)^*) + 2\phi  \mathds{1}^*_t \mathds{1}_t \right)^2 {(s, s)  ds} \right)\\
&<\infty, 
\end{align*}
where the second inequality follows from Definition \ref{assum1:parameters} and \eqref{h-tilde}. 
From \eqref{eq:schur} and \eqref{s-op} it follows that we can rewrite  $\boldsymbol D^{\theta}_t= \boldsymbol S^{\theta}_t +\dl \id$ as follows, 
$$ \boldsymbol D^{\theta}_t = \sum_{k} \left(\dl + \mu^{\theta}_{t,k}\right) \langle  \varphi^{\theta}_{t,k}, \cdot \rangle_{L^2}  \varphi^{\theta}_{t,k}. $$
We can therefore represent $(\boldsymbol{D}^{\theta}_t)^{-1}$ as follows, 
$$(\boldsymbol{D}^{\theta}_t)^{-1}= \sum_{k}  \frac{1} {\left(\dl + \mu^{\theta}_{t,k}\right)}\langle \varphi^{\theta}_{t,k}, \cdot \rangle_{L^2} \varphi^{\theta}_{t,k}.$$
Since $\dl>0$ and $\mu^{\theta}_{t,k} \geq 0$, for all $t\in [0,T]$, $\theta \in \Xi_{\eps}$ and $k=1,2,...$, we get that for any $f\in L^2([0,T], \mathbb{R})$, 
$$
\| (\boldsymbol{D}_t^\theta)^{-1}f \|_{L^2} \leq \frac{1}{\dl} \|f\|_{L^2}, \quad \textrm{for all } 0\leq t \leq T, \ \theta \in \Xi_{\eps}.
$$
Together with \eqref{op-norm} we get \eqref{d-inv-bnd} and this completes the proof. 
  \end{proof} 
 
 \begin{proof}[Proof of Proposition \ref{thm:Psi_lipschitz1}] 
We observe that 
\be \label{gam-id} 
 (\boldsymbol {\Gamma}_t^\theta  )^{-1}
- ( \boldsymbol {\Gamma}_t^{\theta'}  )^{-1}
= (\boldsymbol {\Gamma}_t^\theta  )^{-1}\left( \boldsymbol {\Gamma}_t^\theta 
-  \boldsymbol {\Gamma}_t^{\theta'} \right)( \boldsymbol {\Gamma}_t^{\theta'}  )^{-1}.
\ee
By taking the operator norm on both sides of \eqref{gam-id} 
and using Lemma \ref{lemma-bnd-g-inv} and then Lemma \ref{lemma-lip-gam} we get   for all $\theta,\theta'\in \Xi_{\eps} $ and $t\in [0,T]$,
\begin{align*}
\left\|(\boldsymbol {\Gamma}_t^\theta  )^{-1}
- ( \boldsymbol {\Gamma}_t^{\theta'}  )^{-1}\right\|_{\rm op} 
&\le C \| \boldsymbol {\Gamma}_t^\theta -  \boldsymbol {\Gamma}_t^{\theta'}\|_{\rm op}
\leq  C\left(|\lambda-\lambda'| +\|G-G'\|_{L^2([0,T])}\right). 
\end{align*}

 \end{proof}

\appendix

\section{Regression-based algorithm for signal estimation} \label{app-signal} 
Observe that the trading strategy
\eqref{eq:optimalcontrol_monotone_str}
(i.e., the process $a$ in  \eqref{eq:aB})
involves the  conditional  process 
$(t,s)\mapsto  \sE\left[   A_s\mid \cF_t\right]$ 
of the signal process $A$
on 
 $\Delta_T=\{(t,s)\mid 0\le t\le s\le T\}$.
As  the conditional distribution of $A$ is in general not known analytically,
the section  proposes a regression-based  algorithm 
  to estimate the process $(t,s)\mapsto  \sE\left[   A_s\mid \cF_t\right]$ 
based on observed signal trajectories. 
Since the agent's trading strategy does not affect the   signals, 
we   assume for simplicity that 
the signal estimation has been carried out separately  from 
learning the price impact coefficients $(\lambda^\star,G^\star)$.

Throughout this section,
we assume that there exist 
independent copies $(A^m)_{m\in \sN}$ of $A$,
and impose the  regularity condition of the signal $A$
as specified in Assumption \ref{assum:signal}.
To simplify the notation, 
 we write $\|X\|_{L^p(\Omega)}$, $p>1$,
for the $L^p$-norm of a random variable $X:\Omega\to \sR$,
 and write $\sE_t[\cdot]=\sE[\cdot\mid \cF_t]$  
for each $t\in [0,T]$. 

\begin{assumption}
\label{assum:signal}
There exists   $I: \Omega\times [0,T]\to \sR$ such that
 $A_t=\int_0^t I_s d s$  for all $t\in [0,T]$,
 and $I$ is Markov with respect to the filtration 
 $(\cF_t)_{t\in [0,T]}$.
Moreover, 
there exists   $\vartheta>2$ and    $L\ge 0$ 
such that
$\sup_{t\in [0,T]}\|I_t\|_{L^\vartheta(\Omega)}<\infty$,
and 
  for all $0\le t\le s\le r\le T$ and $x,y\in \sR$,
  $\| I_t-I_s\|_{L^2(\Omega)}\le L  |t-s|^{1/2}$
  and 
 $|\sE^{t,x}[I_r]-\sE^{s,y}[I_r]|\le L\big(|x-y|+|t-s|^{1/2}(1+|x|+|y|)\big)$,
where $\sE^{t,x}[I_r]\coloneqq \sE[I_r\mid I_t=x]$. 
\end{assumption}

\begin{remark}
Assumption \ref{assum:signal}
allows for non-Markovian signals $A$ (with respect to $(\cF_t)_{t\in [0,T]}$),
but requires the time derivative of  $A$ to be 
  Markov  
   with a sufficiently regular transition kernel.
This assumption includes 
as   special cases 
the signal processes  in \cite{Lehalle-Neum18,NeumanVoss:20},
  where $I$ is an   Ornstein--Uhlenbeck process.
  More generally, 
  Assumption \ref{assum:signal} holds 
  if $I$ solves a jump-diffusion  stochastic differential equation 
  with   sufficiently regular (e.g., Lipschitz continuous \cite[Theorem 4.1.1]{delong2013backward}) coefficients.
   
\end{remark}

 \paragraph{Least-squares Monte Carlo   for signal estimation.}
By
  Assumption \ref{assum:signal},
  for each $m\in \sN$ and $t\in [0,T]$, $A^m_t=\int_0^t I^m_s d s$,
  where  $(I^m)_{m\in \sN}$ are independent copies of the 
 Markov process $I$.
In the sequel, we  approximate  
$  (t,s)\mapsto \sE_t [   A_s ]
=A_t+\int_t^s \sE_t[I_r] dr$
by  constructing a regression-based estimator of  $
\Delta_T\ni (t,s)\mapsto \sE[I_s\mid I_t]\in \sR $
based on $(I^m)_{m\in \sN}$.
We first discretise $(t,s)\mapsto \sE_t [I_s]$ in time. 
More precisely, for each $N\in \sR$, consider the   grid $\pi_N=\{t_i\}_{i=0}^N
\in \mathscr{P}_{[0,T]}$ with $t_i=iT/N$ for all $0\le i\le N$, and 
define the following   approximation of $\sE_t [   A_s ]$:
for all $(t,s)\in \Delta_T$,
\begin{equation}\label{eq:discrete_I}
\sE_t [   A_s ]\approx A_t +\int_t^s \cI^{N}_{t,r}dr,
\quad 
\textnormal{with
$\cI^{N}_{t,r}\coloneqq \sum_{i,j=0}^{N-1}\sE_{t_i}[I_{t_j}]\mathds{1}_{ [t_i,t_{i+1})\times [t_j,t_{j+1}) }(t,r)
$,  $(t,r)\in  \Delta_T$.
}
\end{equation}

The conditional expectations in \eqref{eq:discrete_I} 
are then projected 
on prescribed basis functions via least-squares Monte Carlo methods (see e.g., \cite{gyorfi2002distribution}).
To this end,
let $\cV$ be a finite-dimensional  vector space of functions $\psi:\sR\to\sR$,
and
 for each $R\ge 0$,
let $\cT_R:\sR\to \sR$ be the truncation function such that $\cT_R(x)=\max(-R,\min(x,R))$ for all $x\in \sR$.
Then for each $N,M\in \sN$ and $R\ge 0$,
 we define the following truncated least-squares estimate of $\sE_{t_i}[I_{t_j}]$: for all $0\le i\le j\le N$, 
 \begin{equation}
 \label{eq:truncated_LS_V}
\psi_{i,j}(\cdot)=\cT_R\tilde{\psi}_{i,j}(\cdot),
\quad 
\textnormal{with $\tilde{\psi}_{i,j}\in \argmin_{\psi\in \cV}\frac{1}{M}\sum_{m=1}^{M}| \psi(I^m_{t_i})-\cT_R(I^m_{t_j})|^2$}.
 \end{equation}
As  $\cV$ is a vector space,  
$\tilde{\psi}_{i,j}$ in \eqref{eq:truncated_LS_V}  can be computed by solving a  least-squares   problem over the weights of some fixed basis functions, whose solution may not be unique (see e.g., page 162 of \cite{gyorfi2002distribution}). Note that there can be more than one solution to \eqref{eq:truncated_LS_V}. 
We then consider the approximation 
 $\sE_t [   A_s ] \approx A_t +\int_t^s \cI^{N,M,R}_{t,r}dr$,
 with
 \begin{align}
{ \cI}^{N,M,R}_{t,r}&= \sum_{i,j=0}^{N-1}\psi_{i,j}(I_{t_i})\mathds{1}_{ [t_i,t_{i+1})\times [t_j,t_{j+1}) }(t,r),
\quad (t,r)\in \Delta_T.
\label{eq:discrete_sample_I}
 \end{align}

The accuracy of \eqref{eq:discrete_sample_I}  depends on the expressivity and complexity of the  vector space $\cV$
(see Proposition \ref{prop:approximation_error_general}).
By the Lipschitz continuity of  the   map $x\mapsto \sE^{t_i,x}[I_{t_j}]$ (Assumption \ref{assum:signal}),
we set  the vector space $\cV$ 
as the space of piecewise constant functions
defined on a spatial grid. 
This  allows for optimally balancing the expressivity and complexity
of $\cV$
and   obtaining a precise error estimate of \eqref{eq:discrete_sample_I}
(see Theorem \ref{thm:approximation_error_LSMC_constant}).
 More precisely, 
for each $R\ge 0$ and $K\in \sN$, 
 let $\mathfrak{p}_K\coloneqq \{-\frac{R}{2}=x_0<x_i<\ldots<x_K=\frac{R}{2}\}$ be a  uniform grid of   $[-\frac{R}{2},\frac{R}{2}]$ such that $x_{i+1}-x_i=\frac{R}{K}$ for all $0\le i\le K-1$,
 and let $\cV_K$ be the   space of real-valued functions 
 that are piecewise constant 
 on the grid $\mathfrak{p}_K$ and 
  zero outside $[-\frac{R}{2},\frac{R}{2}]$.
 It is clear that $\cV_K$ is of the dimension $K$. 
 Then for each $N,M, K\in \sN$ and $R\ge 0$,
consider the approximation  $\sE_t [   A_s ] \approx A_t +\int_t^s \cI^{N,M,K, R}_{t,r}dr$,
where 
the process 
$\Delta_T\ni (t,r)\mapsto \cI^{N,M,K, R}_{t,r}\in\sR$ is defined by (cf.~\eqref{eq:discrete_sample_I}:
\begin{align}
{ \cI}^{N,M,K, R}_{t,r}&= \sum_{i,j=0}^{N-1}\psi_{i,j}(I_{t_i})\mathds{1}_{ [t_i,t_{i+1})\times [t_j,t_{j+1}) }(t,r),
\quad (t,r)\in \Delta_T,
\label{eq:discrete_sample_space_I}
 \end{align}
with 
$\psi_{i,j}$ being  the   truncated least-squares estimate  \eqref{eq:truncated_LS_V} 
over the space 
 $\cV=\cV_{K}$:
$$
\psi_{i,j}(\cdot)=\cT_R\tilde{\psi}_{i,j}(\cdot),
\quad 
\textnormal{with $\tilde{\psi}_{i,j}\in \argmin_{\psi\in \cV_K}\frac{1}{M}\sum_{m=1}^M | \psi(I^m_{t_i})-\cT_R(I^m_{t_j})|^2$}.
$$

\paragraph{Convergence rates of \eqref{eq:discrete_sample_I} and \eqref{eq:discrete_sample_space_I}.}

To quantify the   accuracy of \eqref{eq:discrete_sample_I} and \eqref{eq:discrete_sample_space_I},
we start with the following technical lemma.

\begin{lemma}
\label{Lemma:c_I}
Suppose  that Assumption \ref{assum:signal}  holds.
Then 
there exists   $C\ge 0$ such that 
for all $0\le t\le s\le r\le T$, 
$\|\sE_t[I_r]-\sE_t[I_s]\|_{L^2(\Omega)}
\le C|r-s|^{1/2}$ and 
$\|\sE_t[I_r]-\sE_s[I_r] \|_{L^2(\Omega)} \le C|t-s|^{1/2}$.
 
\end{lemma}

\begin{proof}
Let $0\le t\le s\le r\le T$, 
and  $C\ge 0$ be a generic constant independent of $t,s$ and $r$.
Jensen's inequality and Assumption \ref{assum:signal}  imply that
$\sE[|\sE_t[I_r]-\sE_t[I_s]|^2]
\le \sE[|I_r-I_s|^2]\le L^2|r-s|$.
Moreover, by the Markov property of $I$ and Assumption \ref{assum:signal},
\begin{align*}
|\sE_t[I_r]-\sE_s[I_r]|=|\sE^{t,I_t}[I_{r}]-\sE^{s,I_s}[I_{r}]| 
\le L\big(|I_t-I_s|+|t-s|^{1/2}(1+|I_t|+|I_s|\big).
\end{align*}
Taking $L^2$-norm on both sides and apply Young's inequality yield
\begin{align*}
\sE[|\sE_t[I_r]-\sE_s[I_r]|^2]
& \le 2L^2\big(\sE[|I_t-I_s|^2]+|t-s|\sE[(1+|I_t|+|I_s|)^2]\big)
\le C|t-s|,
\end{align*}
where the last inequality used 
 $\sup_{t\in [0,T]}\sE[|I_t|^2]<\infty$.
\end{proof}

The following proposition 
quantifies  the time discretisation error     of \eqref{eq:discrete_I}
based on Lemma \ref{Lemma:c_I}.
\begin{proposition}
\label{prop:signal_time}
Suppose that Assumption \ref{assum:signal}  holds.
Then there exists $C\ge 0$ such that 
$$
\sup_{t\in [0,T]}\left\|  \sup_{s\in [t,T]}\bigg|\sE_t [  A_s]-
\left(A_t+ \int_t^s  \cI^{N}_{t,r}   d r\right)\bigg|\right\|_{L^2(\Omega)}
\le C   N^{-\frac{1}{2}},
\quad \textrm{for all } N\in \sN. 
$$
\end{proposition}
\begin{proof}
Throughout this proof, let   $t\in [0,T)$ and $N\in \sN$ be fixed and  $C$ be a generic constant independent of   $t$ and $N$.
Let 
$t_i\in \pi_{N}$ such that 
$t\in [t_i,t_{i+1})$. 
By \eqref{eq:discrete_I}, 
for all $s\in [t,T]$, 
\begin{align*}
&\left|\int_t^s  \sE_t[I_r]   d r-\int_t^s   \cI^{N}_{t,r} d r\right|
\le 
 \left|\int_t^s  ( \sE_t[I_r]-\sE_{t_i}[I_r]+\sE_{t_i}[I_r] -   \cI^{N}_{t,r}) d r\right|
\\
&\le 
  \int_t^s |\sE_t[I_r]-\sE_{t_i}[I_r]|  d r
+
\int_t^s  \bigg| \sE_{t_i}[I_r]- \sum_{j=i}^{N-1}\sE_{t_i}[I_{t_j}]\mathds{1}_{ [t_j,t_{j+1}) }(r)\bigg| dr 
\\
&\le 
  \int_t^T |\sE_t[I_r]-\sE_{t_i}[I_r]|  d r
+
\sum_{j=i}^{N-1}
\int_{t_j}^{t_{j+1}}   | \sE_{t_i}[I_r]- \sE_{t_i}[I_{t_j}]  | dr.
\end{align*}
By taking
the supremum over $s\in [t,T]$ and the $L^2$-norm on both sides of the above estimate,  and  
applying   Lemma \ref{Lemma:c_I}, 
\begin{align*}
&\left\|  \sup_{s\in [t,T]}\bigg|\sE_t [  A_s]-
\left(A_t+ \int_t^s  \cI^{N}_{t,r}   d r\right)\bigg|\right\|_{L^2(\Omega)}
\\
&\le
 \int_t^T \|\sE_t[I_r]-\sE_{t_i}[I_r]\|_{L^2(\Omega)} d r
+
\sum_{j=i}^{N-1} \int_{t_j}^{t_{j+1}}   \| \sE_{t_i}[I_r]- \sE_{t_i}[I_{t_j}] \|_{L^2(\Omega)} d r
\\
&\le  
 C\int_t^T   |t-t_i|^{1/2} d t+
\sum_{j=i}^{N-1} \int_{t_j}^{t_{j+1}}  |r-t_j|^{1/2} d r
\le C   {N}^{-1/2},
\end{align*}
due to the fact that $t\in [t_i,t_{i+1})$. 
Taking
the supremum over $t\in [0,T]$ yields
    the desired estimate. 
\end{proof}

 The following proposition quantifies the accuracy  
 of $\left(A_t +\int_t^s \cI^{N,M,R}_{t,r}dr \right)_{(t,s)\in \Delta_T}$
  in terms of the number of time discretisation  $N$, 
 the sample size $M$, the truncation level $R$ and the complexity  of the function space $\cV$.  

 \begin{proposition}
 \label{prop:approximation_error_general}
 Suppose that Assumption \ref{assum:signal}   holds.
Then there exists $C\ge 0$  
 such that 
for all   $N,M\in \sN$, $R\ge 0$
and vector spaces $\cV$ of functions $\psi:\sR\to \sR$,
\begin{align*}
&\sup_{t\in [0,T]}\left\|  \sup_{s\in [t,T]}\bigg|\sE_t [  A_s]-
\left(A_t+ \int_t^s   \cI^{N,M,R}_{t,r}   d r\right)\bigg|\right\|_{L^2(\Omega)}
\\
&\le
 C  \bigg(
 \frac{1}{\sqrt{N}}
 +
R\sqrt{\frac{(\ln M+1) n_{\cV}}{M}}
+
\sup_{t\in [0,T]} \sE[|I_{t}|^2\mathds{1}_{|I_{t}|\ge R}]
+
\max_{0\le i\le j\le N}
\inf_{\psi\in \cV}\|\sE_{t_i} [I_{t_j} ]-\psi(I_{t_i})\|_{L^2(\Omega)}
\bigg),
\end{align*}
where 
$n_\cV$ is the vector space dimension of $\cV$.
 \end{proposition}

\begin{proof}
Throughout this proof, let    $N,M,R\in \sN$, $t\in [0,T)$ and $\cV$ be fixed, and 
let $C$ be a generic constant independent of the above quantities.
By Jensen's inequality,
\begin{align}
\label{eq:truncation_error}
\begin{split}
 \|\sE_{t_i}[\cT_R(I_{t_j})]-\sE_{t_i}[I_{t_j}]\|^2_{L^2(\Omega)}
&\le 
\|\sE_{t_i} [(|I_{t_j}|+R)\mathds{1}_{|I_{t_j}|> R}]\|^2_{L^2(\Omega)}
\\
&\le 4\sE[|I_{t_j}|^2\mathds{1}_{|I_{t_j}|> R}].
\end{split}
\end{align}
Then 
by observing that $|\cT_R(I_{t_j})|\le R$ and applying \cite[Theorem 11.3]{gyorfi2002distribution},
there exists $C\ge 0$ such that 
for all  $0\le i\le j\le N$, 
\begin{align}
\label{eq:regression_error}
\begin{split}
&\|\sE_{t_i} [I_{t_j}]-\psi_{i,j}(I_{t_i})\|^2_{L^2(\Omega)} 
\\
&  \le 
2\|\sE_{t_i} [I_{t_j}]-\sE_{t_i}[\cT_R(I_{t_j})]\|^2_{L^2(\Omega)} 
 +2 \|\sE_{t_i}[\cT_R(I_{t_j})]-\psi_{i,j}(I_{t_i})\|^2_{L^2(\Omega)}
\\
&  \le
8\sE[|I_{t_j}|^2\mathds{1}_{|I_{t_j}|> R}]+
C\bigg(
R^2\frac{(\ln M+1) n_{\cV}}{M}
+\inf_{\psi\in \cV}\|\sE_{t_i} [\cT_R(I_{t_j})]-\psi(I_{t_i})\|^2_{L^2(\Omega)}
\bigg)
\\
&  \le 
C\bigg(
R^2\frac{(\ln M+1) n_{\cV}}{M}
+
\sup_{t\in [0,T]} \sE[|I_{t}|^2\mathds{1}_{|I_{t}|\ge R}]
+
\inf_{\psi\in \cV}\|\sE_{t_i}[I_{t_j}]-\psi(I_{t_i})\|^2_{L^2(\Omega)}
\bigg),
\end{split}
\end{align}
where $n_\cV$ is the vector space dimension of $\cV$.
Hence,
let $i\in\sN$ be such that $t\in [t_i,t_{i+1})$, 
 by 
\eqref{eq:discrete_I},
 \eqref{eq:discrete_sample_I} and  \eqref{eq:regression_error},
\begin{align*}
&\bigg\|
\sup_{s\in [t,T]}
\left| \int_t^s   \cI^N_{t,r} d r-\int_t^s   \cI^{N,M,R}_{t,r} d r
\right|
\bigg\|_{L^2(\Omega)}
\le 
 \int_t^T\| \cI^N_{t,r}-  \cI^{N,M,R}_{t,r} \|_{L^2(\Omega)} dr 
\\
&\le 
 \int_{t_i}^T \left \|\sum_{j=i}^{N-1}
 \sE_{t_i} [I_{t_j} ]\mathds{1}_{   [t_j,t_{j+1}) }(r)-\sum_{j=i}^{N-1} \psi_{i,j}(I_{t_i})\mathds{1}_{   [t_j,t_{j+1}) }(r)
 \right\|_{L^2(\Omega)} dr
 \\
 &\le 
 \sum_{j=i}^{N-1}
 \int_{t_j}^{t_{j+1}} \|\sE_{t_i} [I_{t_j}]-\psi_{i,j}(I_{t_i})\|_{L^2(\Omega)}  dr 
 \\
 &\le 
 C \bigg(
R\bigg(\frac{(\ln M+1) n_{\cV}}{M}\bigg)^{\frac{1}{2}}
+
\sup_{t\in [0,T]} \sE[|I_{t}|^2\mathds{1}_{|I_{t}|\ge R}]
+
\max_{0\le i\le j\le N}
\inf_{\psi\in \cV}\|\sE_{t_i} [I_{t_j} ]-\psi(I_{t_i})\|_{L^2(\Omega)}
\bigg),
\end{align*}
which along with 
Proposition \ref{prop:signal_time} leads to the desired estimate.
\end{proof}

The following theorem simplifies 
  the error bound in Proposition \ref{prop:approximation_error_general}
  with     $\cV$ 
being the space of piecewise constant functions.
It 
specifies the precise dependence of the hyperparameters 
$N,K, R$ on the sample size 
$M$,
and 
 establishes the convergence rate of 
  $\left(A_t +\int_t^s \cI^{N,M,K, R}_{t,r}dr \right)_{(t,s)\in \Delta_T}$ 
 as  $M$ tends to infinity. 
\begin{theorem} \label{theorem-signal} 
\label{thm:approximation_error_LSMC_constant}
 Suppose that Assumption \ref{assum:signal}    holds.
Then there exists $C\ge 0$
 such that 
for all   $N,M, K\in \sN$ and $R\ge 0$,
\begin{align}\label{eq:error_piecewise_const_basis}
\begin{split}
&\sup_{t\in [0,T]}\left\|  \sup_{s\in [t,T]}\bigg|\sE_t [  A_s]-
\left(A_t+ \int_t^s   \cI^{N,M,K,R}_{t,r}   d r\right)\bigg|\right\|_{L^2(\Omega)}
\\
&\le
 C  \bigg(
 \frac{1}{\sqrt{N}}
 +
R\sqrt{\frac{(\ln M+1) K}{M}}
+\frac{R}{K}
+
\frac{1}{R^{\frac{\vartheta -2}{2}}}
\bigg).
\end{split}
\end{align}
Consequently, 
if one sets 
$(N_M)_{M\in \sN}, (K_M)_{M\in \sN}\subset \sN$, and $(R_M)_{M\in \sN}
\subset (0,\infty) $ such that  
$$
N_M\sim
\left(
\frac{M}{\ln M+1 }
\right)^{\frac{2} {3}},
\;
 K_M\sim
\left(
\frac{M}{\ln M+1 }
\right)^{\frac{1} {3}},
\; 
R_M\sim 
 \left(
\frac{M}{\ln M+1 }
\right)^{\frac{2}{3\vartheta}},
\footnotemark
$$
\footnotetext{
For any sequences $(a_M)_{M\in \sN}, (b_M)_{M\in \sN}\subset (0,\infty)$, 
we write 
$a_M\sim b_M$ if 
$0<\lim\inf_{M\to \infty} \frac{a_M}{b_M}
\le \lim\sup_{M\to \infty} \frac{a_M}{b_M}<\infty
$.
}
then  for all $M\in \sN$,
\begin{align}
\label{eq:error_piecewise_const_final}
\begin{split}
&\sup_{t\in [0,T]}\left\|  \sup_{s\in [t,T]}\bigg|\sE_t [  A_s]-
\left(A_t+ \int_t^s   \cI^{N_M,M,K_M,R_M}_{t,r}   d r\right)\bigg|\right\|_{L^2(\Omega)}
 \le
 C  
\left( \frac{\ln  M+1  }{ M}
\right)^{\frac{\vartheta-2}{3\vartheta}}.
\end{split}
\end{align}
 
\end{theorem}

\begin{proof}
Throughout this proof, let  $N,M, K\in \sN$, $R\ge 0$ and $t\in [0,T)$  be fixed, and 
 let $C$ be a generic constant independent of the above quantities.
For each $0\le i\le j\le N$, let  $\varphi_{i,j}: \sR\to \sR$ be
such that $\varphi_{i,j}(x)=\sE^{t_i,x}[I_{t_{j}}]$ for all $x\in \sR$.
By Assumption \ref{assum:signal},
$|\varphi_{i,j}(x)-\varphi_{i,j}(y)|\le L|x-y|$ for all $x,y\in \sR$.
For any given $0\le i\le j\le N$,
by considering  
$\tilde{\beta}_{\ell}=\varphi_{i,j}(x_\ell)$ for all $\ell=0,\ldots, K-1$,
\begin{align*}
&\inf_{\psi\in \cV_K}\|\sE_{t_i}[I_{t_j}]-\psi(I_{t_i})\|_{L^2(\Omega)}
\le \left\|\sE_{t_i}[I_{t_j}]-\sum_{\ell=0}^{K-1}\tilde{\beta}_\ell\mathds{1}_{[x_\ell,x_{\ell+1})}(I_{t_i})\right\|_{L^2(\Omega)}
\\
&\quad \le
\left\|\varphi_{i,j}(I_{t_i})-\varphi_{i,j}(I_{t_i})\mathds{1}_{|x|\le \frac{R}{2}}(I_{t_i})\right\|_{L^2(\Omega)}
\\
&\quad\quad 
+ \left\|\varphi_{i,j}(I_{t_i})\mathds{1}_{|x|\le \frac{R}{2}}(I_{t_i})-\sum_{\ell=0}^{K-1}\varphi_{i,j}(x_\ell)\mathds{1}_{[x_\ell,x_{\ell+1})}(I_{t_i})\right\|_{L^2(\Omega)}
\\
&\quad
\le \|\varphi_{i,j}(I_{t_i}) \mathds{1}_{|I_{t_i}|> \frac{R}{2} }\|_{L^2(\Omega)}+
\left\|  \sum_{\ell=0}^{K-1} |\varphi_{i,j}(I_{t_i})-\varphi_{i,j}(x_\ell)|\mathds{1}_{[x_\ell,x_{\ell+1})}(I_{t_i})\right\|_{L^2(\Omega)}
\\
&\quad 
 \le 
 \|I_{t_j} \boldsymbol{1}_{|I_{t_i}|> \frac{R}{2}}\|_{L^2(\Omega)}+\frac{L R}{K},
\end{align*}
 where the first term in the last inequality used 
$ \varphi_{i,j}(I_{t_i}) \mathds{1}_{|I_{t_i}|>  {R}/{2}} 
=\sE_{t_i}[I_{t_j}\mathds{1}_{|I_{t_i}|> {R}/{2}} ]$ and 
 Jensen's inequality,
 and the second term used  
 the  $L$-Lipschitz continuity of $\varphi_{i,j}$. 
By 
H\"{o}lder's inequality (with $p=\frac{\vartheta}{2}>1$) and 
  Markov's inequality,
\begin{align}\label{eq:truncation_error_general}
\begin{split}
\left\|I_{t_j}  \boldsymbol{1}_{|I_{t_i}|>   \frac{R}{2}}\right\|_{L^2(\Omega)}
&\le 
\|I_{t_j}\|_{L^\vartheta(\Omega)}\sP(|I_{t_i}|>   \tfrac{R}{2})^{
\frac{\vartheta -2}{2\vartheta}}
\le \|I_{t_j}\|_{L^{\vartheta}(\Omega)}\left( 
\frac{2^\vartheta\|I_{t_i}\|^\vartheta_{L^\vartheta(\Omega)}}{R^\vartheta}
\right)^{
\frac{\a -2}{2\a}}
\\
&\le  2^{\frac{\vartheta -2}{2}}  {R^{-\frac{\vartheta -2}{2}}} {\sup_{t\in [0,T]}\|I_t\|^{\frac{\vartheta}{2}}_{L^\vartheta(\Omega)}}.
\end{split}
\end{align}
This along with Proposition 
\ref{prop:approximation_error_general}
and $n_{\cV_K}=K$ leads to the desired estimate \eqref{eq:error_piecewise_const_basis}.

Finally, let $ \mathfrak{C}_\vartheta= 2^{\frac{\vartheta -2}{2}}  \sup_{t\in [0,T]}\|I_t\|^{\frac{\vartheta}{2}}_{L^\vartheta(\Omega)}<\infty$.
Observe that 
if 
\begin{equation}
\label{eq:lse_MC_criterion}
K_M\sim \sqrt{N}, 
\quad 
R_M\sim \frac{\mathfrak{C}_\vartheta^{{2}/{\vartheta}}}{
\Big(\sqrt{\frac{(\ln  M+1) \sqrt{N }}{ M}}
+\frac{1}{\sqrt{N }}\Big)^{{2}/{\vartheta}}},
\quad 
N_M\sim
\left(
\frac{M}{\ln M+1 }
\right)^{{2}/{3}},
\end{equation}
then the following error estimate holds:
for all $M\in \sN$, 
\begin{align*}
&\frac{1}{\sqrt{N_M}}
 +
R_M\sqrt{\frac{(\ln M+1) K_M}{M}}
+\frac{R_M}{K_M}
+
\frac{\mathfrak{C}_\vartheta }{R_M^{\frac{\vartheta -2}{2}}}
\\
& 
\le
C\Bigg(
R_M\Bigg(
 \sqrt{\frac{(\ln  M+1) \sqrt{N_M }}{ M}}
+\frac{1}{\sqrt{N_M }}\Bigg)
+
\frac{\mathfrak{C}_\vartheta  }{R_M^{\frac{\vartheta -2}{2}}}
\Bigg)
\\
&  
\le 
C \mathfrak{C}_\vartheta^{\frac{2}{\vartheta}}
\bigg( 
 \sqrt{\frac{(\ln  M+1) \sqrt{N_M }}{ M}}
+\frac{1}{\sqrt{N_M }}
\bigg)^{1-\frac{2}{\vartheta}}
\le 
C \mathfrak{C}_\vartheta^{\frac{2}{\vartheta}}
\left( \frac{\ln  M+1  }{ M}
\right)^{\frac{\vartheta-2}{3\vartheta}}.
\end{align*}
The desired estimate \eqref{eq:error_piecewise_const_final}
follows from the fact that 
$1\le (2^{\frac{\vartheta -2}{2}} )^\frac{2}{\vartheta}\le 2$ for all $\vartheta>2$,
and the observation that
  the choices of hyperparameters $N,K, R$ in the statement 
satisfies the criterion \eqref{eq:lse_MC_criterion}. 
\end{proof}

\section*{Acknowledgments}
We are very grateful to the Associate Editor and to the anonymous referees for careful reading of the manuscript and for a number of useful comments and suggestions that significantly improved this paper.
\\

\noindent \textbf{Funding.}   Not applicable.
 \medskip  \\
\textbf{Availability of data and material.} Not applicable.
\medskip \\
\textbf{ Compliance with ethical standards.} The authors have no conflicts of interest to declare that are relevant to the content of this article.
\medskip \\
\textbf{Code availability.} Not applicable.

 \bibliographystyle{plainnat}
 \bibliography{MarketImpact.bib}

\end{document}